\documentclass[smallextended]{svjour3}

\def\makeheadbox{\relax}

\AtBeginDocument{%
  \paperwidth=\dimexpr
    1in + \oddsidemargin
    + \textwidth
    + 1in + \oddsidemargin
  \relax
  \paperheight=\dimexpr
    1in + \topmargin
    + \headheight + \headsep
    + \textheight
    + 1in + \topmargin
  \relax
  \usepackage[pass]{geometry}\relax
}

\usepackage{amssymb}
\usepackage{graphicx}

\usepackage{enumitem} 

\usepackage{caption}
\usepackage{hyperref}
\usepackage{subcaption}

\usepackage{verbatim}
\usepackage{enumitem}
\usepackage{graphicx}
\usepackage{algpseudocode}
\usepackage{algorithmicx}
\usepackage{algorithm}

\usepackage{epsfig}
\usepackage{epstopdf}

\usepackage{multicol}

\usepackage{color}

\newcommand{\Tin}{T_{\mathrm{in}}}
\newcommand{\Tmin}{T_{\mathrm{min}}}
\newcommand{\Vin}{V_{\mathrm{in}}}
\newcommand{\Ein}{E_{\mathrm{in}}}
\newcommand{\Sout}{S_{\mathrm{out}}}
\newcommand{\barSout}{{\bar S}_{\mathrm{out}}}
\newcommand{\Pout}{\mathcal{P_{\mathrm{out}}}}
\newcommand{\barPout}{\bar {\mathcal P}_{\mathrm{out}}}
\newcommand{\Pnew}{{P}_\mathrm{new}}
\newcommand{\Pnewi}[1]{{P}_{\mathrm{new}}({#1})}
\newcommand{\Pstar}{{\mathcal P}^*}
\newcommand{\barbarSout}{{\bar{\bar S}}_{\mathrm{out}}}
\newcommand{\barbarPout}{\bar {\bar {\mathcal P}}_{\mathrm{out}}}
\newcommand{\barbarS}{{\bar {\bar S}}}
\newcommand{\barbarT}{{\bar {\bar T}}}
\newcommand{\barbarV}{{\bar {\bar V}}}
\newcommand{\barbarBP}{{\bar {\bar {BP}}}}
\newcommand{\emptyf}{{f(\cdot,\cdot)}}

\begin{document}
%\renewcommand{\topfraction}{.75}
%\title{Minmax Tree Facility Location With Applications to Sink Evacuation}
%\title{Sink Evacuation on Trees with Dynamic Confluent Flows}
%\title{Minmax Tree Facility Location and Sink Evacuation with Dynamic Confluent Flows}
\title{Minmax Centered $k$-Partitioning of Trees and Applications to  Sink Evacuation with Dynamic Confluent Flows}
\author{Di Chen \and Mordecai Golin}
\institute{Di Chen \at HK University of Science \& Technology;  \email{di.chen@connect.ust.hk}  \and Mordecai Golin \at HK University of Science \& Technology;  \email{golin@cse.ust.hk}}
%\author[*]{Di Chen}
%\author[*]{Mordecai Golin}
%\affil[*]{Hong Kong University of Science and Technology}
%\affil[2]{Hong Kong University of Science and Technology}

\titlerunning{Minmax Centered {\em k}-Partitioning of Trees}
\authorrunning{D. Chen and M. J. Golin}
%\Copyright{Di Chen and Mordecai J. Golin}
%\subjclass{G.2.2 [Graph Theory]: Graph Algorithms--facility location}


\maketitle

%\vspace*{-.2in} 
\begin{abstract}
Let $T=(V,E)$ be a tree with associated costs on its subtrees.   A {\em minmax $k$-partition of $T$} is a partition  into $k$ subtrees, minimizing the maximum cost of a subtree over all possible partitions. In the {\em centered} version of the problem,  the cost of a subtree cost is defined as the minimum cost of ``servicing''  that subtree   using a center located within it.
 The problem motivating this work was  the sink-evacuation problem on trees, i.e.,  finding a collection of $k$-sinks that minimize the time required by a  confluent dynamic network flow to {\em evacuate} all  supplies to  sinks.

This paper provides the first polynomial-time algorithm for solving this problem, running in 
$O\Bigl( \max(k,\log n) k n \log^4 n\Bigr )$ time.  The technique developed can be used to solve any Minmax Centered $k$-Partitioning  problem on trees  in which the servicing costs satisfy some very general conditions.  Solutions can be found for both the discrete case, in which centers must be on vertices, and the continuous case, in which centers may also be placed on edges.  The technique developed also improves previous results for finding  a minmax cost $k$-partition of a tree given the location of  the sinks in advance.
\end{abstract}

\keywords{Sink Evacuation, Dynamic Flows, Confluent Flows, Facility Location, Parametric Search, Tree Partitioning, Tree Centroid.}

\section{Introduction}
\label{sec: mjgintro}

%\footnote{This introduction is modelled on \cite{ArumugamAGS14}.}.  
The main result of this paper is the derivation of a new method for solving the general  minmax centered $k$-partitioning  problem on trees. The initial motivation was the construction of quickest evacuation protocols on  dynamic tree  flow networks, a problem that was not solvable within previous tree partitioning frameworks.

\begin{figure}[t]
	\centering
	% \begin{subfigure}[t]{0.49\textwidth}
%	\includegraphics[width=3in]{maxf.eps}
\includegraphics[width=3in]{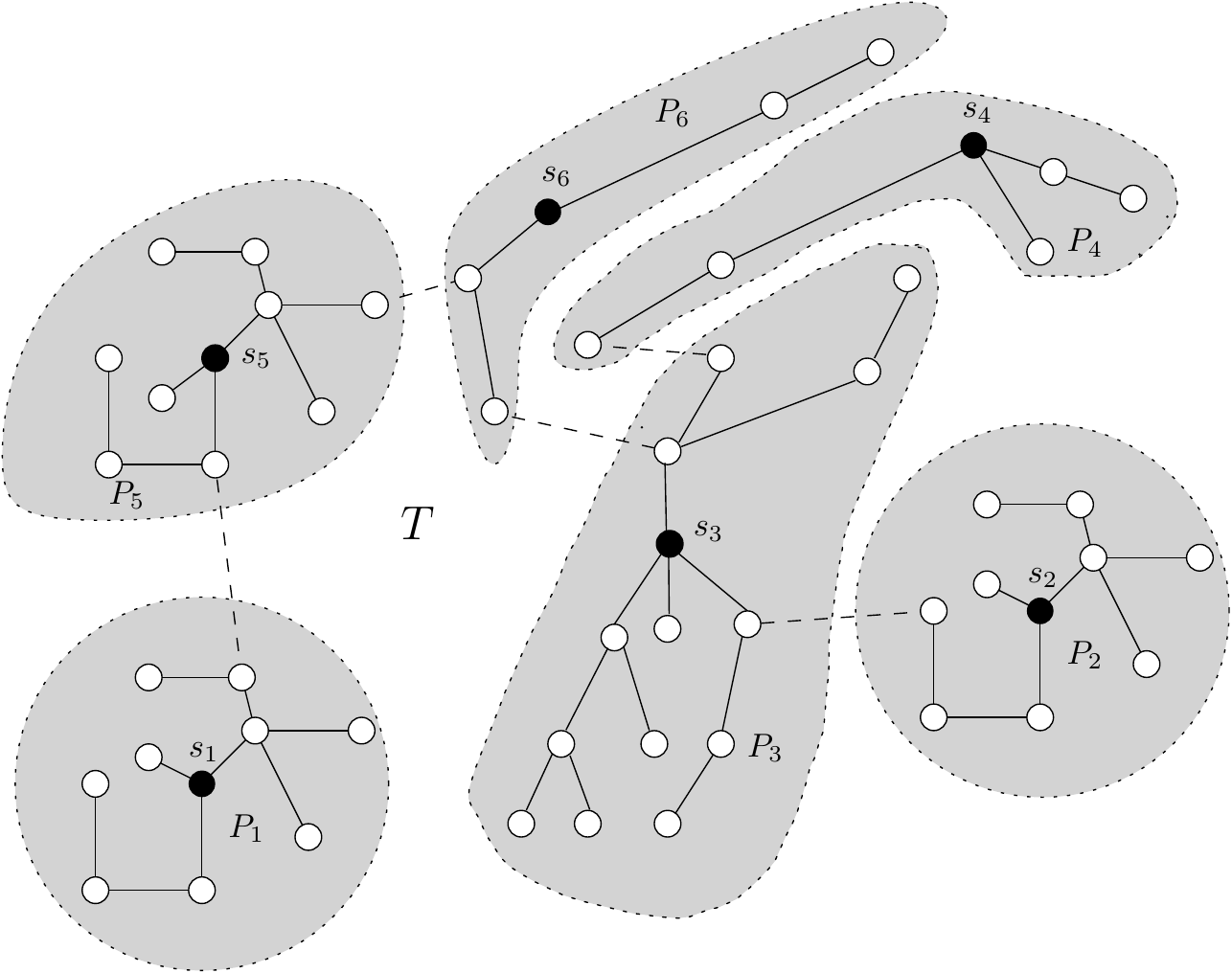}
%	\caption{Illustration of $F_S$: Partition $\mathcal{P}=\{P_1,P_2,P_3\},$ sinks $S = \{s_1,s_2,s_3\}$, $F_S(\mathcal{P}) = \max(f(P_1,s_1),f(P_2,s_2),f(P_3,s_3))$.} 
\caption{A 6-partition $\mathcal{P}=\{P_1,\ldots,P_6\},$ of tree $T$.  Component $P_i$ has associated sink or center  $s_i.$
$f(P_i,s_i)$ is the cost of servicing $P_i$ using  $s_i$.  $F_S(\mathcal{P}) = \max_{1 \le i \le 6} f(P_i,s_i)$ is the full cost of servicing $T$ with partition
$\mathcal P$ and sink set $S = \{s_1,\ldots,s_6\}$.}
	\label{fig:maxcomposition}
	%  \end{subfigure}
\end{figure}

A  {\em $k$-partition} of a tree $T=(V,E)$ is the removal of $k-1$ edges to create $k$ subtrees.  Let $f(P)$ denote the cost of subtree $P \subseteq V$ (subtrees will be denoted by their nodes). The cost of partition  $\mathcal{P}=\{P_1,\ldots,P_k\}$ is  $F(\mathcal{P}) = \max_i f(P_i).$  The {\em  minmax $k$-partition problem}  is  to find a $k$-partition $\mathcal{P}$ of $T$ that minimizes $F(\mathcal{P}).$

$f(P)$ may sometimes be  further defined as $f(P,s),$  the cost of {\em servicing} the subtree from some sink or center $s \in P.$   The cost of the partition will then be $F(\mathcal{P,S}) = \max_i f(P_i,s_i)$ where $S=\{s_1,\ldots,s_k\}$ and $s_i \in P_i.$ See Fig.~\ref{fig:maxcomposition}.
The {\em minmax centered $k$-partition problem} is to find 
${\mathcal{P}}, S$  that minimizes $F(\mathcal{P,S}).$

Becker, Perl and Schach   \cite{Becker1980} introduced a {\em shifting} algorithm for constructing minmax partitions of trees when $f(P)$ is the sum of the weights of the nodes in $P.$  This technique was then improved and generalized to other functions by them and other authors 
\cite{becker1983shifting,Perl1985,Agasi1993,Becker1995}.  \cite{Lari2015,Lari2016} discuss extensions to centered partitions.  These results only hold for the very restrictive class of {\em Invariant} functions $f(P)$  (see \cite{Becker1995} for a definition). In particular, the QFP cost that will interest us and  be defined below  will not be an invariant function.

If all nodes  $v \in V$ have given weights $w_v$ and $d(v,s)$ is the path-length distance from $v$ to $s$,  then  $f(P,s) = \max_{v \in P} w_v d(v,s)$  defines the {\em  $k$-center problem} which has its own separate  literature.
Frederickson  \cite{frederickson1991parametric}  gives an  $O(n)$ algorithm for $k$-center in an  {\em unweighted } tree, i.e.,$w_v \equiv 1$, while the weighted case can be solved in $O(n \log^2 n)$ time \cite{megiddo1981n,Cole87}.

The problem motivating this paper arises from evacuation using {\em Dynamic Confluent Flows}.
{\em Dynamic flow networks} model movement of items on a graph.

Each vertex $v$ is assigned   some initial set of supplies $w_v$.  Supplies flow across edges.  Each edge $e$ has a length $\tau_e$ -- the time required to traverse  it  -- and  a   capacity $c_e$, limiting how much flow can enter the edge in one time unit.  If all edges have the same capacity $c_e=c$ the network has  {\em uniform capacity}.   As supplies move around the graph,
 {\em congestion} can occur as supplies back up waiting to enter a  vertex, increasing the time needed to send a flow.
 
 Dynamic flow networks  were introduced by Ford and Fulkerson 
in \cite{Ford1958a}   and have since been extensively used and  analyzed.   
The {\em Quickest Flow Problem (QFP)} starts with $w_v$ units of flow on (source)  node $v$ and asks how quickly all of this flow   can be moved to designated sinks.  Good surveys of the problem and applications can be found in
\cite{Skutella2009,Aronson1989,Fleischer2007,Pascoal2006}.

One variant  of the QFP is the transshipment  problem in which  each sink has a specified demand with total source availability   equal to total demand requirement.   The problem is to find the minimum time required to satisfy all of the demands. 
 The first polynomial time algorithm for that problem was given by  \cite{Hoppe2000b} with later improvements by \cite{fleischer1998efficient}.  % for the case of integral travel times.

A variant of the QFP can also model {\em evacuation problems}, see e.g,   \cite{Higashikawa2014}  for a history.  In this,  vertex supplies can be visualized as people in one or multiple buildings and the problem is to find a routing strategy (evacuation plan) that evacuates all of them to specified sinks (exits)  in minimum time.   This differs from the transshipment problem in that the problem is to fully evacuate the sources, not to satisfy the sinks; sinks do not have predefined demands and may absorb arbitrarily large units of supply. 

An optimal solution to this problem   could assign different paths to different units of supply starting from the same vertex.  Physically, this could correspond to two people starting from the same location travelling radically different evacuation  paths, possibly even to different exits.

A constrained version of the problem, the one addressed here, is for the plan to assign to each vertex $v$ exactly one  {\em evacuation} edge, $e_v=(v,u_v),$  
 i.e., a sign stating ``this way out''. All people starting at or passing through  $v$ must evacuate through $e_v.$ After arriving at $u_v$ they continue onto $u_v$'s unique evacuation edge $e_{u_v}.$  They continue following  these unique evacuation edges  until reaching a sink,  where they exit.
The  initial problem is, given the sinks,   to determine  a plan minimizing the maximum time needed to evacuate everyone. Note that if each $v$  has a unique evacuation edge $e_v$ then the $e_v$ must form  a directed forest with the sinks being the roots of the trees.  Thus,  an evacuation plan of tree $T$ using  $k$ sinks is a centered  $k$-partition of $T.$  See  Fig.~\ref{fig:Evac Partition}. A different  version of the problem  is, given $k$, to find the (vertex) {\em locations} of the $k$ sinks/exits {\em and} associated evacuation plan that  together minimizes the evacuation time. This is the {\em $k$-sink location problem}.  

\begin{figure}[t]
	\centering
\includegraphics[width=3in]{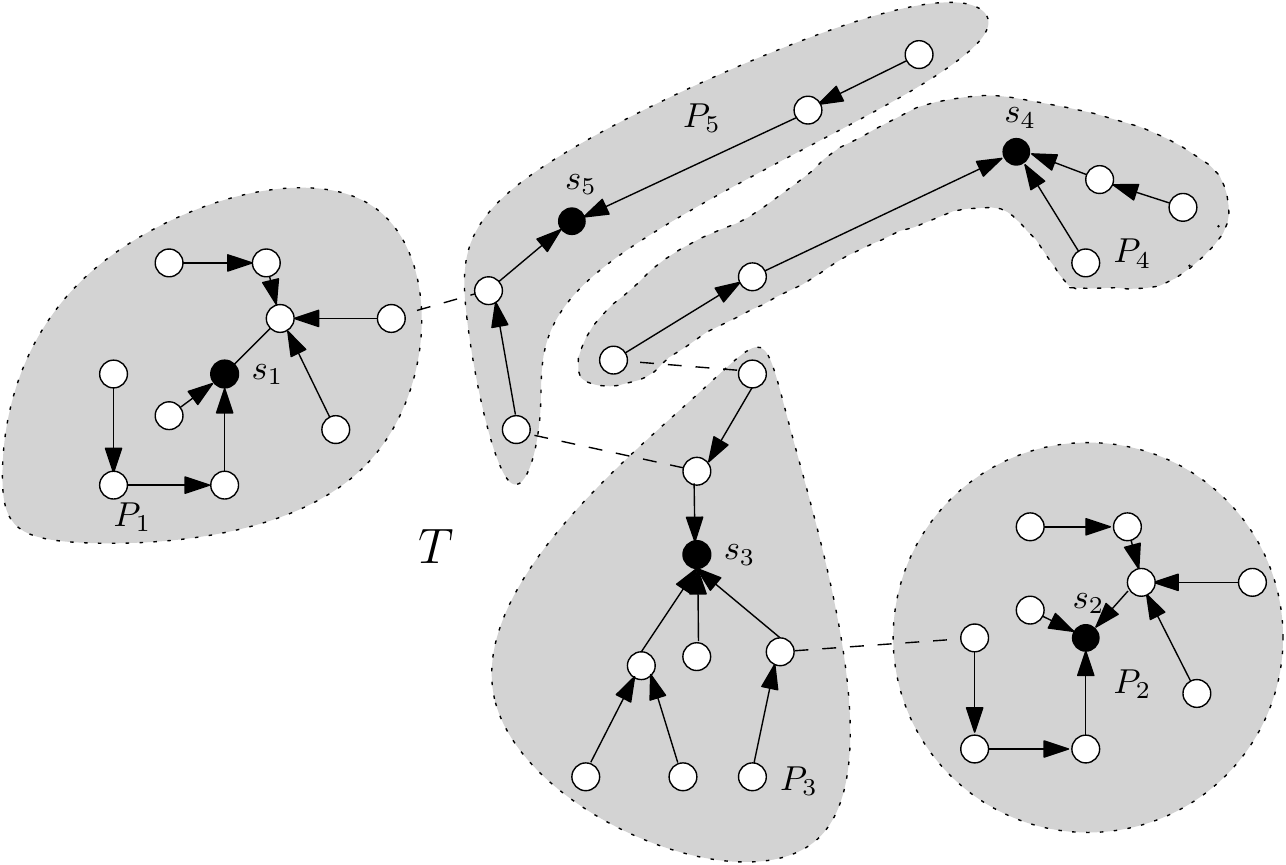}
\caption{Each vertex except for the $s_i$ has a unique associated evacuation edge (the $s_i$ evacuate to themselves). These edges form a  forest of directed in-trees with the tree roots being the $s_i$.  This forest defines  a centered $5$-partition of $T$ with centers $s_i.$  The cost of the partition will be the maximum time required for a node to evacuate to its assigned exit $s_i$.  Flows can merge and cause congestion so this evacuation time  is a function of entire subtrees and not just of individual node-sink pairs.}
	\label{fig:Evac Partition}
\end{figure}

Flows with the property that all  flows entering a vertex  leave along  the same edge are known as {\em confluent}\footnote{Confluent flows occur naturally in problems other than evacuations, e.g.,  packet forwarding and railway scheduling \cite{Dressler2010b}.}; even  in the static case constructing  an optimal confluent flow in  a general graph $G$ is known to be very difficult.   If P $\not=$ NP, then it is  impossible to construct a constant-factor approximate optimal confluent flow in polynomial time on a general graph  \cite{Chen2007,Dressler2010b,Chen2006,Shepherd2015} even with only one sink.

If  edge capacities are  ``large enough'' then no congestion occurs and every person starting at node $v$  should follow the same shortest path it can to an exit.  The cost of the plan will be  the length of the  maximum shortest path.   Minimizing this is  is exactly the $k$-center problem on graphs which is already known to be NP-Hard \cite[ND50]{garey1979computers}. Unlike  $k$-center, which is polynomial-time solvable for fixed $k$,  Kamiyama {\em et al.}~\cite{Kamiyama} proves 
 by reduction to {\em Partition},   that, even for  $k=1$,   finding the min-time evacuation protocol is still NP-Hard for general graphs.
 This was later extended \cite{Golin2017sink}  to show that even for  $k=1$ and the sink  location fixed in advance, it is still impossible to approximate the QFP time to within a factor of $o(\log n)$ if P $\not=$ NP.

The only solvable known case for the sink location problem for general $k$ is for $G$ a path \cite{bhattacharya2017improved}.
For paths with uniform capacities this runs in
$\min\bigl(O(n + k^2   \log^2  n),\ O(n \log n)\bigr)$ time; for paths with general capacities in 
$\min\bigl(O(n \log n + k^2  \log^4  n),\, O(n \log^3  n)\bigr)$ time.

When $G$ is a tree,
the $1$-sink location problem can be solved    \cite{Mamada2006} in $O(n \log^2 n)$ time.   This can be reduced  \cite{Higashikawa2014,bhattacharya2015improved}  down to $O(n \log n)$ for the uniform capacity version, i.e., all  the $c_e$ are identical. If the {\em locations of the $k$ sinks are given as input},
 \cite{Mamada2005a} gives a  $O(n (c \log n)^{k+1})$ time algorithm  evacuation protocol,  where $c$ is some constant. This is the problem of partitioning the tree optimally, given that the centers are already known.   For ``large'' $k$,  \cite{Mamada2005}  reduced the time  down to $O(n^2 k \log^2  n)$. The literature does not contain any algorithm for solving the sink-location problem on trees.
The best solution using current known results would be to try all possible  $\Theta(n^{k-1})$ decompositions of the tree into $k$ subtrees and apply the algorithm of \cite{Mamada2006}, yielding
$O(n^k \log^2 n)$ time.

When $k=1$,  \cite{Mamada2005a}  also provides  an $O(n \log^2 n)$ algorithm for calculating the evacuation cost to a  single known sink.  For the uniform capacity case,  \cite[p.~34]{Higashikawa2014c} gives a formula that reduces  the calculation time down to  $O(n \log n)$.  These two calculation algorithms will be used as oracles in the sequel.

The discussion above implicitly assumed that the sinks must be vertices of the original graph.  This is known as the {\em discrete} case. 
Another possibility would be to permit  sinks to be located anywhere,  on edges as well as vertices.  This variation is  known as the {\em continuous} case.

 This distinction occurs in evacuation modelling, e.g., locating an emergency exit in a hallway {\em between rooms}.  Historically, this distinction is also explicit  in the $k-$center in a tree  literature.  More specifically, Frederickson's  \cite{frederickson1991parametric}  $O(n)$ algorithm for $k$-center in an  {\em unweighted } tree worked in both the continuous and discrete cases.  For weighted $k$-center, though,  the two cases needed two different sets of techniques.
 \cite{megiddo1981n} gave an $O(n \log^2 n)$ algorithm for the discrete case while the continuous case required  $O(n \log^2 n \log \log n)$  time  \cite{megiddo1983new}.  It was only later  realized that a parametric searching technique \cite{Cole87} could reduce the continuous case down to $O(n \log^2 n)$ as well.   Weighted $k$-center restricted  to the line can be solved in $O(n\log n)$ in both the discrete and continuous cases but were also originally solved separately; \cite{chen2015efficient} provides a good discussion of the history of that problem.

\subsection{Our contributions}
This paper gives the first polynomial time algorithm for solving the $k$-sink location problem on trees. It uses as an oracle a known algorithm for calculating the cost of the problem when $k=1$ and the sink is known in advance.
Our results will be applicable to both the discrete and continuous versions of the problem.

\begin{theorem}
	The $k$-sink evacuation problem can be solved in
	\begin{itemize}
	\item $O( \max(k,\log n) \, k n  \log^4 n)$ time 
	 for general-capacity edges and
	 \item  $O( \max(k,\log n)\,  k n  \log^3 n)$ time for uniform-capacity edges.
	 \end{itemize}
	\label{theorem:kSinkRunTime}
\end{theorem}

This  result will be a special case of a general technique that works for a large variety of minmax cost functions on trees.  Section \ref{Sec:Prob Def} formally defines the Sink-Evacuation problem on trees, 
the more general class of functions for which our technique works and then states our results.

It is instructive to compare our approach to Frederickson's  \cite{frederickson1991parametric}  $O(n)$ algorithm for solving the unweighted $k$-center problem on trees, which was built from the  following two ingredients.

\begin{enumerate}
\item An $O(n)$ time previously known algorithm for checking {\em feasibility}, i.e., given $\alpha >0$, testing whether a  $k$-center solution with cost $\le \alpha$ {\em  exists}
\item A clever {\em parametric search} method to filter the $O(n^2)$ pairwise distances between nodes, one of which is the optimal cost, via the feasibility test.
\end{enumerate}

The main difficulty in solving the sink-evacuation problem is that no polynomial time feasibility test for $k$-sink evacuation on trees was previously known.  The majority of this paper is devoted to constructing such a test.  Section \ref{Section: Bounded Cost} derives useful properties of the feasibility problem and 
Section \ref {Sec: Bounded Algorithm} utilizes these properties to construct an algorithm. This algorithm works by making 
$O(k \log n)$ (amortized) calls to the $1$ fixed-sink algorithm oracle.

There is also no small set of easily defined cost values known to contain the optimal solution.  We sidestep this issue in Section \ref{Section: Full Problem} by doing parametric searching {\em within} our feasibility testing algorithm,  leading to Theorem \ref{theorem:kSinkRunTime}.

Sections   \ref{Section: Bounded Cost}, \ref {Sec: Bounded Algorithm}  and  \ref{Section: Full Problem}   assume the discrete version of the problem.  Section \ref{Sec:Continuous} describes the modifications necessary to extend the algorithm to work in the continuous case.

In Section \ref{Sec:Fixed Sink} we conclude by noting that a slight modification to the algorithm allows improving, for almost all $k$,  the best previously  known algorithm for solving the problem when the $k$-sink locations are predetermined;  from  $O(n^2 k \log^2 n)$ \cite{Mamada2005}  down to  $O(n k^2 \log^4 n)$.

\section{Definitions and Results}
\label{Sec:Prob Def}

Let $G=(V,E)$ be an undirected graph. Each  edge $e=(u,v)$   has a travel time $\tau_e$;  flow leaving  $u$ at time $t=t_0$ arrives at $v$ at time $t=t_0 + \tau_e.$ 
Each edge also has  a
{\em capacity} $c_e \ge 0$.
 This  restricts  at most  $c_e$ units of flow to enter edge $e$ per every unit of time.
 % For our version of the problem we restrict $c$ to be integral; the capacity can then be visualized as the {\em width} of the edge with only $c_e$ people being allowed to travel in parallel along  the edge.

Consider $w_u$ units of (supply) flow  waiting at vertex $u$ at time $t=0$ to traverse edge $e=(u,v).$  They enter $e$ at a  rate of $c_e$ units of flow  per  unit time so the last flow  enters $e$  at time   $w_u/c_e.$ This flow then travels another $\tau_e$ time to reach $v.$  The total time required to move all flow from $u$ to $v$ is then $w_u/c_e + \tau_e$.

If two edges were combined in a path from $u \rightarrow v \rightarrow s$ then flow  from $u$ travelling to $s$ might have to wait at $v$ for all the $w_v$ flow to first enter $(v,s)$.  When multiple paths meet, this results in {\em congestion} that can delay evacuation time in strange ways.

\begin{figure}[t]
	\centering
	% \begin{subfigure}[t]{0.49\textwidth}
%	\includegraphics[width=3in]{maxf.eps}
\includegraphics[width=\textwidth]{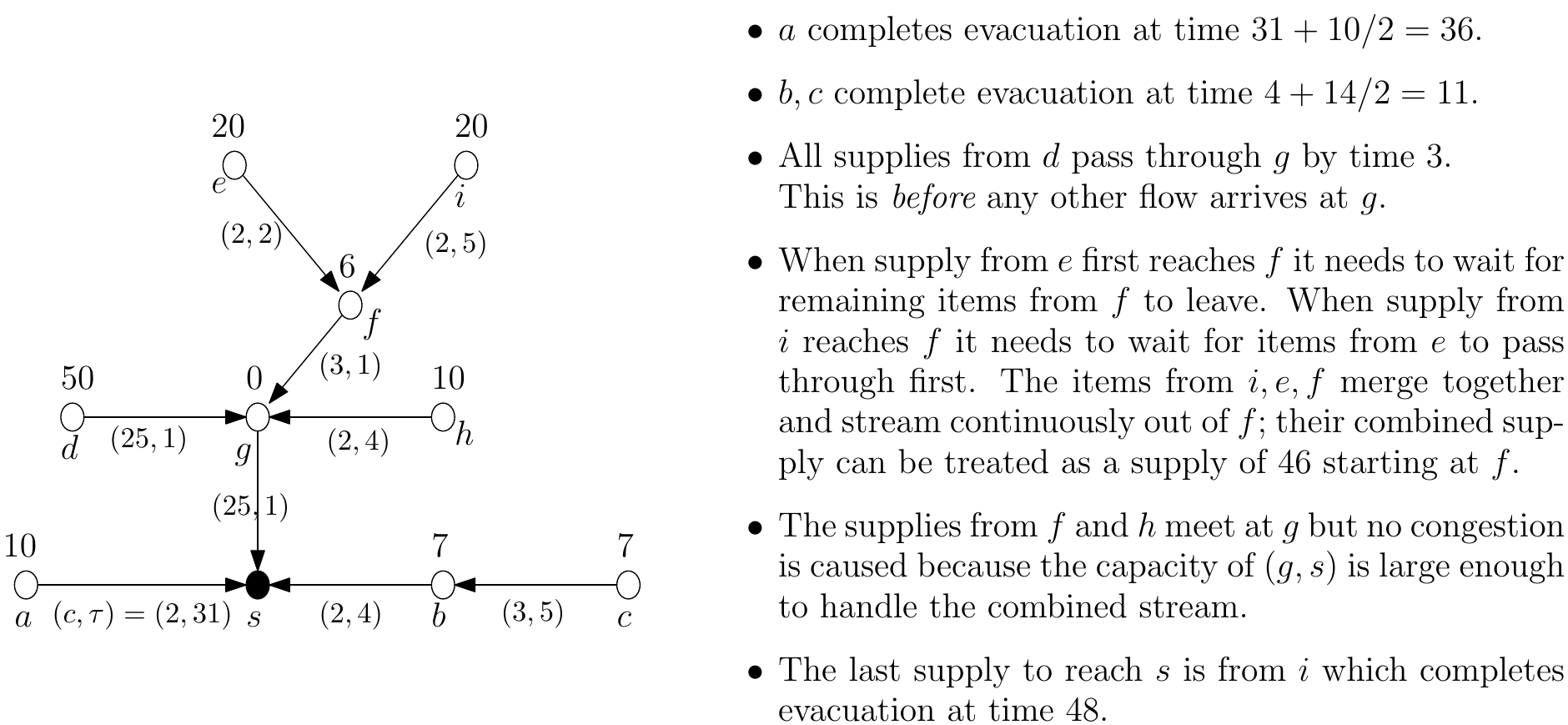}
%	\caption{Illustration of $F_S$: Partition $\mathcal{P}=\{P_1,P_2,P_3\},$ sinks $S = \{s_1,s_2,s_3\}$, $F_S(\mathcal{P}) = \max(f(P_1,s_1),f(P_2,s_2),f(P_3,s_3))$.} 
\caption{An illustration of evacuation of a tree to  sink, $s.$  Initial values $w_v$ are above $v.$  Each edge $e$ is labelled with a (capacity, length) pair $(c_e,\tau_e).$  The goal is to evacuate all supplies to $s$. Note that the tree contains 3 different branches containing, respectively,  $\{a\},$  $\{b,c\}$  and  $\{d,e,f,g,h,i\}$ whose evacuation times can be calculated separately.  The time required to evacuate all supplies to $s$ is $48,$ which is when the last supply from $i$ arrives at $s.$ }
	\label{fig:Evac Full Example}
	%  \end{subfigure}
\end{figure}

Figure \ref {fig:Evac Full Example} illustrates different types of congestion and gives an example of calculating the evacuation time of a tree to a given sink.

Given a graph $G$, distinguish a subset $S \subseteq V$ with $|S|=k$  as sinks (exits).
An evacuation plan specifies, for each vertex $v\not\in S$,  the unique edge along which all flow starting at or passing through $v$ evacuates.  Furthermore, starting at any $v$ and following  the edges will lead from $v$ to one of the $S$  (if $v \in S$, flow at $v$ evacuates immediately through the exit at $v$). 
As noted earlier  (Figure \ref{fig:Evac Partition}) the evacuation plan defines a confluent flow.  The evacuation edges form a directed forest;  the root of each tree  is  one of the designated sinks in $S.$

Given evacuation plan $\mathcal{P}$ and the $w_v$ specifying the initial flow supply starting at each node,  one can calculate, for each vertex, the  time (with congestion) required for all of its flow supply to evacuate.  The maximum of this over all $v$  is the minimum time required to required to evacuate  {\em all}  items  to some exit  using the rules above. Call this the {\em cost for $S$ associated with the evacuation plan} 
and  denote it by  $f(\mathcal{P},S).$ 

The {\em $k$-sink location} problem is to find a subset $S$ of size $k$ and associated $\mathcal{P}$ that minimizes $f(\mathcal{P},S).$

\subsection{General problem formulation}
\label{subsec:GPF}
The input  will be 
  a tree $\Tin = (\Vin,\Ein)$,  and  a positive integer $k$.  Let $n = |\Vin| = |\Ein| + 1$. The output will be    $S \subseteq \Vin$,  $|S| \le k,$  and an associated partition $\mathcal{P}$ of $\Tin$ into $|S|$ subtrees, each containing  one vertex  in $S,$ that minimizes $f(\mathcal{P},S)$ over all possible such pairs.

  The  algorithms will not explicitly deal with the complicated mechanics of evacuation calculations.  Instead they will solve the location problem for any {\em minmax  monotone cost $\emptyf$},  given an oracle for solving a one-sink problem in which the location of the sink is pre-specified. 

This level of abstraction  simplifies the formulation and understanding of the algorithms.  It can also be useful for solving other similar problems. 

%%\vspace*{-.2in}
\subsubsection{Minmax monotone cost functions.}
\label{subsec:minmaxdef}
%\mc{Weighted center}

Minmax monotone cost functions are defined below.  Note that this definition is 
 consistent with the specific properties of the evacuation problem. 
\begin{definition} 
\label{def:Partitions}
Let $\Tin=(\Vin,\Ein)$  be a tree.

Let $U \subset V.$  The phrase ``$U$ is a subtree of $T$'' will denote that the graph induced by $U$ in $T$ is a subtree of $T.$

For any $u \in \Vin$,  $\Gamma(u) =\{v \in \Vin \,:\, (u,v) \in\ inE\}$ are the {\em neighbors} of $u$.  For
$U \subset \Vin,$  $\Gamma(U) = \bigcup_{u \in U} \Gamma(u)$  are the {\em neighbors} of $U$.

A {\em Partition} of  $V \subseteq \Vin$ is $\mathcal{P} =\{P_1,P_2,\ldots,P_t\}$ such that each $P_i$ is a subtree, $\cup_i P_i = V$,  and $\forall i \not=j$,  $P_i \cap P_j=\emptyset$.
$P_i$ are  the {\em blocks} of $\mathcal{P}$.

Let $S= \{s_1,s_2,\ldots,s_t \} \subseteq V.$  
$\Lambda[S]$ will denote the set of all partitions $\mathcal{P} =\{P_{s_1},P_{s_2},\ldots,P_{s_t}\}$ of $\Vin$ such that $\forall i,$   $ S \cap P_{s_i} = \{s_i\}$.

 Nodes in $P_{s_i}$ are {\em assigned to} the \emph{sink} $s_i$.
For simplicity, we often will % not fully describe partitions but 
just say that {\em  (node) $v$ is assigned to (sink) $s_i.$ }
\end{definition}

%For any $P \subseteq \Vin$ s.t. $ | S \cap P | = 1$, denote by $\langle S \cap P \rangle$ the unique node $s \in S \cap P$.
% Nodes in $P$ are {\em assigned to} the \emph{sink} $\langle S \cap P \rangle$.
%For simplicity, we often will % not fully describe partitions but 
%just say that {\em  node $v$ is assigned to sink $s.$ }

Let $f : 2^{\Vin} \times \Vin \rightarrow [0,+\infty]$  be an \emph{atomic cost function}. 
$f(P,s)$ can be interpreted as 
{\em the cost for sink $s$ to serve the set of nodes $P$}. 
This interpretation of $f$ imposes the following  natural constraints:
\begin{enumerate}
\item For $U \subseteq \Vin$, $s \in \Vin$,
  \begin{itemize}
\item if $U = \{s\}$, then $f(U,s) = 0$.
\item if $U$ is not  a subtree of $\Tin$, then $f(U,s) = +\infty$.
  \item if $s \notin U$ then $f(U,s) = +\infty$;
  \end{itemize}
\item {\em Set monotonicity} \\If $s \in U_1 \subseteq U_2 \subseteq \Vin$, then $f(U_1,s) \leq f(U_2,s)$,\\ i.e. the cost can  not decrease when a sink has to serve additional nodes.

\item {\em Path monotonicity}  \\Let $U \subset \Vin$ and $s \not\in U$ but $S \in \Gamma(U).$ Then $f(U \cup \{ s \},s) \geq f(U,u)$. Intuitively, this means  that as a  sink serving $U$ moves away  from $U,$ the cost of servicing $U$ can not decrease. 

\item {\em Max tree composition} (Fig. \ref{fig:Max_Comp})\\ Let $T = (U,E')$ be a subtree of $\Tin$ and $s \in U$ a node with $t$ neighbors.  Set $\mathcal{F} = \{ T_1,...,T_t \}$ to be the forest created by removing $s$ from $T$, and $U_1,...,U_t$ the respective vertices of each tree in $\mathcal{F}$. Then $$f(U,s) = \max_{1\leq i\leq t} f(U_i \cup \{ s\},s).$$
The subtrees $U_i$ will be called {\em slices} of $U$ defined by $s.$
\end{enumerate}

\begin{figure}[t]
	\centering
\includegraphics[width=2.3in]{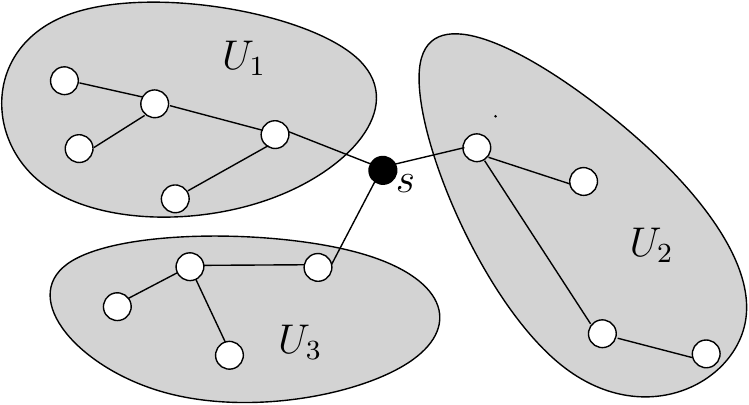}
\caption{Example of Max-Composition.  Let  $U$ denote the complete tree.  Removing $s$ creates a forest with three trees, $U_1, U_2, U_3.$
By definition,  $f(U,s) = \max\{f(U_1\cup \{ s\},s), \, f(U_2\cup \{ s\},s),\, f(U_3\cup \{ s\},s)\} .$}
	\label{fig:Max_Comp}
\end{figure}

Note  that 1-5 only define  a cost function over one subtree  and one single sink. Function $\emptyf$ is now naturally extended to work on  on partitions and sets  (Fig. \ref{fig:maxcomposition}).
\begin{enumerate}[resume]
\item {\em Max partition  composition}   
\begin{equation}
\forall \mathcal{P} \in \Lambda[S],\quad 
f(\mathcal{P},S) = \max_{P_{s_i} \in \mathcal{P}} f(P,s_i).
\label{eq:MaxCompositionB}
\end{equation}
\end{enumerate}

\begin{definition}
A  cost function $f(\mathcal{P},S$   that satisfies   properties 1-5 is called {\em minmax monotone}. 
\end{definition}

 Given $k > 0$, the  main  problem will be to find an $S^*$ and  $\mathcal{P}^* \in \Lambda[S^*]$ that satisfy
\begin{equation}
\label{eq:goal}
f(\mathcal{P}^*,S^*) = \min_{S \subseteq V,\, |S| \le k,\, \mathcal{P} \in \Lambda[S]} f(\mathcal{P},S).
\end{equation}

Our algorithms  make calls directly to an oracle $\mathcal{A}$ that,  given subtree $U$ of $\Tin$ and $v \in U,$ 
computes $f(U,v)$.
As mentioned, in our case of interest, \cite{Mamada2005a} provides an  $O(n \log^2 n)$ oracle for general-capacity sink evacuation and \cite[p.~34]{Higashikawa2014c} provides $O(n \log n)$ oracle for uniform-capacity sink evacuation.

Finally, later amortization arguments will require the following definition:
\begin{definition}
If $\mathcal{A}$ runs in time $t_\mathcal{A}(n)$, then  $\mathcal{A}$  is {\em asymptotically subadditive} if
\begin{itemize}
\item   $t_\mathcal{A}(n) = \Omega(n)$ and is non-decreasing.
\item For all nonnegative $n_i$,  $\sum_i t_\mathcal{A}(n_i)   = O \left( t_\mathcal{A}\Bigl( \sum_i n_i \Bigr) \right).$
\item $t_\mathcal{A}(n+1) = O\left(t_\mathcal{A}(n)\right)$
\end{itemize}
\end{definition}
Note that for $x\ge 1$ and  $y\ge 0$, any function of the form  $n^x \log^y n$ is asymptotically subadditive so, in particular, the oracles mentioned above are asymptotically subadditive.

%In the sequel, the statement ``The procedure takes $O(f(n) )$ amortized oracle calls'' denotes  that the running time of all of the oracle calls will use  $O(f(n) t_{\mathcal{A}})$ time. A statement of this type will result from partitioning of oracle calls and asymptotic suboptimality.   It does {\em  not} inply that only $O(f(n))$ oracle calls were performed.  This distinction is emphasized because it will become crucial when implementing parametric search later in Section  \ref{Section: Full Problem}.

\subsection{Results}
\label{subsec:results}
The remainder of the paper is devoted to deriving two algorithms.

The first  algorithm  and the majority of the paper,  provides a feasibility test,   which solves a simplified, \emph{bounded-cost} version of the problem.  Given $k$ and $\mathcal{T},$ determine whether  there exists a $k$-partition with cost at most $\mathcal{T}.$

\begin{table}[H]
	\centering
	\begin{tabular}{ |r|p{0.8\textwidth}| }
		\hline
		{\bfseries Problem} & Bounded cost minmax $k$-sink \\
		\hline
		Input & Tree $\Tin=(\Vin,\Ein)$, \ $k \geq 1$, \ $\mathcal{T} \geq 0$\\[0.01in]
		Output & $\Sout \subseteq \Vin$ and
		                                   $\Pout  \in \Lambda[\Sout]$ \ s.t. \  $|\Sout| \leq k$ and $f(\mathcal{P_{\mathrm{out}}},{\Sout}) \leq \mathcal{T}$.\\\
		& If such  a  $(\Sout, \mathcal{P_{\mathrm{out}}})$  pair does not exist, output `No'. \\
		\hline
	\end{tabular}
\end{table}

The second algorithm is for the original general problem.  To find the location of $k$ sinks that minimize the cost of a $k$-partition.

\begin{table}[H]
	\centering
	\begin{tabular}{ |r|p{0.75\textwidth}| }
		\hline
		{\bfseries Problem} & Minmax $k$-sink \\
		\hline
		Input & Tree $\Tin=(\Vin,\Ein)$, $k \geq 1$\\
		Output & $\Sout \subseteq \Vin$ satisfying  $|S| \le k$  and  $\Pout  \in \Lambda[\Sout]$ satisfying Eq.~(\ref{eq:goal})\\
%		                                   $\Pout  \in \Lambda[\Sout]$ \ s.t. \\ 
%		           &$F_{\Sout}(\Pout) = \min_{S \subseteq V, \, |S| \le k} F(S)$ \\ % \min_{\mathcal{P} \in \Lambda(S)} F_S(\mathcal{P})$ \\
		\hline
	\end{tabular}
\end{table}

Our first result is 

\begin{theorem}
	If $\mathcal{A}$ is an asymptotically subadditive algorithm for solving the fixed $1$-sink problem that runs  in   $t_\mathcal{A}(n)$ time, then  the bounded cost minmax $k$-sink problem can be solved in time $O(k t_\mathcal{A}(n) \log n)$. \label{theorem:FastBC}
\end{theorem}

Combining this algorithm with  a careful application of parametric searching will yield a 
a solution to the general problem:

\begin{theorem}
	If $\mathcal{A}$ is an asymptotically subadditive algorithm for solving the fixed  $1$-sink problem that runs  in   $t_\mathcal{A}(n)$ time then  the  minmax $k$-sink problem can be solved in time $O( \max(k,\log n) k t_{\mathcal{A}}(n) \log^2 n ).$
	 \label{theorem:FastC}
\end{theorem}

Theorem \ref{theorem:kSinkRunTime} follows directly from this and  the $1$-sink algorithms  given by 
 \cite{Mamada2005a} and  \cite[p.~34]{Higashikawa2014c}.

A simple modification of the 2nd algorithm will also solve the specialized partitioning version in which  the sinks are fixed in advance. We will call a minmax monotone function {\em relaxed} if the defining $\emptyf$ satisfies properties 1,2, and 4 from Section \ref{subsec:minmaxdef} but does not necessarily satisfy property 3 (path monotonicity).\footnote{Because the sinks are predefined, they never move and path monotonicity is superfluous.}

\begin{table}[H]
	\centering
	\begin{tabular}{ |r|p{0.7\textwidth}| }
		\hline
		{\bfseries Problem} &Relaxed  Minmax $k$ fixed-sink  \\
		\hline
		Input & Tree $\Tin=(\Vin,\Ein)$,  $S \subseteq V,$  \, $|S|=k$ \\
%		Output & $\Pout  \in \Lambda[S]$ \ s.t.  $ F_S(\Pout)  =  F(S)$ \\
		Output & $\Pout  \in \Lambda[S]$ \ s.t.  $ f(\Pout,S)  =  \min_{\mathcal{P} \in  \Lambda[S]} f(\mathcal{P},S)$ \\
		\hline
	\end{tabular}
\end{table}
\begin{theorem}
	If $\mathcal{A}$ is an asymptotically subadditive algorithm for solving the   fixed relaxed  $1$-sink problem that runs  in   $t_\mathcal{A}(n)$ time and that further satisfies $t_\mathcal{A}(2n) = O(t_\mathcal{A}(n))$, 
	 then  
	the minmax $k$ fixed-sink    problem can be solved in time $O(  k^2  t_{\mathcal{A}}(n)  \log^2 n).$
	 \label{theorem:FastCF}
\end{theorem}

For  the sink evacuation problem, plugging  the $O(n \log^2 n)$ oracle  into Theorem \ref{theorem:FastCF} leads to a $O(n k^2 \log^4 n)$ time  algorithm, substantially improving upon the previously known 
 $O(n (c \log n)^{k+1})$  \cite{Mamada2005a} and   $O(n^2 k \log^2  n)$ \cite{Mamada2005} algorithms when $4 < k \ll n$.

All the results mentioned above hold for both the discrete and the continuous versions of the problem.

\subsection{More Applications}
Although our  algorithm was motivated by  confluent dynamic flows it is surprisingly easy to apply to  unrelated problems.  We provide three examples below. The input is always a tree $\Tin=(\Vin,\Ein).$

Given a tree $U$ and center $s \in U$ recall the definition of {\em slices $U_i$}  in the max-composition rule. 
The $U_i$ were the subtrees that resulted by removing $s.$

\medskip
\par\noindent\underline{Example 1: Weighted $k$-center}\\
Each vertex has weight $w_v$ and each edge $(u,v)$ has  length $d(u,v)$.
For  any pair $(u',v') \not\in E$,   $d(u',v')$ is the sum of the lengths of the edges on the unique path connecting $(u',v')$

As  a warm-up application we note that our algorithm immediately yields a (non-optimal)  algorithm for weighted $k$ center by setting
$$
f(U \cup\{s\},s) =\max_{u \in U} w_u d(s,u) 
$$
where $U$ is a subtree, $s\not\in U$ but  $s \in \Gamma(U)$.  
This $f(\cdot,\cdot)$ satisfies the minmax  monotone cost function properties laid out in Section \ref{subsec:minmaxdef}
and can be evaluated in  $O(|U|)$ time using a breadth-first search scan of the tree.
 Thus Theorem \ref{theorem:FastC} yields an $O(\max(k, \log n)  k n \log^2 n)$ time algorithm for solving
 the weighted $k$-center problem.

 \medskip

 The algorithm above is slower than  the  $O(n \log^2)$ algorithms of \cite{megiddo1981n,Cole87}.
 But, those algorithms strongly use parametric searching in a polynomially bounded space (costs defined by pairs of vertices).  It would be difficult to modify them to include general constraints. As illustrated below,  Theorem \ref{theorem:FastC}   permits adding  many types of constraints without any increase in running time.

\medskip

\par\noindent\underline{Example 2: Weight constrained weighted $k$-center}\\
Now denote the weight of a subtree $U \subseteq \Vin$ by $W(U) =\sum_{u \in U} w_u.$

Consider the following combination of the weighted $k$-center problem and minmax weight-partitioning problem \cite{Becker1980} that adds  the  constraint that the weight of all slices is at most some  fixed threshold  $W >0.$
This $W$ can be viewed as a natural limit on the capacity of the service center $s.$

For $U$ is a subtree, $s\not\in U$ but  $s \in \Gamma(U)$ set 
\begin{equation}
\label{eq:WCkC}
f(U \cup\{s\},s) =
\left\{
\begin{array}{ll}
\max_{u \in U} w_u d(s,u) \quad & \mbox{if $W(U) \le W$,}\\
\infty			& \mbox{Otherwise}
\end{array}
\right.
\end{equation}
This function also satisfies the minmax  monotone cost function properties
and can still easily be evaluated by a breadth-first search scan of the tree in $O(|U|)$ time.  Solving the minmax $k$-sink problem for this $\emptyf$ function using Theorem \ref{theorem:FastC}  exactly solves the weighted $k$-center problem in which each slice is constrained to have weight at most $W$  in  %The algorithm with this constraint still runs in  
$O(\max(k, \log n)  k n \log^2 n)$ time.

Adding additional constraints is not difficult.  If $d_H(u,v)$ is defined to be the number of edges (hop distance) on the path connecting $u$ and $v$ we could replace 
(\ref{eq:WCkC}) with
$$
f(U \cup\{s\},s) =
\left\{
\begin{array}{ll}
\max_{u \in U} d(s,u) \quad & \mbox{if $W(U) \le W$ and  $d_H(u,v) \le h$,}\\
\infty			& \mbox{Otherwise}
\end{array}
\right.
$$
and the algorithm now  exactly solves the weighted $k$-center problem in which each slice is constrained to have weight at most $W$ and no node can be more than $h$ edges from a center.  The running time remains the same because $f(U,s)$ can still be evaluated in $O(|U|)$ time.

\medskip

\par\noindent\underline{Example 3: Minmax range partitioning}\\
Motivated by obtaining {\em balanced} solutions \cite{Lari2015}  discusses partitioning using {\em range criteria}.  In this problem
the $k$ sinks $S$ are specified in the input.  For every $u \in \Vin \setminus S$ and $s \in S$, $c_{i s}$ is a given  cost of servicing $u$ with sink $s$.

The {\em range-cost} of $U$ serviced by  $s$ is
$$f(U,s) = \max_{u,v\in U} |c_{u s} - c_{v s}|,$$  i.e., the difference between the maximum and minimum service costs.
The problem is to do centered $k$-partitioning of the tree so as to minimize the maximum range-cost of a subtree. 

\cite{Lari2015} gives an $O(k^2 n^2)$ algorithm for this problem.  
While our algorithm  can not solve this exact problem it can solve the variation  when the range-costs are restricted to slices. That is
when   $s\not\in U$ but  $s \in \Gamma(U)$ set 
$$
f(U \cup\{s\},s) = \max_{u,v\in U} |c_{u s} - c_{v s}|.
$$
Note that this yields a {\em relaxed} minmax monotone function.
%(Set Monotonicity still holds because increasing $U$ can only increase the range-cost).
Since the range-cost can be calculated in $O(|U|)$ time,  Theorem \ref{theorem:FastCF} yields an $O(k^2 n \log^2 n)$ algorithm for finding a $k$-partition in which the {\em max range-cost of a slice} is minimized, almost an order of magnitude faster  than the algorithm for the original problem, 

We end by noting that the algorithm would remain valid if  the range-cost was defined by minimizing the ratio between servicing costs within a slice rather than the absolute difference, i.e.,  setting
$$
f(U \cup\{s\},s) = \max_{u,v\in U} \frac {c_{u s}} {c_{v s}}.
$$

\section{Useful Properties of the  Discrete Bounded-Cost Problem}
\label{Section: Bounded Cost}

This section derives structural properties  that will permit designing  an algorithm.
 In both this  section 
and Section \ref{Sec: Bounded Algorithm}, $k$ and $\mathcal T$ are %considered 
fixed given values.

\begin{definition}\ 
\begin{itemize}
\item  A {\em  sink configuration} is a set of sinks $S=\{s_1,s_2,\ldots,s_t\} \subseteq \Vin$ %,  $|S| \le k,$  
and associated  partition $\mathcal{P}= (P_{s_1}, P_{s_2},\ldots, P_{s_t}) \in \Lambda(S)$. 
\item 
A {\em feasible sink configuration} is a  sink configuration  satisfying $F_S(\mathcal{P}) \leq \mathcal{T}$;  $S$ is  a {\em feasible sink placement}, and $\mathcal{P}$ is {\em a partition  witnessing the feasibility of $S$}. 
\item An {\em optimal} feasible configuration is a { feasible  sink configuration $(S^*,\Pstar)$ with minimum cardinality}; we write $k^* := |S^*|$.
\end{itemize}
\end{definition}

\begin{definition}
Let $\Sout=\{s_1,s_2,\ldots,s_t\}$ and  $\Pout = (P_{s_1}, P_{s_2},\ldots, P_{s_t})$  be a partition of some $V \subseteq \Vin$ such that 
 $\forall i,$  $s_i \in P_i$ and 
$f(P_i,s_i)\le \mathcal{T}$. 
%the graph induced by $P_i$ in $T$ is a subtree of $T.$ 
Then $(\Sout,\Pout)$ is a {\em partial sink configuration}
\end{definition}

\begin{definition}
Let  $S \subseteq U  \subseteq V$ where $U$  is a subtree of $\Tin$.  
%$T_\mathrm{in}$.
 $U$ is {\em served by $S$} if, for some partition $\mathcal{P}$ of $U$, for each $P \in \mathcal{P}$ there exists $s \in S$ such that $f(P,s) \leq \mathcal{T}$.
\end{definition}
Note that  $(\Sout,\Pout)$ being  a partial sink configuration implies  that  $\bigcup_{s_i \in \Sout} P_{s_i}$ is served by  $\Sout$.

\begin{definition}
  Let $U \subseteq V$ be a subtree of $T_\mathrm{in}$   and $v\in V$ (not necessarily in $U$). $v$ {\em supports} $U$ if one of the following holds:  
  \begin{itemize}
    \item If $v \in U$, then $f(U,v) \leq \mathcal{T}$.
    \item If $v \notin U$, let $\Pi$ be the set of nodes on the path from $v$ to $U$, inclusive of $v.$  Then $f(U  \cup \Pi, v) \leq \mathcal{T}$.
  \end{itemize}
\end{definition}

Note that if $U$ can be served by $S$, then for any node in $u \in U$, $\{u\}$  is supported by some $s \in S$. The converse is not generally true.

\begin{definition}
Let $u,v \in \Vin$.  $\Pi(u,v)$ denote the unique directed path from $u$ to $v$ inclusive of $u,v.$ 
\end{definition}

%\vspace*{-.2in}
\subsection{Greedy construction}
\label{subsection: greedy construction}

Our algorithm  greedily grows  $(\Sout,\Pout)$, maintaining   the property that it will always be able to be completed to be an optimal  feasible configuration.  Thus, when the algorithm stops,  $(\Sout,\Pout)$ is either an optimal feasible configuration  with $k^* \le k$, or the algorithm answers no because $k^*  \ge \Sout > k$.
The algorithm also maintains a {\em Working Tree $T=(V,E)$} containing uncommitted  vertices  and  a set $S  \subseteq \Sout$ of sinks that may still have more nodes committed to them. 
%and set of sinks  $S = \Sout \cap V.$  

 At the start of the algorithm, $T = \Tin$ and $(\Sout,\Pout) = (\emptyset,\emptyset).$
 
\begin{figure}
\centerline{\includegraphics[width=2.5in]{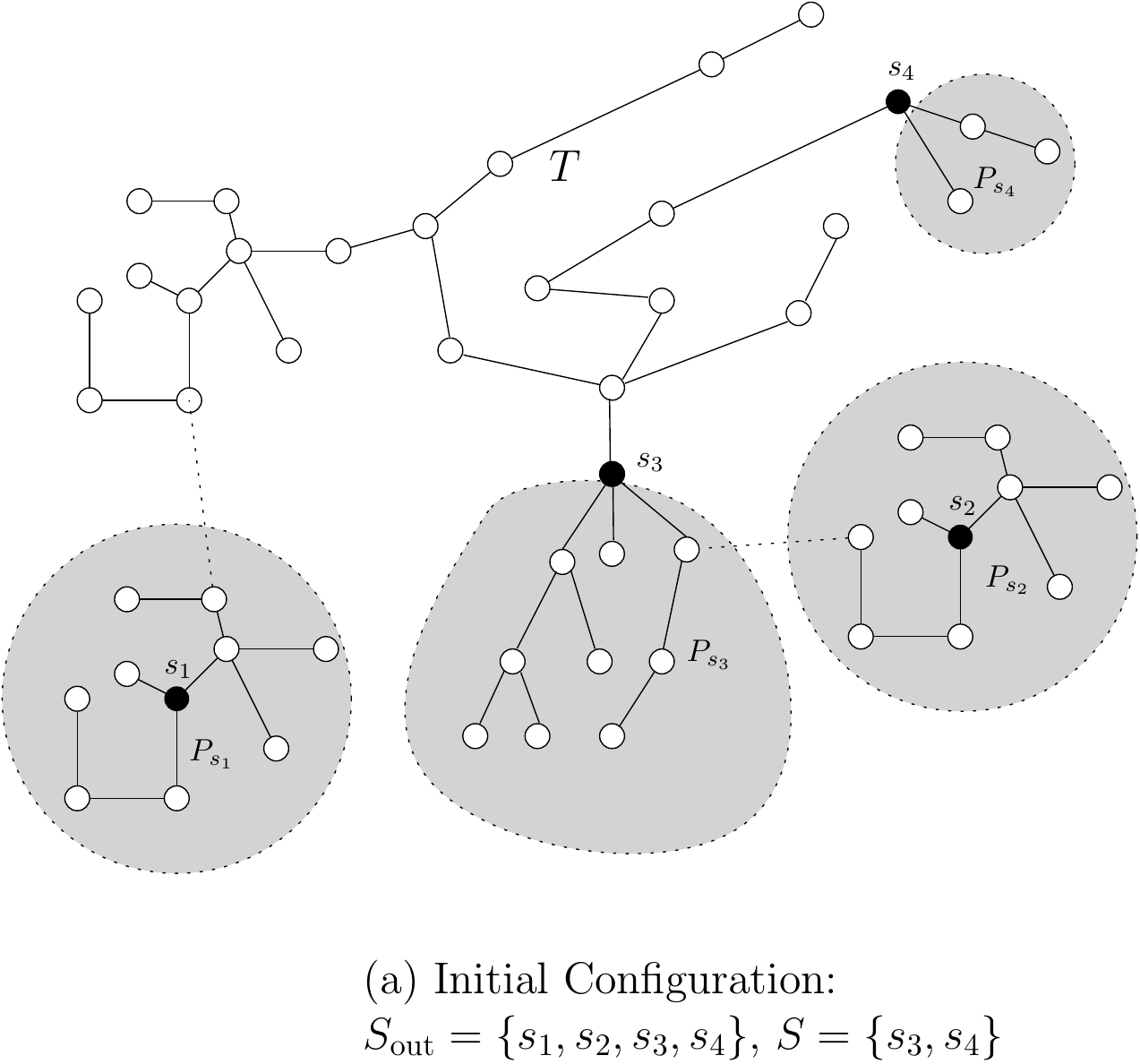} \hspace*{.3in}
\includegraphics[width=2.5in]{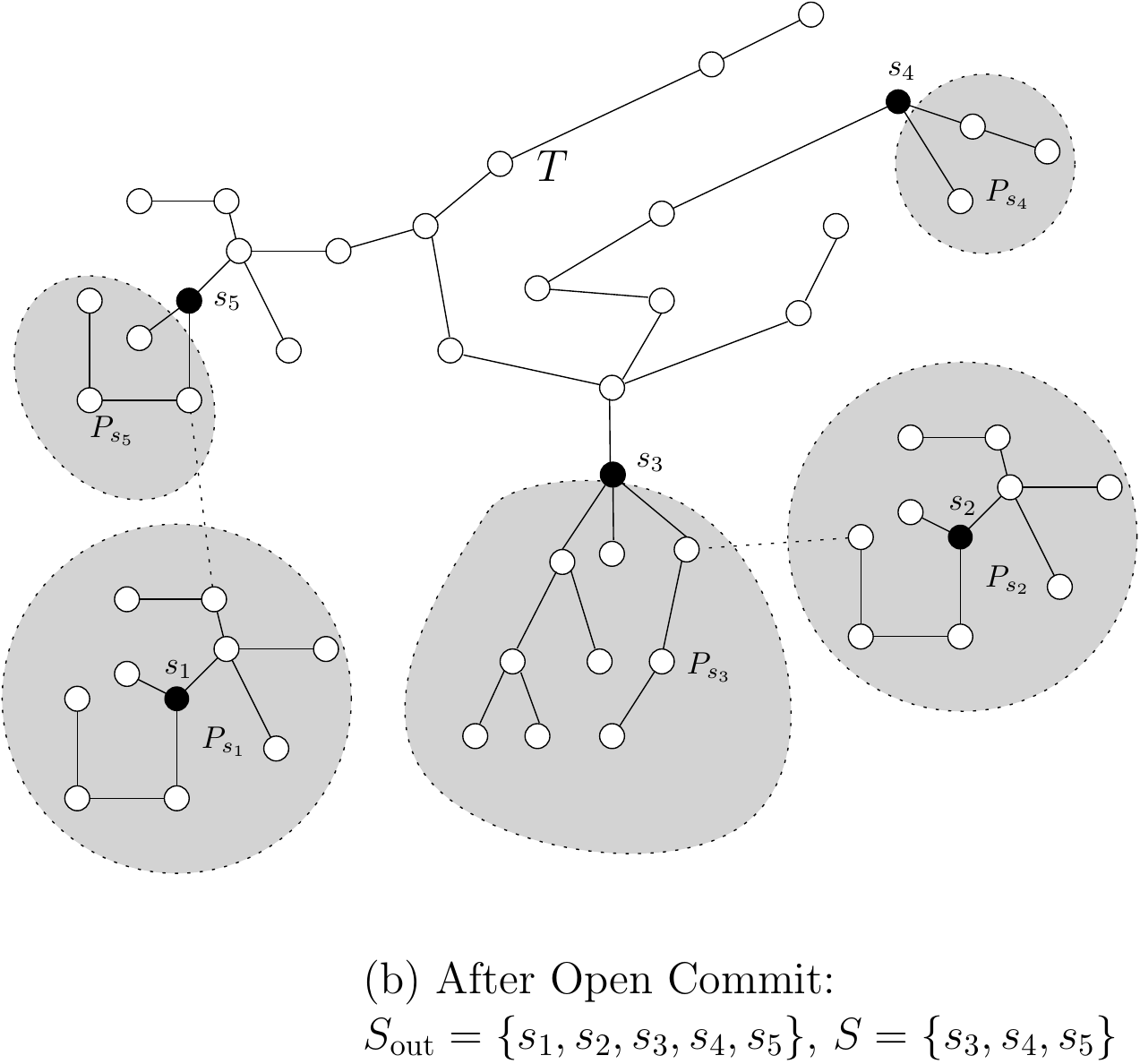}}
\vspace*{.3in}

\centerline{\includegraphics[width=2.5in]{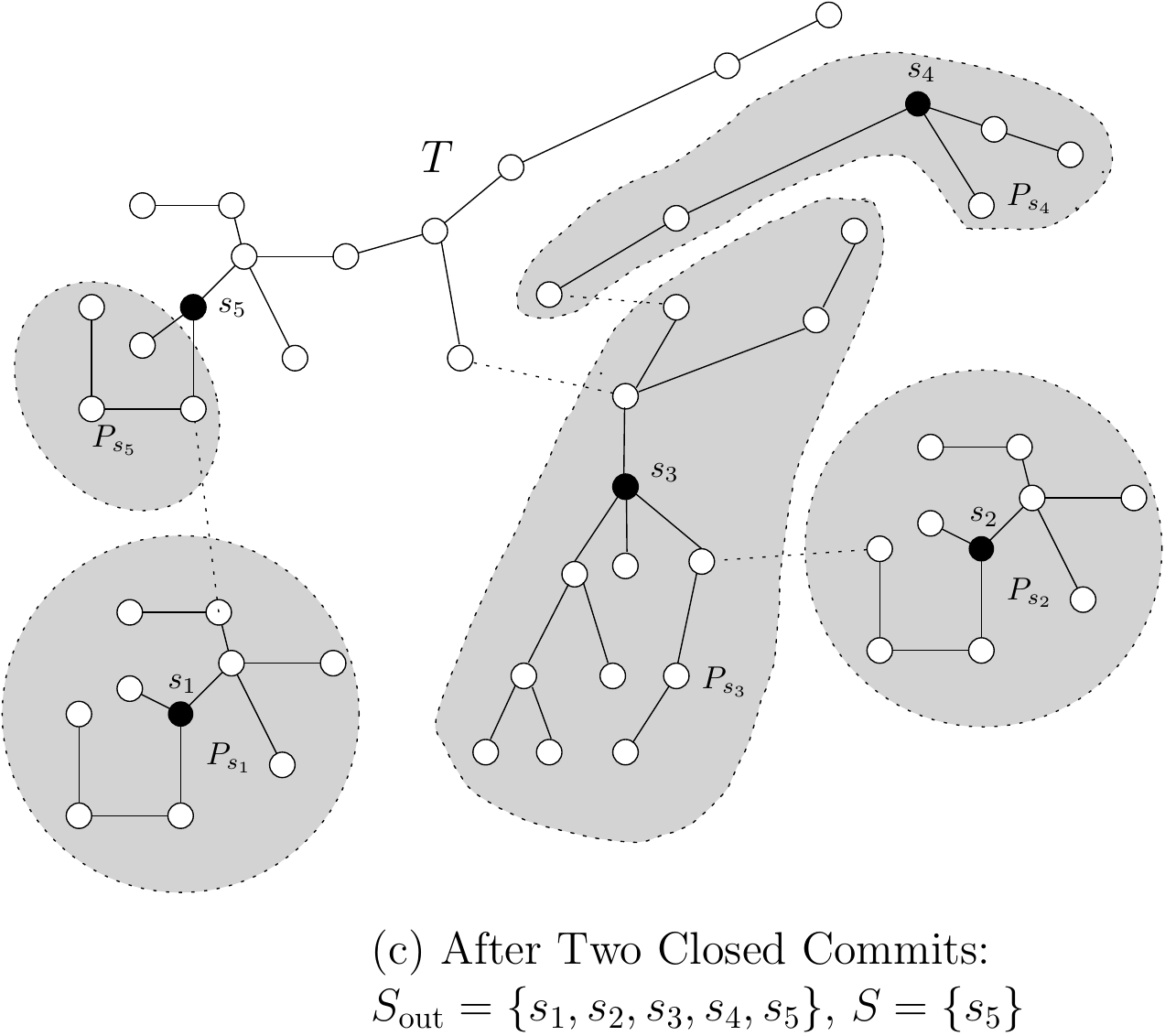}}
\caption{ (a) is a partial sink current configuration.  Black nodes are sinks.  Gray areas are the $P_s$   associated with those sinks.  
Note that the  ``open'' sinks  $S=\{s_3,s_4\}$ are leaves of the working tree $T.$  
(b) is the result of  an  open commit creating  new sink $s_5$ and its associated $P_{s_5}.$  (c) results from  two closed commits performed on  $s_3$ and then $s_4.$  Note that $s_3$ and $s_4$ are now closed and will never have any further nodes committed to them.}
\label{fig:commits}
\end{figure}

At each step, the algorithm will   {\em commit}  a subtree block $\Pnew \subseteq V$ of previously unserviced nodes  to a sink $s$. 
There will be two types of commits, 
(Fig.~\ref{fig:commits}) with the following properties:
 \begin{itemize}
 \item  {\em Open commit:}  of  $\Pnew \subseteq V$ to {\bf new}  sink $s \in V \setminus  \Sout$.   $s \in \Pnew$
 \begin{itemize}
 \item   $s$ will  be added to $\Sout$.  
\item $P_s = P_\mathrm{new}$   will  be added 
  to $\mathcal{P}_{\mathrm{out}}$.  %$s \in P_\mathrm{new}$.
  \item $\Pnew \setminus \{s\}$ is removed from working tree $T$ which remains  a tree.
  \item $s$ becomes  a leaf  of $T$.
 \end{itemize}
\medskip
 \item {\em Closed commit:} of $P_\mathrm{new}  \subseteq V$  to {\bf existing}  $s \in \Sout$ which is a leaf of $T$
 	\begin{itemize} 
 %		\item  $P_s \in \Sout$ already exists.
                 \item If $\Pnew \not=\emptyset$, it contains unique  neighbor of $s$ in $V.$
 		\item  $P_\mathrm{new}$ is merged into $P_s$  and $s$ will be closed;\\ no new blocks will henceforth be added to $P_s.$  
 		\item $\Pnew \cup \{s\}$ is removed from  $T$, which will remain a tree.
 %		\item  {\em  Closed} commits will be made by  Reaching subroutine in Section \ref{subsec: Subroutine: Dual Peaking Criterion}.
 		\end{itemize}
 \end{itemize}
 
Algorithm~\ref{alg:CommitBlock} encapsulates the above.

Later subsections will define the  Peaking  (Section \ref {subsec: Subroutine: Peaking Criterion})  and Reaching  (Section \ref{subsec: Subroutine: Dual Peaking Criterion})  subroutines that, respectively, implement {\em Open} and {\em  Closed} commits.

Set $S = \Sout \cap V$ to be the current sinks in the working tree $T=(V,E)$.
By construction, the sinks in  $S$ will all  be leaves of $T.$

\begin{algorithm}[h]
	\begin{algorithmic}[1]
		\State Given $\Pout, \Sout$
		
		\Procedure{Commit}{$P_\mathrm{new} \subseteq V_{\mathrm{in}},s$}
		\If {$s \in \Sout$}
		 \State{$P_s := \Pnew \cup P_s$}
		 \State {Remove $\Pnew \setminus\{s\}$ from Working Tree $T.$}
		\Comment {Closed Commit}
		\Else
		\State{$\Sout := \Sout \cup \{s\}$}
		\Comment{Open Commit}
		\State {Create $P_s := \Pnew$ and add $P_s$ to $\Pout$}
		\State{Remove $\Pnew \cup \{s\}$ from Working Tree $T$.}
		\EndIf
		\EndProcedure

	\end{algorithmic}
	\caption{Committing block $\Pnew$ to $s$}
	\label{alg:CommitBlock}
\end{algorithm}

The  final algorithm will maintain {\em optimality} of  $(\Sout,$  $\mathcal{P}_{\mathrm{out}})$. Informally this means that $(\Sout,$  $\mathcal{P}_{\mathrm{out}})$ can be completed to an optimal  $(S^*,\mathcal{P}^*)$. Formally  

\begin{definition}
\label{def:ropt}
A partial sink configuration $(\Sout,\Pout)$  is {\em optimal} relative to Working Tree $T=(V,E)$  % (and  $S= \Sout \cap V$ )
if 
\begin{enumerate}[label=(\textrm{C}\arabic{*})]
 \item There exists some optimal feasible sink configuration   $(S^*,\mathcal{P}^*)$ satisfying:
\item  $\Sout \subseteq S^*$ and
	\begin{enumerate}
	\item  $S^*\setminus V = \Sout \setminus V$
	\item $S = \Sout \cap V$ are leaves of $T$
	\item $S^* \setminus\Sout \subseteq V$  \ {\small \em (follows from (a) and (b))}
	\end{enumerate}
 \item  Let $P^*_s$ be the partition block in $\mathcal{P}^*$ associated with $s \in S^*.$  
\begin{enumerate}
 	\item If $s \in S^*\setminus V$  then $P_s = P^*_s$ and $P^*_s \cap V = \emptyset$
      \item    If $s \in S$, then  $P_s \subseteq P^*_s$  and $P^*_s \setminus P_s \subseteq V$
      % $s \in S  =\Sout \cap V$, then  $P_s \subseteq P^*_s$  and $P^*_s \setminus P_s \subseteq V$
      \item If $s \in S^* \setminus\Sout$ then $P^*_s \subseteq V.$
      % $s \in S^* \setminus\Sout \subseteq V  =  S^* \cap (V \setminus \Sout)  $ then $P_s \subseteq V.$
      \end{enumerate}
\end{enumerate}
\end{definition}

\medskip
Finally, suppose that $(\Sout,\Pout)$ is an optimal partial  sink configuration relative to $T.$ 
 Let $k^* = |S^*|$ and  $j = |\Sout|$. For later use we note that from  (C2) 
%Since $S = \Sout \cap V$,  $\Sout \subseteq S^*$  and $\Sout \setminus V = S^* \setminus V$ 
 the nodes in $S^*$ can be ordered as follows
 \begin{equation}
 \label{eq:S*order}
 S^*= \{    \overbrace{ s_1,\,\ldots,\,s_i}^{\Sout\setminus V},\,   \overbrace{s_{i+1},\,\ldots,\,s_j}^{S= \Sout\cap V},\, \overbrace{s_{j+1},\,\ldots,\,s_{k^*}}^{S^*\setminus\Sout \subseteq V}\}.
 \end{equation}

The intuition is that a closed  commit will move a  sink from  $S$    to $\Sout\setminus V$ and an open commit will move a sink from  $S^*\setminus \Sout$  to $S$. %   both  operations  thus maintaining the optimality of
 $(\Sout,\Pout)$.

%The definitions below will be needed.  They are both relative to the current $(\Sout,\Pout)$ and $T$.
The definitions below are both relative to the current $(\Sout,\Pout)$ and $T$.
\begin{definition}[Self-sufficiency] Fig.~\ref {fig:SS}.\\
\label{def:self-sufficiency}
A subtree $T' = (V',E')$ of $T$  is {\em self-sufficient} if $V'$ can be served by $S'=\Sout \cap V'$. 

A {\em partition of $T'$ induced by its self-sufficiency} is a partition of $V'$ into  blocks $P_s$,  $s \in S'$,  such that $s \in P_s$,  $P_s$ is a subtree of $T'$   and   $f(P_s,s) \leq \mathcal{T}.$
\end{definition}

\begin{figure}[t]
	\centering
	\includegraphics[width=0.5\textwidth]{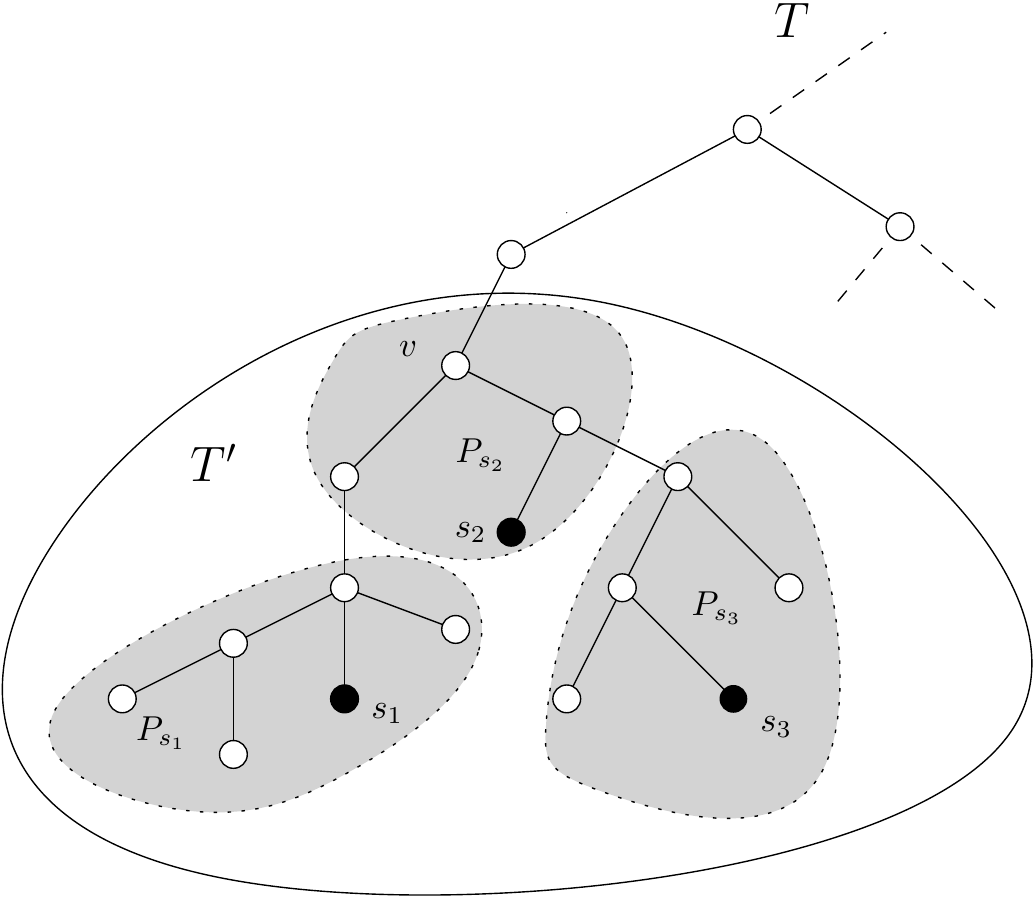}
	\caption{Self-Sufficiency. $T'=(V',E')$ is the tree ``below'' $v.$  It contains sinks  
	$S'=\Sout \cap V =\{s_1,s_2,s_3\}$. 
%If  $f(P_1,s_1) \le \mathcal{T},$ $f(P_2,s_2) \le \mathcal{T}$ and $f(P_3,s_3) \le \mathcal{T}$ 
If  for $i=1,2,3,$ $f(P_i,s_i) \le \mathcal{T},$ 
 then $T'$ is self-sufficient and $P_1,P_2,P_3$ is a partition of $T'$ induced by its self-sufficiency.}
\label{fig:SS}
\end{figure}

\begin{definition}[$T_{-v}(u)$]   Fig.~\ref{fig:Peaking}.\ \\
Let $v \in V$ be an internal node of tree   $T = (V,E)$  and $u \in V$ be  a neighbor of $v.$
Removing $v$ from $T$ creates a forest $\mathcal{F}_{-v}$ of disjoint subtrees of $T$.

$T_{-v}(u) = (V_{-v}(u), E_{-v}(u))$ denotes the unique subtree  $T' = (V',E') \in \mathcal{F}_{-v}$  such that $u \in V'$.

\end{definition}

The removal of edge $(u,v)$ splits $T$  into $T_{-v}(u)$  and $T_{-u}(v).$
The blocks greedily committed by the algorithm will all be  self-sufficient subtrees in these forms. % form  $T_{-v}(u)$. 

\subsection{Subroutine: Peaking Criterion}
\label{subsec: Subroutine: Peaking Criterion}

\begin{figure}[t]
	\centering
	\includegraphics[width=0.7\textwidth]{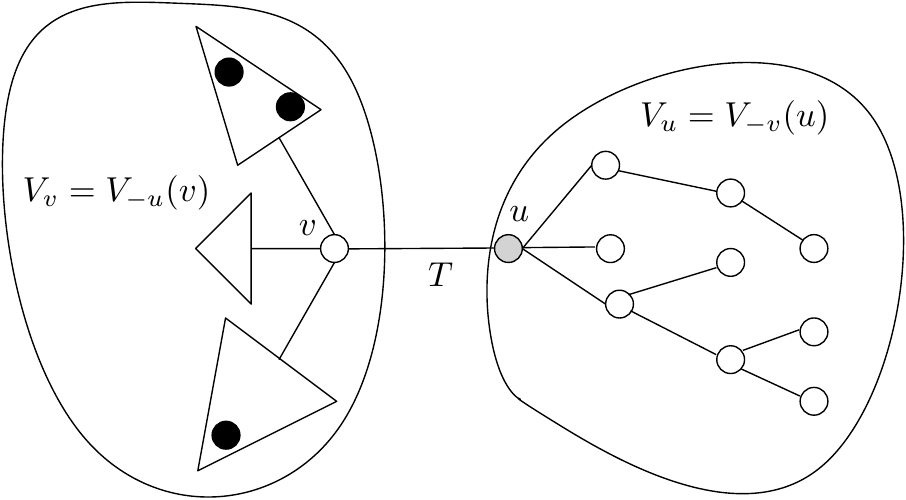}
	\caption{Peaking criterion. Note that $V$ is partitioned  into  $V_v= V_{-v}(u)$  and $V_u=V_{-v}(u)$.  $V_u$ originally contains no sinks while  $V_v$ does (the black nodes).  If $f(V_u,u) \leq \mathcal{T}$, then $u$ can serve $V_u$, so no sink is needed  below $u$;  if $f(V_u\cup \{ v \},v) > \mathcal{T}$, then no node in $V_v$   can singlehandedly support $V_u$. This pinpoints the position of exactly one sink to be placed at $u$.}
\label{fig:Peaking}
\end{figure}

 The definition and lemmas below  will justify a mechanism for  greedily performing open commits.  $T=(V,E)$ will always be the current working tree, $(\Sout,\Pout)$ will always be an optimal partial sink configuration relative to $T$ and $S = \Sout \cap V.$

\begin{definition}[Peaking criterion] 
The ordered pair of points $(u,v) \in V \times V$ satisfies the {\em peaking criterion} (abbreviated {\em PC}) if and only if (Fig.~\ref{fig:Peaking})
 \begin{itemize}
 \item  $(u,v) \in E$,
 \item $T_{-v}(u)$ contains  no sink in $S$, and 
 \item $f\left(V_{-v}(u), u\right) \leq \mathcal{T}$ but $f(V_{-v}(u) \cup \{ v \} , v) > \mathcal{T}$.
 \end{itemize}
 %\end{enumerate}
\end{definition}

\begin{lemma}[Peaking Lemma]
 Let $(u,v)$ satisfy the peaking criterion. 
%Then performing $Commit(V_{-v}(u),\,u)$ using Algorithm \ref{alg:CommitBlock}  maintains $(\Sout,\Pout)$ as an optimal partial sink configuration.
Then adding $u$ to $\Sout$ and committing $V_{-v}(u)$ to sink $u$ using Algorithm \ref{alg:CommitBlock}  maintains $(\Sout,\Pout)$ as an optimal partial sink configuration.
 \label{lemma:PeakingCriterionPutSink}
\end{lemma}
{\em Note:  This is exactly an ``open commit'' as defined in Section \ref{subsection: greedy construction} (Fig.~\ref{fig:commits}(b).}

\begin{proof}
Let $(S^*, {\mathcal P}^*)$ be the feasible sink configuration given by (C1)-(C3). Set $k^* = |S^*|$, $j = |\Sout|$.  Recall from (\ref{eq:S*order}) that $S^*$ can written as 
$$ S^* = \{    \overbrace{ s_1,\,\ldots,\,s_i}^{\Sout\setminus V},\,   \overbrace{s_{i+1},\,\ldots,\,s_j}^{S= \Sout\cap V},\, \overbrace{s_{j+1},\,\ldots,\,s_{k^*}}^{ S^*\setminus\Sout \subseteq V}\}.$$

%Since $S = \Sout \cap V$,  $\Sout \subseteq S^*$  and $\Sout \setminus V = S^* \setminus V$ 
 
 %(C3)(a)  implies that all nodes in $V_{-v}(u)$ must be served by sinks in $S^* \cap V$.
 Recall  too (Fig. \ref{fig:Peaking}) that $V$ can be partitioned into 	$V_u =V_{-v}(u)$ and $V_v = V_{-u}(v).$  For all $u'\in V$, let $s(u') \in S^*$ denote the unique sink  such that 
 $u' \in P^*_{s(u')}.$ From (C3),  $\forall u' \in V$,  $s(u')  \in V$ and  $\Pi(u',s(u'))$  lies in $T$ as well.
 
 Note the following properties with their  justifications
 \begin{itemize}
 \item[P1] If $u' \in V_u$ and $s(u') \in V_v$, then $s(u') = s(u) = s(v).$\\
 Because the path  $\Pi(s(u'),u')$  passes through $v$ and then $u.$
  \item[P2] If $v' \in V_v$ and $s(v') \in V_u$, then $s(v') = s(v) = s(u).$\\
 Because the path $\Pi(s(v'),v')$ passes through $u$ and then $v.$
 \item[P3] No $s \in S^* \cap V_v$ can support $V_u$.\\
 Otherwise,  from P2, $s= s(v).$  Path monotonicity then implies that $v$ can support $V_u$,  contradicting that  $(u,v)$ satisfies the  peaking criterion. 
 \item[P4] $S \cap V_u = \emptyset$.\\
 Follows directly from  $(u,v)$ satisfying the peaking criterion.
 %\item[P5] If $s \in S^*\setminus\Sout$ then $P^*_{s} \subseteq V.$\\
 %First recall that $ S^*\setminus \Sout \subseteq V.$  P5 then follows from (C3)(c).
 \item[P5] $S^*$ must contain at least one sink   $s \in (\Sout\setminus S) \cap V_u.$\\
 Follows directly from P3  and P4 and fact that $\forall u' \in V$,  $s(u')  \in V$.
 \end{itemize}

%
%
%Suppose there exists some sink  $s \in S^*$ and $s \in V\setminus V_{-v}(u)$ that supported  some node $u' \in V_{-v}(u).$  

From P5, $\left(S^*\setminus \Sout\right) \cap V_u \not=\emptyset.$  
Without loss of generality assume that $ \left(S^*\setminus \Sout\right) \cap V_u = \{s_{j+1}, \ldots, s_{r}\}$ and set 
$$ V'= V_u \cup \left(\bigcup_{\ell=j+1}^r P^*_{\ell}\right)
=  V_u  \cup V'' 
\quad\mbox{where}\quad 
V'' =  \left(\left(\bigcup_{\ell=j+1}^r P^*_{\ell}\right)\setminus V_u\right) \cup \{u\}.$$
Because $(u,v)$ satisfies the peaking criterion,  $f(V_u,u) \le \mathcal{T}$.   We claim that $f(V'' \cup \{u\}, u) \le  \mathcal{T}$  as well and thus $f(V', u) \le  \mathcal{T}$.
There are two possible cases.

\smallskip

\par\noindent\underline 
{Case (i): $s(u) \in V_v.$}\\
Suppose for some $v' \in V_v,$  $s(v') \in V_u.$  Then, from P2, $s(u) = s(v')$, contradicting the assumption.  Thus, 
 for all $j < \ell \le r$, $P^*_{s_\ell} \subseteq V_u$.
Then  $V'' = \{u\}$ so $V' = V_u$ and  $f(V',u) \le \mathcal{T}$

\medskip

\par\noindent\underline 
{Case (ii): $s(u) \in V_u.$}\\
Since $s(u) \not\in \Sout,$ WLOG assume $s(u) = s_{j+1}.$

From P2 and C3(c),  if $s \in V_u$ and $s \not=s_{j+1}$ then  $P^*_{s_{\ell}} \subseteq V_u.$
%Let   $j < \ell \le r.$  Then for some value of $\ell$ distinguished as $\ell'$,  $s(u) = s_{\ell'}.$  From P2 and P5, for $\ell \not=\ell',$  
% $P^*_{s_{\ell}} \subseteq V_u.$
Thus
%This means that $s(u) = s_\ell$ for some  $j < \ell \le r.$ 
%If $s_{\ell'} \in V_u$ with $\ell' \not = \ell$ then P2 combined with P5 again implies that  $P^*_{s_{\ell'}} \subsetneq V_u.$  So,
$$ V'' = \left(P^*_{s_{j+1}} \cap V_v\right) \cup \{u\}.$$
By construction,  $V''$ is a tree in which  $u$ is a leaf. 
From the fact that the path from $s_{j+1}$ to any node in $V_v$ passes through $u$ and path monotonicity,  
$$f(V'',u) \le f(V'' \cup \Pi(u,s_{j+1}),s_{j+1})  \le  f(P^*_{s_{j+1}},s_{j+1})\le \mathcal{T}.$$

   Thus by max-composition, again  
\begin{equation}
\label{eq:VpT}
f(V', u)  = \max\left( f(V_u,u),  f(V'', u)\right)  \le  \mathcal{T}.
\end{equation}

In both case (i) and (ii) perform an open $Commit(V_u,u)$,  removing  $V_u\setminus\{u\}$ from $T$.    Label  these new $\Sout$, $\Pout$, $T$ and $S$ as $\barSout$, $\barPout$, $\bar T= (\bar V, \bar E)$ and $\bar S.$

We can now see that  $\barSout$ and $\bar T$ maintain (C1)-(C3) with the  new configuration $(\bar S^*, \bar {\mathcal P}^*)$, where
$$ \bar S^* = \{    \overbrace{ s_1,\,\ldots,\,s_i}^{\barSout\setminus \bar V},\,   \overbrace{s_{i+1},\,\ldots,\,s_j,u}^{\bar S= \barSout\cap \bar V},\, \overbrace{s_{r+1},\,\ldots,\,s_t}^{ S^*\setminus\barSout \subseteq\bar  V}\},$$
i.e., $s_{j+1}, \ldots, s_{r}$ were removed and $u$  added and, for $s \in \bar S^*,$
$$
\bar P^*_s =
\left\{
\begin{array}{ll}
P^*_s\setminus V' & \mbox{if $s \not=u,$}\\
V'        & \mbox{if $s = u.$}\\
\end{array}
\right.
$$
By construction, $V'$ is a subtree and, if $s \not=u,$ then  $P^*_s\setminus V'$ is also a subtree.
Furthermore, also by construction,  $\forall s \in  \bar S^*,$   $f(P^*_s,s)\le \mathcal{T}.$ Thus, 
$(\bar S^*, \bar {\mathcal P}^*)$ is feasible.
 It is also  optimal because $( S^*,  {\mathcal P}^*)$ was optimal and  $|\bar S^*| \le | S^*|.$
$\bar S^*$ satisfies (C2) by construction and noting  that $u$ is now a sink of $\bar T.$ $(\bar S^*, \bar {\mathcal P}^*)$ satisfies (C3) by construction and noting that
$\bar P_u  = V_{-v}(u) \subseteq V' = \bar P^*_u.$

Thus $(\barSout,\barPout)$ is an optimal partial sink configuration relative to $\bar T.$
\qed

\end{proof}

 The algorithm will  keep attempting to add sinks by finding  edges that satisfy the peaking criterion. The next lemma, with its corollary,  exactly characterizes when no such edge exists and the process must stop.

\begin{lemma}
 Suppose  for some $u,v$, $f(V_{-v}(u) \cup \{v\}, v) > \mathcal{T}$, and $S \cap V_{-v}(u) = \emptyset$. Then there exists $u',v' \in V_{-v}(u) \cup \{v\}$ satisfying  the peaking criterion.
 \label{lemma:PCRecurseInto}
\end{lemma}
\begin{proof}%[Proof of Lemma \ref{lemma:PCRecurseInto}]
	Assume by contradiction that no such pair $u',v'$  exists.
In particular, this requires  that $(u,v)$ doesn't satisfy the peaking criterion so 
	$f(V_{-v}(u), u) > \mathcal{T}$. This implies $|V_{-v}(u)| \geq 2$, because otherwise   
	$$f(V_{-v}(u),\{u\}) = f(\{u\},u) = 0\leq \mathcal{T}.$$
	 Max composition implies there exists some neighbor of $u$, $\eta_2 \in V_{-v}(u),$ such that $f(V_{-u}(\eta_2) \cup \{u\}, u) > \mathcal{T}$. Set $\eta_0 = v, \eta_1=u$ 
	
	Applying  the above argument repeatedly with
	$\eta_{i}, \eta_{i-1}$ in the place of  $u,v$  will generate an infinite sequence of distinct nodes $\eta_0,\eta_1,\eta_2,...$ such that  $\eta_{i+1} \in V_{\eta_{i-1}}(\eta_i)$ but $f(V_{-\eta_i}(\eta_{i+1}) \cup \{\eta_i\}, \eta_i) > \mathcal{T}$, which is impossible. \qed
\end{proof}

\begin{corollary}
\label{corollary:PCStoppingCondition}
  \label{corollary:PCStoppingCondition2}
% Let $(\Sout,\Pout)$  be optimal relative to working tree $T=(V,E)$.   Recall that
 $S = \Sout \cap V.$   If no pair $(u,v)$ in $V$ satisfy the peaking criterion, then exactly one of the following three situations occurs:
 \begin{enumerate}
 \item {\bf $\bf S = \emptyset$:}\\
 Then   for all $s \in V$, $f(V,s) \leq \mathcal{T}$.
  Furthermore,   for all $s \in V,$ $Commit(V,s)$  will create an optimal feasible sink configuration $(\barSout,\barPout).$ 
 % $(\Sout,\Pout)$ being an optimal feasible configuration.
 \item {\bf  $\bf S=\{s\}$ for some $\bf s \in V$:}\\
 Then  $f(V,s) \le   \mathcal{T}$.
   Furthermore,    $Commit(V,s)$ will create an optimal feasible sink configuration $(\barSout,\barPout).$ 
%   and  if an optimal feasible configuration exists then   commiting $V$ to $s$   terminates the algorithm with
%  $(\Sout,\Pout)$ being an optimal feasible configuration.
\item  {\bf $\bf |S| = |V| =2$}\\
Let $V =S=\{s,s'\}.$  Then   $Commit(\{s\},s)$  followed by $Commit(\{s'\},s')$will create an optimal feasible sink configuration $(\barSout,\barPout).$
 \item {\bf $\bf |S| \ge 2$ and $|V| >2.$}
 \end{enumerate}
 
\end{corollary}
\begin{proof}

Let $S^* \supseteq \Sout$ be the optimal  feasible sink placement defined by (C1)-(C3).

(1) If $S = \emptyset$ choose any $s \in V$. Let $r\in V$ be any neighbor of $s$.  %Then, 
The assumption that no edge satisfies the peaking criterion combined with  Lemma \ref {lemma:PCRecurseInto} implies $f(V_{-s}(r) \cup \{s\}, s) \le \mathcal{T}$. % since otherwise there exists a pair that satisfies the peaking criterion. 
Thus, by max composition, $f(V,s) \leq \mathcal{T}$.

(C2)  and  (C3) imply  that  $V$ must be serviced by a sink in  $S^*\setminus\Sout$  so  $|S^*| \ge |\Sout| +1.$   On the other hand (C3) also implies % that $s' \in \Sout$ feasibly services $P_{s'} \in \Pout$ and  
$\cup_{s' \in \Sout} P_{s'} = \Vin \setminus V$.
% where $s' \in \Sout$ feasibly services $P_{s'} \in \Pout$.

%Thus the sinks in $S^*\setminus \Sout$ only service nodes in $V.$
Thus performing  $Commit(V,s)$  creates a sink configuration that services all nodes in $\Vin$ and has size  
$|\Sout| + 1| \le |S^*|$ and is therefore an 
optimal feasible sink configuration.

\medskip

(2) If $S=\{s\}$, then $s$ is a leaf in $T.$  Let  $r\in V$ be the unique neighbor of $s$ in $V.$  From Lemma \ref {lemma:PCRecurseInto}, $f(V_{-s}(r) \cup \{r\}, r) \le \mathcal{T}$, since otherwise there exists a pair that satisfies the peaking criterion. Since  $V_{-s}(r) \cup \{s\}= V,$ 
$f(V,s) \leq \mathcal{T}$. 

(C2) and (C3) imply that for $s' \in S^* \setminus \{s\}$, $P_s = P^*_s$ and furthermore that $\cup_{s \in \Sout} P_s = \Vin \setminus V$.  Thus performing $Commit(V,s)$  will then set $\cup_{s \in \Sout}  \bar P_s=  (\Vin \setminus V) \cup V = \Vin$.  Since this did not change $\Sout$ and $|\Sout| \le |S^*|$,  $\Sout$ must therefore be optimal.

\medskip

(3) is obvious.

\medskip

(4) If none of (1)  (2)  or (3) occur, then $|S| \ge 2$ and $|V| >2.$
\qed
\end{proof}

Our algorithm will repeatedly  place sinks using the Peaking Lemma  until no edge satisfying the peaking condition can be found.  It maintains the invariant that  $(\Sout,\Pout)$ remains optimal relative to  working tree $T.$ 

 Corollary \ref{corollary:PCStoppingCondition}    implies  that if no edge satisfying the lemma can be found and    $|S| <  2$,  or
 $|S| = |V|=2$ 
 then  the resulting sink configuration constructed  is optimal.

\subsection{The Hub Tree}
\label{subsec: hub}
The peaking lemma will be the only method of  addding {\em new} sinks to $\Sout.$
%The algorithm's can only  {\em add} sinks to $\Sout$    via the peaking lemma.
If the peaking criterion no longer hold for any pair  $(u,v)$ but $\mathcal{P}_{\mathrm{out}}$   still does not contain a full partition, another  mechanism  be needed to perform closed commits of unserviced blocks to already existing sinks.
 
  Section \ref{subsec: Subroutine: Dual Peaking Criterion}  introduces the {\em Reaching Criterion}   for this.  
 It  will first require  defining  the {\em hub tree}.

\begin{definition}[Hubs] See  Fig.  \ref{fig:HubTree}.\\
\label{def:hub tree}
Let $T=(V,E)$ be the working tree and $S$ the set of sinks in  $T.$ Recall that nodes in $S$ are   leaves of $T.$ 
Assume $|S| \ge 2$, $|V| >2$  and that $T$ is rooted at some non-sink  $r$ such that  at least two of $r$'s  children are sinks  or have  sink descendants
 \begin{itemize}
 \item Let $H(S) \subseteq V$ be the set of lowest common ancestors of all pairs of sinks in $T$. The nodes in $H(S)$ are the {\em hubs} associated with $S$. 
 
 \item The {\em hub tree}  $T_{H(S)} = (V_{H(S)}, E_{H(S)})$ is the rooted subtree of $T$ that contains all vertices and edges contained in all of  the paths 
 $\Pi(s,r)$  where $s \in S$. 

\item For $u \in V_{H(S)},$  set $\bf T(u)=(V(u),E(u))$ to be the subtree of $T$ rooted (down) at $u$.
% $\bf T_{H(S)}(u)=(V_{H(S)}(u),E_{H(S)}(u))$ to be the directed subtree
\end{itemize}
\end{definition}

\begin{definition}[Outstanding branches]
 A  node $w \in V_{H(S)}$ {\em branches out} to $\eta$ if $\eta$ is a neighbor of $w$ in $T$ that does not exist in $V_{H(S)}$. The subtree $T' := T_{-w}(\eta)$ is called an {\em outstanding branch}; we say that $T'$ is {\em attached} to $w$.
\end{definition}
From the definition of the hub tree,  outstanding branches contain no sinks.
\begin{figure}
	\centering
	\includegraphics[width=0.5\textwidth]{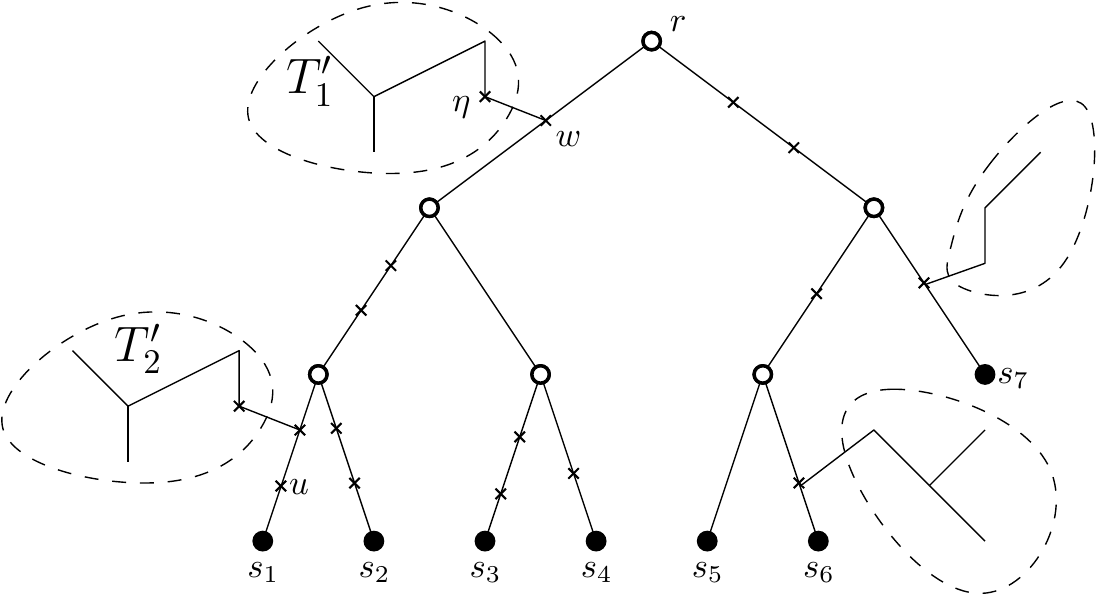}
	\caption{Visualization of hub tree $T_H(S)$ with root $r.$  Areas enclosed by dashed lines are outstanding branches (and are not in the hub tree),  filled  circles denote sinks, and unfilled circles denote  hubs. $BP(u,w)$ is the union of the path $\Pi(u,w)$ in the hub tree and the two outstanding branches $T'_1$ and $T'_2$ along with the edges connecting $T'_1,T'_2$ to the path. $T(w)$ is the tree rooted at $w$ and includes everything below it, including $T'_1,T'_2$ and sinks $s_1$-$s_4$.}
		\label{fig:HubTree}
\end{figure}

\begin{definition}[Bulk path]
 Given  $u,v \in V_{H(S)}$, the {\em bulk path} $\mathrm{BP}(u,v)$ is  the union of path $\Pi(u,v)$   with all the nodes in all outstanding branches that are attached to any node in $\Pi(u,v)$. $\mathrm{BP}(v,v)$ denotes the special case of the union of $v$ and of all the outstanding branches falling off of $v.$
\end{definition}

%We can now define  a useful property that will characterize $T$ when  the peaking criterion is inapplicable.
We can now describe what occurs when  the peaking criterion is inapplicable.

\begin{definition}[RC-viable]
\label{def:RCV}

$T$ is {\em RC-viable}  (with respect to $S$) if
for every $T' = (V',E')$ that  is an outstanding branch attached to $w \in V_{H(S)}$, $f(V' \cup \{ w \},w) \leq \mathcal{T}$.
\end{definition}

\begin{lemma}
%Let $T$, $S$ be %the rooted  working tree and sinks
 %as described  in Definition \ref{def:hub tree}. 
If no ordered pair $(u,v) \in V \times V$ satisfies the peaking criterion then $T$ is {RC-viable} (with respect to $S$).
 \label{lemma:RCViability}
\end{lemma}
\begin{proof}%[Proof of 	Lemma \ref{lemma:RCViability}]
	Let $T' := (V',E') = T_{-w}(\eta)$   be an arbitrary outstanding branch attached to some node $w \in V_{H(S)}$. If  $f(V' \cup \{ w \},w)  >  \mathcal{T}$ then from Lemma \ref{lemma:PCRecurseInto}, $T'$ would contain an ordered  pair $(u,v)$ satisfying the peaking criterion.  Since no such pair exists, $f(V' \cup \{ w \},w)  \le  \mathcal{T}$.  This is true for all outstanding branches and thus  $T$ is RC-viable.
	\qed
\end{proof}

\begin{lemma}
\label{lemma:RCV no branch sink}
 Let $(\Sout,\Pout)$  be optimal relative to working tree $T=(V,E)$ and    $T$ is  {\em  RC-viable.}  Then in  the optimal feasible sink configuration 
$(S^*,\mathcal{P}^*)$ referenced in Definition \ref {def:ropt} conditions (C1)-(C3), $S^*$   can be assumed not to contain any 
sink in an outstanding branch.
\end{lemma}
\begin{proof}
Suppose   $S^*$ did contain a sink $s$ located in an outstanding branch of $T$ attached to some  $w \in V_{H(S)}.$   Then RC-viability implies that $s$  could be moved to $w$, not increasing the size of $S^*$,  while maintaining the  feasibility of the sink configuration  (this might require modifying $\mathcal{P}^*$) and the validity of (C1)-(C3). \qed
\end{proof}

This last lemma permits assuming that  all sinks in  $S^*\setminus\Sout$  are in  $V_{H(S)}$.

\subsection{Subroutine: Reaching Criterion}
\label{subsec: Subroutine: Dual Peaking Criterion}

The definition and lemmas below  will justify a mechanism for  greedily performing closed  commits.

\begin{definition}
A   node $v \in V_{H(S)}$ \emph{can evacuate} to $s \in S$ if $f(\mathrm{BP}(v,s),s) \leq \mathcal{T}$.
\end{definition}

\begin{definition}[Reaching criterion]
%Let $T$, $S$ be % the rooted  working tree and sinks  as described.
 Let $T$ be RC-viable with respect to $S$ and  $(u,v) \in V_{H(S)} \times V_{H(S)}$ be an ordered pair of nodes.
 Then $(u,v)$ satisfies the {\em Reaching Criterion (RC)} if and only if (Fig.~\ref{fig:DualPeakingx})
\begin{itemize}
\item  $v$ is the parent of $u$ in $T_{H(S)}$,
\item  $T_{-v}(u)$ is self-sufficient and 
\item  $\mathrm{BP}(v,v) \cup V_{-v}(u)$ is not self-sufficient.
\end{itemize} 
\end{definition}

\begin{figure}[t]
	\centering
	\includegraphics[width=0.5\textwidth]{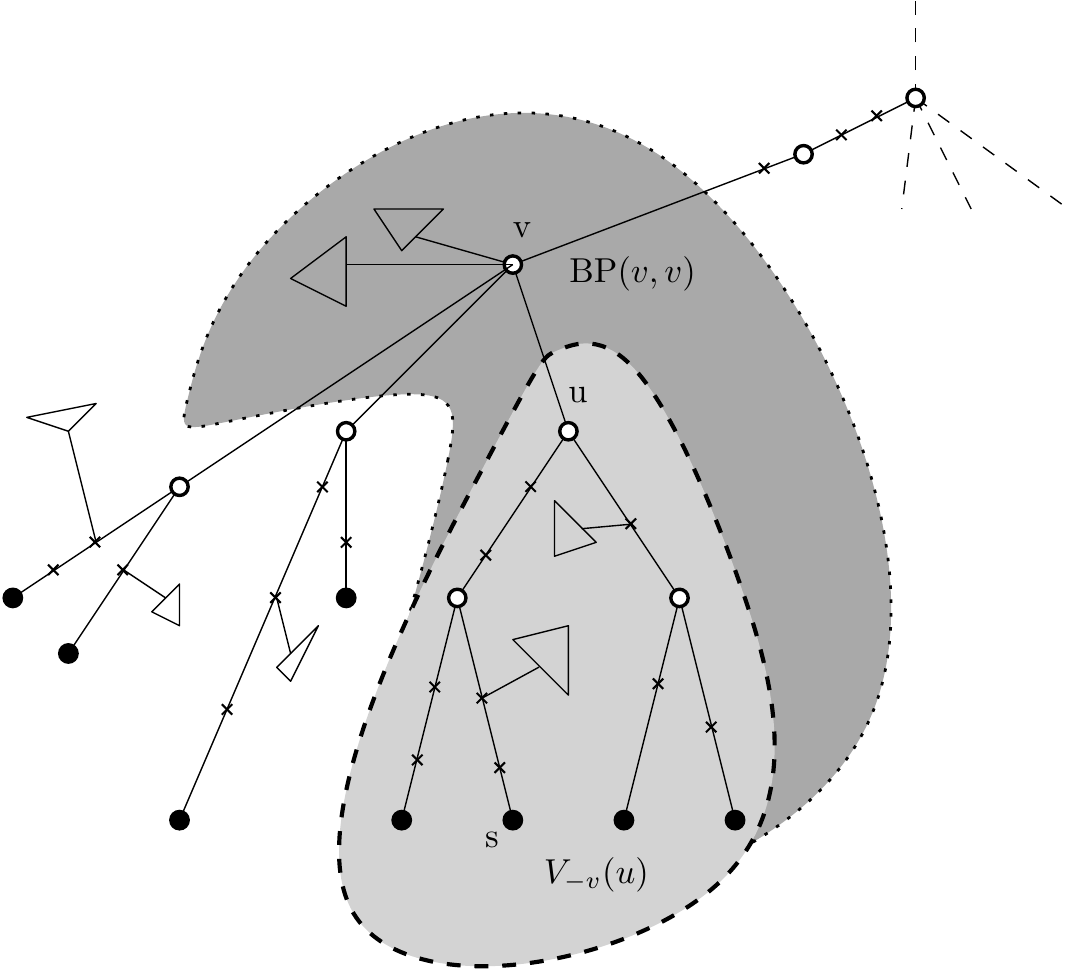}
	\label{fig:DualPeaking}
	\caption{Reaching criterion.  Filled circles are sinks, unfilled  circles are hubs, triangles are outstanding branches.  
$V_{-v}(u)$  is the light gray area;  $BP(v,v)$ is the dark gray area; it contains $v$ and the outstanding branches falling off of it.
If $T$ is  RC-viable,  sinks never need to  be placed in outstanding branches. If $T_{-v}(u)$ is self-sufficient, then  $T_{-v}(u)$ can be evacuated  to the sinks in $V_{-v}(u)$.   If $BP(v,v) \cup V_{-v}(u)$ is not self-sufficient, RC-viability implies that $v$  will not be served by  any sink $s \in V_{-v}(u)$ because assigning $v$ to sink $s \in T_{-v}(u)$ would  force all nodes in $BP(v,v)$ to also be assigned to $s$. This would be  infeasible unless a new sink was placed in $T_{-v}(u)$.  But this new sink should better  be placed at $v$, because $v$ could also serve any outstanding branch attached to it (due to RC-viability), obviating  the need to assign $v$ to $s$.}
	\label{fig:DualPeakingx}
\end{figure}

\begin{lemma}[Reaching Lemma]
%Let $(\Sout,\Pout)$  be optimal relative to working tree $T=(V,E)$.   Recall that
 %$S = \Sout \cap V.$  
%Let $\Sout,\Pout$ and $T=(V,E)$ satisfy invariants (C1)-(C3).
% Let $T$, $S$ be % the rooted  working tree and sinks 
%as described. 
Let $T$ be  RC-viable with respect to $S$ and  $(u,v)$  satisfies the reaching criterion.

 Partition $T_{-v}(u)$  into the subtrees  induced by by the corresponding sinks in $V_{-v}(u) \cap S$ as  implied by the self-sufficiency of $T_{-v}(u)$ (Definition \ref{def:self-sufficiency}). Using Algorithm \ref{alg:CommitBlock}  commit those subtrees to the blocks associated with those sinks. 
 
 This process   maintains $(\Sout,\Pout)$ as an optimal partial sink configuration.
\label{lemma:RCTrimTree}
\label{theorem:RCTrimTree}
\end{lemma}
{\em 
 Note:  The commits performed above   are  {\em closed commits} on leaf sinks in  $T_{-v}(u)$ as defined in Section \ref {subsection: greedy construction}.  Fig.~\ref{fig:commits}(c).}
 
\begin{proof}%[Proof of Theorem \ref{theorem:RCTrimTree}]

As in the proof of Lemma \ref{lemma:PeakingCriterionPutSink} simplify by writing  	$V_u =V_{-v}(u)$ and $V_v = V_{-u}(v)$ and recalling  that $V = V_u \cup V_v.$  Furthermore,  for every node $u\ \in V$,  set $s(u) \in S^*$ to denote the unique sink  such that 
 $u \in P^*_{s(u)}.$

Let $(S^*, {\mathcal P}^*)$ be a feasible sink configuration given by (C1)-(C3), $k^* = |S^*|$ and $j = |\Sout|$.  Recall from (\ref{eq:S*order}) that $S^*$ can be written as 
$$ S^* = \{    \overbrace{ s_1,\,\ldots,\,s_i}^{\Sout\setminus V},\,   \overbrace{s_{i+1},\,\ldots,\,s_j}^{S= \Sout\cap V},\, \overbrace{s_{j+1},\,\ldots,\,s_{k^*}}^{S^*\setminus\Sout \subseteq V}\}.$$
Without loss of generality assume that  $S \cap V_u =\{s_{i+1},\,\ldots,\,s_r\}$ with $r \le j.$
For  $i <  \ell \le r$ let $\Pnewi {s_{\ell}}$ be the nodes in $V_u$ that evacuate to $s_{\ell}$ as implied by the self-sufficiency of $T_{-v}(u).$  By definition, 
$\cup_{i < \ell \le r}\left(\Pnewi {s_{\ell}}  \cup  \{s_\ell\}\right)= V_u$ and, 
$\forall \ell,$  $f(\Pnewi {s_{\ell}},s_{\ell}) \le \mathcal{T}$.

 We claim  that $(S^*, {\mathcal P}^*)$  satisfies the 3 properties below (or if it doesn't, it can be replaced by a  new $(S^*, {\mathcal P}^*)$ that does).

\medskip

\par\noindent\underline{Property 1:}   $s(v) \not \in V_u.$\\
 % {\bf  If $v \in P^*_s$, then  $s \not\in V_{-v}(u):$}\\
Suppose  $s(v) \in V_u.$   The non self-sufficiency of  $BP(V,u) \cup V_{-v}(u)$   implies that $S^* \cap V_u$  must then include some sink $s' \not \in S$ and (C3)(c) implies that  $P_{s'} \subseteq V.$   In $S^*$, replace $s'$ with $v$ and modify $\mathcal{P}^*$ as follows:
$$
P^*_{s} :=
\left\{
\begin{array}{ll}
\Bigl(P^*_{s} \setminus V_u \Bigr) \cup \Pnewi {s} \cup \{s_\ell\} \ \ & \mbox{if  $s = s_\ell$ where $i < \ell \le r$}\\[0.05in]
P^*_{s'} \setminus V_u & \mbox{if $s=v$} \\[0.05in]
 P^*_{s}      & \mbox{Otherwise}\\
\end{array}
\right.
$$
By construction 
 this modified $(S^*,\mathcal{P}^*)$  is also an optimal feasible sink 
configuration, maintains (C1)-(C3) and has $s(v) =v \in V_v.$

\medskip
\par\noindent\underline{Property 2:} 
{\bf If $v' \in V_v$ then $s(v') \in V_v.$}\\
%{\bf If $u \in V \setminus V_{-v}(u)$ evacuates to sink $s\in S^*$, then $s \in V \setminus V_{-v}(u)$:}
 If $s(v') \in V_v$, then 
$v \in \Pi(v',s(v'))$  so $s(v) = s(v') \in V_v$,  contradicting Property 1.

\medskip
\par\noindent\underline{Property 3:} {\bf If $u' \in V_u$ and   $s(u') \in V_v$ then $s(u') = s(v)$}.\\
%If $s(u') \in V_v$  then $v \in \Pi(u',s(u'))$   $ s(v) =s(u') \in V_v$,  contradicting Property 1.
 If $s(u') \in V_v$  then $v \in \Pi(u',s(u'))$   so $ s(u') =s(v)$.

 \medskip

We now prove the Lemma by creating a new  optimal feasible sink configuration $(\bar S^*, \bar {\mathcal P}^*)$  for which (C1)-(C3) will be correct.  First set $\bar S^* := S^*$. Then (possibly) reallocate the nodes in $V_u$ by creating a new 
$\bar {\mathcal P}^*$ as follows: 
$$
\bar P^*_{s_\ell} :=
\left\{
\begin{array}{ll}
\Bigl(P^*_{s_\ell} \setminus V_u \Bigr) \cup \Pnewi {s_{\ell}} \cup \{s_\ell\} \ \  & \mbox{if $i < \ell \le r$}\\[0.05in]
P^*_{s_\ell} \setminus V_u   & \mbox{if $s_\ell = s(v)$}\\[0.05in]
 P^*_{s_\ell}      & \mbox{Otherwise}\\
\end{array}
\right.
$$

From Property 1, if $s_\ell = s(v)$ then  $\ell  \not\in [i+1,r]$ so this formula is well-defined.  Also from Property 1,  $P^*_{s(v)}\setminus V_u$ is a tree (since the nodes in $V_u$ are removed from $P^*_{s(v)}$ by removing the single edge $(u,v)$).

Now consider   $i < \ell \le r.$  From Property 2,  $P^*_{s_\ell} $ contains no nodes in $V_v$  so $\Bigl(P^*_{s_\ell} \setminus V_u\Bigr)\cup\{s_\ell\}$ is a tree.  Since, by construction, each 
$\Pnewi {s_{\ell}} \cup \{s_\ell\}$ is a tree, each $\Bigl(P^*_{s_\ell} \setminus V_u \Bigr) \cup \Pnewi {s_{\ell}}\cup\{s_\ell\} $ is also a subtree.
  Thus, all  of the $\bar P^*_{s_\ell}$ are trees.

The only nodes that are reallocated in the move from  the  $P^*_{s_\ell}$ to the  $\bar P^*_{s_\ell}$ are nodes in  $V_u$. 
From Property 3 the only sinks that serve nodes in $V_u$ are the ones in $V_u$  and, possibly,  $s(v)$.  Since  $\cup_{i < \ell \le r} \Pnewi {s_{\ell}}= V_u \setminus S$
the $\bar P^*_{s_\ell}$ then form a legal partition.
Thus  $(S^*, \bar {\mathcal P}^*)$   forms an optimal feasible sink configuration and  $(\Sout,\Pout)$ is optimal relative to $(S^*, \bar {\mathcal P}^*)$ .  It is technically feasible  that for some $\ell$, $\bar P^*_{s_\ell} = \emptyset$ after the reallocation.  This  can not happen though because removing this $s_\ell$ from $\bar S^*$ would create a smaller feasible solution, contradicting the optimality of $S^*$.

Now, for every    $i < \ell \le r$ perform the closed $Commit(\Pnewi {s_{\ell}}, s_\ell)$ and  let
$(\barSout, \barPout) = (\Sout,\barPout)$ be the final resulting partial sink configuration.
Label the new  $T$ and $S$  as  $\bar T= (\bar V=V_v, \bar E)$ and $\bar S$.  
For completeness we note that 
 $S^*$  can now be appropriately partitioned as
$$  S^* = \{    \overbrace{ s_1,\,\ldots,\,s_i,s_{i+1},\,\ldots,\,s_r}^{\Sout\setminus \bar V},\,   \overbrace{s_{r+1},\,\ldots,\,s_j}^{\bar S= \Sout\cap \bar V},\, \overbrace{s_{j+1},\,\ldots,\,s_{k^*}}^{ S^*\setminus\Sout \subseteq\bar  V}\}.$$
To conclude, it follows  directly by the construction $(\barSout,\barPout)$ is optimal relative to $(S^*, \bar {\mathcal P}^*)$.
\qed
\end{proof}

It is important to note that after  $T_{-v}(u)$ is removed by the reaching criterion, the remaining tree $T$ might no longer be RC-viable. 
The peaking criterion would need to be checked again on $T,$ in order to reimpose  RC-viability.%Note that 

\subsubsection{Testing for self-sufficiency.}
\label{subsec: ss testing}

 Self-sufficiency is expensive to test. The following specialization will be more efficient to use:

\begin{definition}[Recursive self-sufficiency] \   \\
 Let  $v \in V_{H}(S)$.  Recall that $T(v)=(V(v),E(v))$  is the subtree of $T$  rooted at $v.$
 %$T' = (V',E')$  be a rooted subtree of $T$ with $V' \cap S \neq \emptyset$.
  $T(v)$  is {\em recursively self-sufficient} if for all $u \in V_{H(S)} \cap V(v)$, $T(u)$  is self-sufficient.
\end{definition}

Recursive self-sufficiency can be tested in a bottom-up manner.

\begin{lemma}Let  $v \in V_{H}(S)$ such that  $T(v)$ is a 
RC-viable rooted subtree of $T.$
\begin{enumerate}
\item 
Let $u \in V_{H(S)}$ be a child of $v$ such that \\ (i)   $T(u)=T_{-v}(u)$ is recursively self-sufficient, and \\ (ii) there is a sink $s \in S \cap V_{-v}(u)$ such that $v$ can evacuate to $s$.
 
 \medskip
 
 Then $BP(v,v) \cup T_{-v}(u)$ is recursively self-sufficient. 
 
 \bigskip
 
 \item Now suppose that in addition to the existence of  $u$ as in (1),\\ for every child $u'$ of $v$ in $V_{H(S)}$, $T(u')=T_{-v}(u')$ is recursively self-sufficient.\\ Then $T(v)$ is recursively self-sufficient.
 \end{enumerate}
 \label{lemma:RecursiveSS}
\end{lemma}
\begin{proof} %[Proof of Lemma \ref{lemma:RecursiveSS}]
$v$ is the only node in $V_{H(S)}$ that is in $BP(v,v) \cup T_{-v}(u)$ but not in $T_{-v}(u)$.  Thus, from the recursive self-sufficiency of $T_{-v}(u)$, to prove (1) it suffices to prove that $BP(v,v) \cup T_{-v}(u)$   itself is self-sufficient.
%\mc{Need Diagram}
Recall  that $v$  being able to evacuate  to sink $s$ means that  $BP(v,s)$ is supported by $s.$
Consider the  remaining rooted graph induced by $V_{-v}(u) \backslash BP(v,s)$. This is a rooted forest.  By the recursive self-sufficiency of $T_{-v}(u),$  each rooted tree in this forest is self-sufficient.  (1) follows.

To prove (2), similarly note that since every child of $v$ in $T_{H(S)}$ is recursively self-sufficient every node $v' \in V_{H(S)} \cap V(v)$ {\em except for $v$} must satisfy that  $T(v')$ is self-sufficient.  It thus suffices to prove that $T(v)$ itself is self-sufficient.

From (1) it is  already know that $BP(v,v) \cup V_{-v}(u)$ is self-sufficient.  Note that removing $BP(v,v) \cup V_{-v}(u)$from $T(v)$ leaves a rooted forest in which the root of each forest is a child $u'$ of $v$ in $T(v)$.  Since each such tree $T(u')$ is given to be self-sufficient, all of $T(v)$ is self-sufficient.
 \qed
\end{proof}

If Lemma \ref{lemma:RecursiveSS}  (1) holds we  say that $s$ {\em  is a witness to Lemma \ref{lemma:RecursiveSS} for $T(v)$} and store  this witness at $v$. 
%we store this witness, as well as the witness for every subtree of $T'$ rooted at some $v \in V'$.
If we do this for every recursively self-sufficient subtree then, from the proof of Lemma \ref{lemma:RecursiveSS},  it is easy to retrieve  in $O(|V'|)$ time a partition $\mathcal{P}'$ of $T'$ that witnesses the self-sufficiency of $T'$.  See Algorithm \ref{alg:FindPartitionRecursiveSS}.

Recursive self-sufficiency will provide an efficient test for the reaching criterion via the following 
 immediate corollary to Lemma \ref {lemma:RecursiveSS}:
\begin{corollary}
\label{corollary: SStesting}
Let  $v \in V_{H}(S)$ such that $T(v)$ 
is a
RC-viable rooted subtree of $T.$  Let $u_i$, $i=1,\ldots,j$ be the children of $v$ in $T_{H(s)}$ and assume  that all the $T_{-v}(u_i)$ are recursively self-sufficient.  Then exactly one of the following two cases must occur
\begin{itemize}
\item[(i)]
$\exists i,$  such that  for all sinks $s\in S$ in $T_{-v}(u_i)$, 
$f(BP(v,s),s) >  \mathcal{T}.$\\
$\Rightarrow$  $(u_i,v)$ satisfies the reaching criterion.
\item[(ii)]
$\forall i,$  there exists $s\in S$ in $T_{-v}(u_i)$ such that $f(BP(v,s),s) \le  \mathcal{T}.$\\
 $\Rightarrow$ $T(v)$ is recursively self-sufficient.
\end{itemize}
\end{corollary}

\begin{algorithm}[t]
	\begin{algorithmic}[1]
		\State $T' = (V',E')$, rooted at $v \in V'$, sinks $S' \subseteq V'$
		\Comment $T'$ is recursively self-sufficient wrt $S'$
		\State $W : V' \rightarrow S'$, where $W(u)$ is a witness to Lemma \ref{lemma:RecursiveSS} for subtree rooted at $u$
		\State$\{P_s : s \in S'\}$, a collection of sets, all initialized to empty
		
		\State $T_0 := T'$
		\Comment We will delete nodes from $T_0$, so $T_0$ may become a forest
		
		\While{$T_0$ is non-empty}
		\State $T'_0 := $ arbitrary connected component of $T_0$, viewed as rooted subtree of $T'$
		\State $v := $ root of $T'_0$
		\State $s := W(v)$
		\State $P_s := P_s \cup \mathrm{BP}(v,s)$
		\State Remove all nodes in $\mathrm{BP}(v,s)$ from $T_0$
		\EndWhile
		
		\State $\{P_s : s \in S'\}$ is a partition witnessing self-sufficiency of $T'$
	\end{algorithmic}
	\caption{Finding partition for recursively self-sufficient trees}
	\label{alg:FindPartitionRecursiveSS}
\end{algorithm}

The algorithm can walk up the hub-tree from its leafs (sinks), testing recursive self-sufficiency using case (ii) of the corollary.  This only fails if case (i) is encountered, yielding  a $(u_i,v)$ pair satisfying the reaching criterion.
The process also terminates if it reaches $r$ and finds that $T(r) = T$  is recursively self-sufficient but in that case the algorithm itself terminates because $T$ can be supported by $S.$ 
This automatically leads to the next corollary

\begin{corollary}
\label{corollary: SSFinal}
Let  $T$ be RC viable.
Then one of the following two cases must occur
\begin{itemize}
\item [(i)] $\exists (u,v) \in E_{H(S)}$, $v$ the parent of $u,$  that satisfies the reaching criterion
\item [(ii)]$T$ is recursively self-sufficient and can be fully evacuated to the nodes in $S.$ 
\end{itemize}
\end{corollary}

\subsection{The Evolution of the Hub Tree}
\label{subsec: Hub Evolution}

We have seen how,  when the peaking criterion can no longer be applied, the  working tree $T=(V,E)$ is  RC-viable with respect to 
$S = V \cap \Sout.$  Let $T_{H(S)} = (V_{H(S)}, E_{H(S)})$ be the 
directed hub-tree  %as defined in Section \ref{subsec: hub} 
with root  $r$.

%From Lemma \ref{lemma:RecursiveSS} and Corollaries \ref{corollary: SStesting}, \ref{corollary: SSFinal}, 
From  Corollary   \ref{corollary: SSFinal}, either  $T$ itself is recursively self-sufficient (and the algorithm terminates)  or  there exists some  $(u,v)$ in $T_{H(s)}$ that satisfies the reaching criterion. This permits removing the tree $T_{-v}(u)$ rooted at $u$, resulting in a new tree $\bar T$.  Since    $\bar T$   might no longer be RC-viable it needs to be checked again for the peaking criterion.

The remainder of this subsection examines what can happen next. It will show that if $\bar T$ does not remain RC-viable then there is exactly one edge, lying on a very specific known path, that satisfies the peaking criterion.  
 The removal of this edge will result in a new RC-viable $\barbarT$.
Deriving this  will require the following definitions:

\begin{definition}   In what follows $u,v \in V_{H(S)}.$ See Fig.~\ref{fig:hubtree2}.
\label{def: HT stuff}
\begin{itemize}
\item Set $\bf p(u)$ to be the parent of $u$.\\  Note that $r$ has no parent  and  $\forall u \in V_{H(S)}\setminus \{r\},$  $p(u) \in V_{H(S)}$. 
%\item For $u,v \in V_{H(S)}$ with $v$ above  $u$ set $\bf \Pi(v,u)$ to be the directed path from $u$ to $v$.\\
% Note that $ \Pi(v,u)$ lies completely  in $T_{H(S)}.$
\item Set $\bf p_H(u)$   to be the  lowest hub -node on path $\Pi(p(u),r).$\\
$p_H(u)$ is the {\em hub-parent} of $u.$   %Note that root $r$ does not have a hub-parent.
\item  Recall that 
% $\bf T_{H(S)}(u)=(V_{H(S)}(u),E_{H(S)}(u))$ to be the directed subtree of $T_{H(S)}$ rooted at $u$ and
$T(v)=(V(v),E(v))$ is  the directed subtree of $T$ rooted at $v$.
%\mc{Check if $ T_{H(S)}(u)$ is ever used}
 Set $\bf S(v)$ to be the sinks in $T(v)$ that can support the bulk path from $v$ 
$$S(v) = \{ s \in S \cap V(v) \,:\,   f(BP(v,s), s) \le \mathcal{T}\}.$$
Furthermore, with each $s \in S(v)$ associate the child of $v$ whose subtree contains $s,$ i.e.,
$$u_v(s) = \mbox{unique $u$ such that $p(u) = v$ and  $s \in V(u)$}.$$
\end{itemize}
\end{definition}

\begin{figure}[t]
\centerline{\includegraphics[width=2.5in]{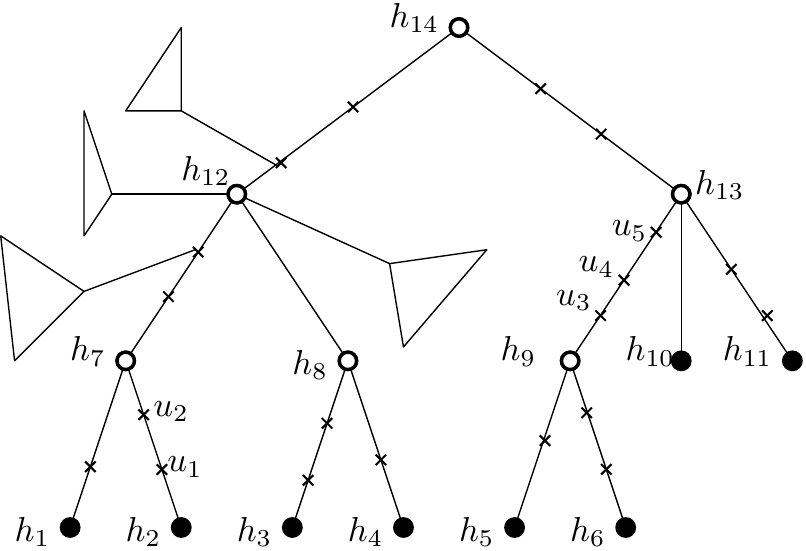}}
\caption{A labeled hubtree $T_{H(S)}.$ Solid nodes are the sinks in $S$; unfilled nodes are hubs. Triangles are outstanding branches.  $p(u_1) = u_2.$   $p_H(u_1) = h_7.$   $\Pi(h_{9},h_{13})$ is the path  $h_9,u_3,u_4,u_5,h_{13}.$ }
\label{fig:hubtree2}
\end{figure}

\begin{figure}
%\vspace{-0.2in}
\centerline{\includegraphics[width=4in]{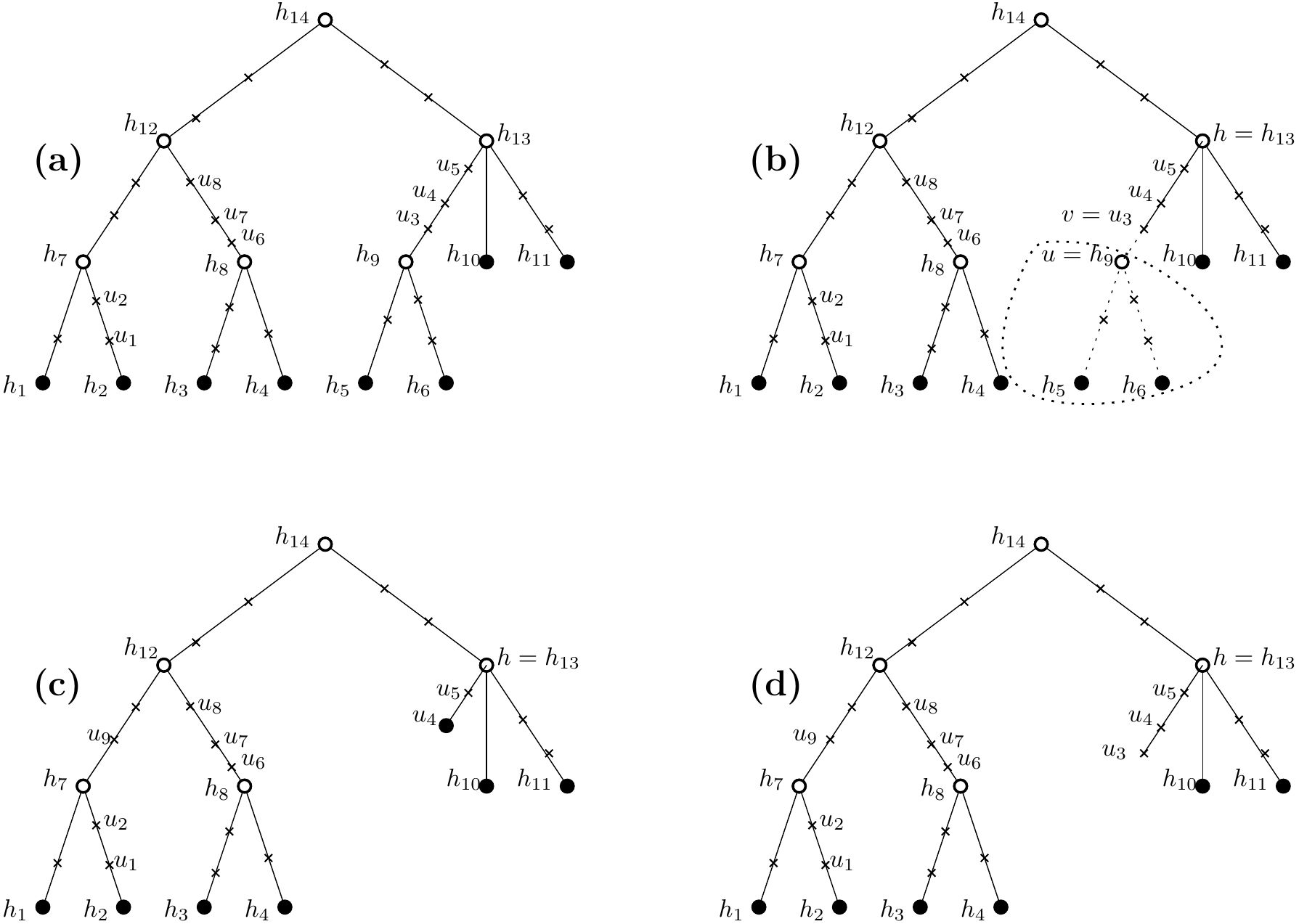}}
\caption{Illustration of the hub tree  evolution in Lemma \ref{lem:hubevolution}.  Note that outstanding branches are not drawn.  (a) is the original hub tree.  (b)  illustrates the edge $(h_9,u_3)$ satisfying the reaching criterion.  The circled subtree  and dotted edge    $(h_9,u_3)$ are then removed and $h= h_{13}.$
If  $f(V_{-h_{13}}(u_5)\cup\{h_{13}\},h_{13})> \mathcal{T}$ then some edge on $\Pi(u_3,h_{13})$ satisfies the peaking criterion so the only change to the hub tree is exactly one new sink being added on on $\Pi(u_3,h_{13})$  (with the subtree below it removed). This is Case 1 and is illustrated in (c).   
 If $f(V_{-h_{13}}(u_5)\cup\{h_{13}\},h_{13}) \le  \mathcal{T}$ and $h$ had at least three children then  $\Pi(u_3,h_{13})$ becomes an outstanding branch as in (d).  
This is case  2(a).  }
\label{fig:hubevolutiona}
\end{figure}

\begin{figure}
\centerline{\includegraphics[width=4in]{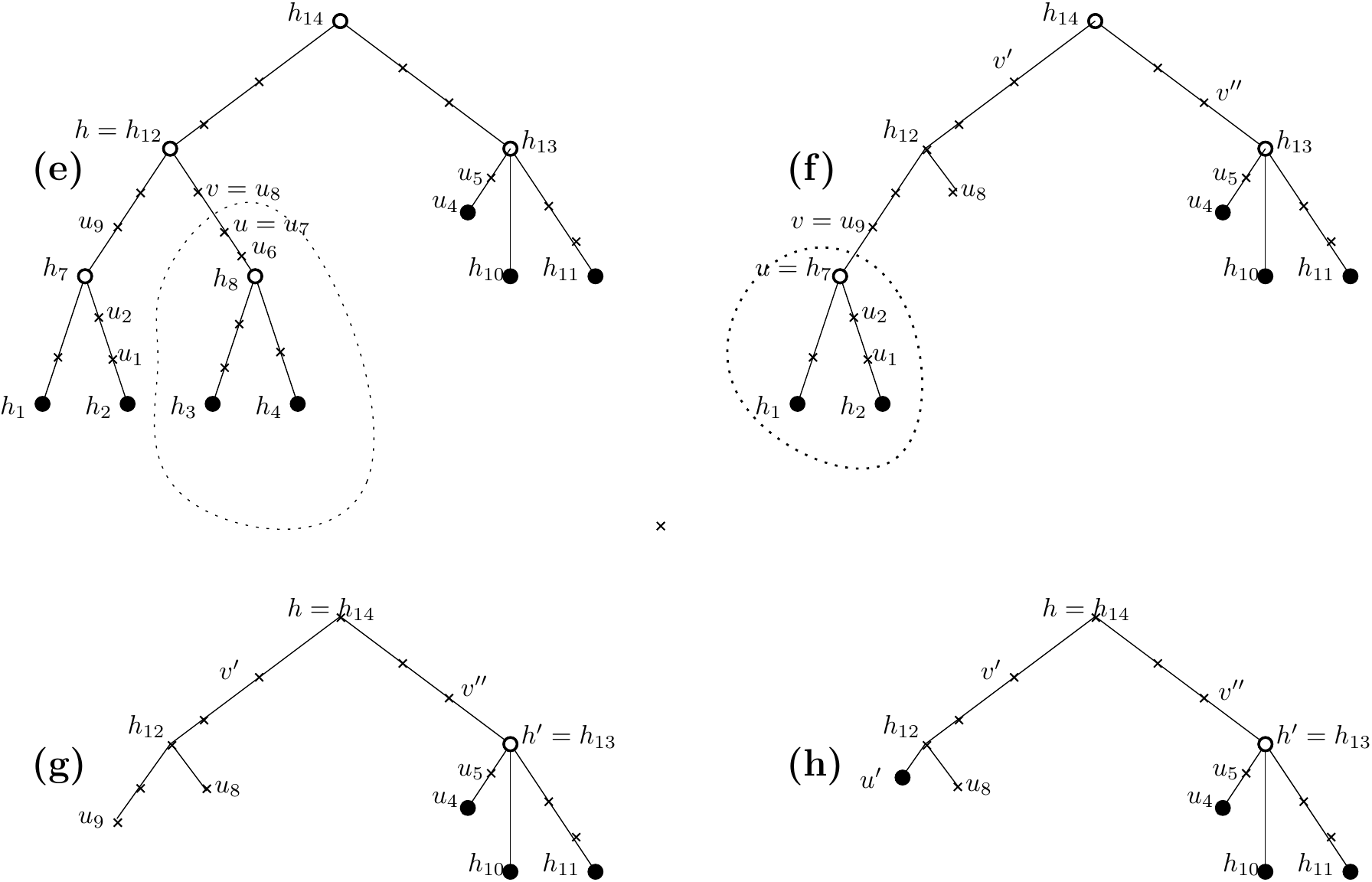}}
\caption{Continuation of the cases in  Lemma \ref{lem:hubevolution}.  (e) is the initial hub-tree.      $V_{-u_8}(u_7)$ and edge $(u_7,u_8)$ are now  removed using the reaching criterion so  $v' = u_8$ and $h = h_{12}$. (f) illustrates the situation when $f(V_{-h_{12}}(u_8) \cup\{h_{12}\},h_{12}) \le \mathcal{T}$ so $\Pi(u_8,h_{12})$ becomes an outstanding branch.  Since $h_{12}$ is not the root and originally only had two hub children the hub tree remains the same except that $h_{12}$ is no longer  a hub. This is case 2(b). Next, starting from (f),  $V_{-u_9}(h_7)$ and  edge  $(h_7,u_9)$ are  removed using the reaching criterion.  $h'=h_{13}$ will now become the new root of the hub tree. If path $\Pi(u_9,h_{13})$  does not contain any edge satisfying the peaking criterion then $\Pi(u_9,h_{13})$  together with its outstanding branches becomes an outstanding branch.  This is case 2(c)(i) and illustrated in (g).   If path $\Pi(u_9,h_{13})$  does  contain an edge $(u',p(u'))$ satisfying the peaking criterion then $\Pi(u',h_{13})$   becomes a path in the new hub tree. This is case 2(c)(ii), illustrated in (h).  Note that the node $v''$ can be {\em anywhere} on $\Pi(u_9,h_{13})$, including  on $\Pi(h_{14},h_{13})$. The remaining two cases 2(c)(iii) and 2(c)(iv) are simpler and not illustrated. }
\label{fig:hubevolutionb}
\end{figure}

\begin{lemma} 
\label{lem:hubevolution}
%Given  $(\Sout, \Pout)$ and associated $T$, let $T_{H(S)}$ be the associated hub tree rooted at $r.$ 
Let $T_{H(S)}$ be rooted at $r$ and  suppose $(u,v)$ satisfies the reaching criterion. 
Set 
$$
h = \left\{
\begin{array}{ll}
v  & \mbox{if $v$ a hub in $T_{H(S)},$}\\
p_H(v) & \mbox{otherwise.}
\end{array}
\right.
$$
Furthermore, if $h \not = v$ set
$$v' = \mbox{immediate child of $h$ such that $v  \in V_{-h}(v').$}$$

We partition the possibilities into 7 different scenarios as described below -- (1),  (2a), (2b),  (2ci), (2cii), (2ciii), (2civ) -- and state the behavior in each separately:
\begin{itemize}
\item [(1)] $h \not= v$ and $f(V_{-h}(v') \cup\{h\},h)  >  \mathcal{T}$.
\item [(2)]   $h=v$ \quad or  \quad  $h \not=v$ and $f(V_{-h}(v') \cup\{h\},h)  \le \mathcal{T}$.   
\begin{itemize}
\item   [(2a)]   {\bf $h$ has at least three children in  $T_{H(S)}$}.
\item   [(2b)] {\bf $h$ has  exactly two children in $T_{H(S)}$ and $h\not=r$} . % [(2a)]
\item    [(2c)]{\bf $h$ has exactly two children  in $T_{H(S)}$ and $h=r$:}\\
		Let $h'$ be  the unique remaining hub-node such that $p_H(h') =r$\\ and  set $v''=p(h').$ 
%		$v'' =\mbox{immediate neighbor of $h'$ such that $h  \in V_{-h'}(v'').$}$  
		\begin{itemize}
		\item  [(2ci)]  $f(V_{-h'}(v'') \cup\{h'\},h')  \le  \mathcal{T}$ \quad and \quad $h'$ is not a sink.
		\item  [(2cii)] $f(V_{-h'}(v'') \cup\{h'\},h')  >  \mathcal{T}$ \quad and \quad $h'$ is not  a sink.
		\item  [(2ciii)]  $f(V_{-h'}(v'') \cup\{h'\},h')  \le  \mathcal{T}$ \quad and \quad $h'$ is  a sink.
		\item  [(2iv)] $f(V_{-h'}(v'') \cup\{h'\},h')  >  \mathcal{T}$ \quad and \quad $h'$ is  a sink.
		\end{itemize}
\end{itemize}
\end{itemize}

Let $(\barSout,\barPout)$, $\bar T$, $\bar S$ be the result after applying the Reaching Lemma  to $(u,v)$,  removing $T_{-v}(u)$   and committing its nodes to the sinks in $S' = \{s \in S \cap V_{-v}(u)\}$. 

Let $(\barbarSout,\barbarPout)$, $\barbarT$, $\barbarS$ then be the result after  applying the {\em next} peaking phase. ${\barbarT}_{H\left(\barbarS\right)}$ 
is the new hub tree that results. All other variables will be renamed accordingly.

The results in the 7 scenarios  then satisfy (See Figs.~\ref{fig:hubevolutiona} and \ref{fig:hubevolutionb}): 
\begin{itemize}
\item  {\bf Case 1:} $\bar T$ is not RC-viable relative to $\bar S$. \\ 
Then some edge  $(u',p(u')) \in \Pi(v,h)$ is the 
 unique  edge that satisfies the peaking criterion for $\bar T$, $\bar S$. Furthermore, after
 (open) $Commit(\bar V_{-p(u')}(u'),u')$ creates $(\barbarSout,\barbarPout)$, $\barbarT$, $\barbarS$

\begin{itemize}
\item New  hub tree ${\barbarT}_{H\left(\barbarS\right)}$ is  $T_{H(S)}$ with  all nodes in $T_{-p(u')}(u')$ removed and 
\begin{itemize}
\item Node  $u'$ added  back as  sink.\\[0.03in] %  and  new hub tree edges $\Pi(h,u')$ added\\[0.03in]
\end{itemize}
\end{itemize}
\item {\bf Case 2a:} $\bar T$ is RC-viable relative to $\bar S$.% so the peaking phase  is skipped %does not add any edges and therefore makes no changes. 
\begin{itemize}
\item  %The new hub tree $\barbarT_{H(\barbarS)}$ is the original hub tree $T_{H(S)}$ with all nodes in $T_{-v}(u)$ removed.
New  hub tree ${\barbarT}_{H\left(\barbarS\right)}$ is  $T_{H(S)}$ with  all nodes in $T_{-v}(u)$ removed.\\[0.03in]
\end{itemize}
\item {\bf Case 2b:} $\bar T$ is RC-viable relative to $\bar S$. % so the peaking phase  is skipped  %  $\bar T$ is RC-viable relative to $\bar S$ so the peaking phase does not add any edges and therefore makes no changes. 
\begin{itemize}
\item  New  hub tree ${\barbarT}_{H\left(\barbarS\right)}$ is  $T_{H(S)}$ with  all nodes in $T_{-v}(u)$ removed % The new hub tree $\barbarT_{H(\bar S)}$ is the original hub tree $T_{H(S)}$ with all nodes in $T_{-v}(u)$ removed and 
\item $h$ remains as node in $T_{H(S)}$  but is no longer a hub.\\[0.03in]
\end{itemize}
\item {\bf Case 2ci:}
$\bar T$ is RC-viable relative to $\bar S$. % so the peaking phase  is skipped % $\bar T$ is RC-viable relative to $\bar S$ so the peaking phase does not add any edges and therefore makes no changes. 
\begin{itemize}
\item New hub tree $\barbarT_{H(\bar S)}$ is   $T_{H(S)}$ with  all nodes in $T_{-v}(u)$ removed  and re-rooted at $h'$.\\[0.03in]
\end{itemize}
\item  {\bf Case 2cii:} 
Some edge  $(u',p(u')) \in \Pi(v,h')$ is the 
 unique  edge that satisfies the peaking criterion for $\bar T$, $\bar S$. 
  Furthermore, after (open) $Commit(\bar V_{-p(u')}(u'),u')$  creates $(\barbarSout,\barbarPout)$, $\barbarT$, $\barbarS$
\begin{itemize}
\item New  hub tree ${\barbarT}_{H\left(\barbarS\right)}$ is $T_{H(S)}$ with  all nodes in $T_{-p(u')}(u')$ removed and
\begin{itemize}
\item re-rooted at $h',$
\item with new hub $u'$ added  as a  sink.\\[0.03in] % and  new hub tree edges $\Pi(h,u')$ added
\end{itemize}
\end{itemize}
\item  {\bf Case 2ciii:}   $V$ is fully served by $h'$, so $\barbarV = \emptyset$ and algorithm terminates.\\[0.03in]
\item  {\bf Case 2civ:} 
Some edge  $(u',p(u')) \in \Pi(v,h')$ is the 
 unique  edge that satisfies the peaking criterion for $\bar T$, $\bar S$. 
  Furthermore, after (open) $Commit(\bar V_{-p(u')}(u'),u')$  creates $(\barbarSout,\barbarPout)$, $\barbarT$, $\barbarS$ with $|\barbarS| =2$.  Then exactly one of the following two cases occur
\begin{itemize}
\item  $\barbarS = \barbarV =\{u',h'\}$  and algorithm terminates after performing $Commit(\{u'\},u')$ and $Commit(\{h'\},h')$, or
\item $|\barbarV| >2$ and new  hub tree ${\barbarT}_{H\left(\barbarS\right)}$ is $T_{H(S)}$ with  all nodes in $T_{-p(u')}(u')$ removed and tree re-rooted at $v''.$
\end{itemize}
\end{itemize}
Furthermore for all non-terminating cases,  for all nodes $w \in  \barbarV_{H\left(\barbarS\right)},$
\begin{equation}
\label{eq:barbarSnew}
\barbarS(w) = 
\left\{
\begin{array}{lcl}
   S(w) \setminus S'   &  \quad & \mbox {In cases  2a,  2b and 2ci}\\
\left(S(w) \setminus S'\right) \cup I(w,u')  && \mbox{In cases 1, 2cii and 2civ}
\end{array}
\right.
\end{equation}
where 
\begin{equation}
\label{eq:barbarSnew2}
I(w,u') =
\left\{
\begin{array}{cl}
\{u'\}  &   \mbox { If  $ u' \in \barbarT(w)$  and  $f(\barbarBP(w,u'), u') \le \mathcal{T}$}\\
\emptyset & \mbox{otherwise}
\end{array}
\right.
\end{equation}
\end{lemma}
\begin{proof}
See Figs.~\ref{fig:hubevolutiona} and \ref{fig:hubevolutionb}.

Suppose $(u,v)$ satisfies the reaching criterion and subtree $T_{-v}(u)$ is removed from $T$ resulting in $\bar T.$

First assume that $h \not = v.$
With the exception of the edges on $\Pi(v,h)$, all other edges in $\bar T$ still either have a sink beneath them or are on an outstanding branch.  Thus, any edge that satisfies the peaking criterion must be on $\Pi(v,h)$.  Furthermore, if any edge $ (u',p(u')) \in \Pi(v,h)$ satisfies the peaking criterion and a sink is placed on $u'$, no other edge on  $\Pi(v,h)$ can then satisfy the peaking criterion. So, at most one edge may satisfy the peaking criterion after the removal of $T_{-v}(u)$.

\smallskip

\par\noindent\underline{Case 1:}
If  $f(V_{-h}(v') \cup\{h\},h)  > \mathcal{T}$,  Lemma \ref{lemma:PCRecurseInto} states that such an edge $ (u'',v'')$ satisfying the peaking criterion must exist so the Lemma is correct for Case 1.

\smallskip 

\par\noindent\underline{Case 2:} 

If  $h \not = v$  and $f(V_{-h}(v') \cup\{h\},h) \le \mathcal{T}$, then by definition, the branch containing 
$\Pi(v,h)$  that falls off of $h$   becomes an outstanding branch falling off of $h.$  Practically, this is equivalent to removing the edges on path  $\Pi(v,h)$ from the hub tree.

 If $h=v$  the path $\Pi(v,h)$ is just the vertex $v$.  So  $\Pi(v,h)$ contains no edges  and the situation is now the same as the previous paragraph, i.e., the (non-existent)  edges on path  $\Pi(v,h)$  are trivially removed.  Thus $h=v$  and   $h\not=v$  with   $f(V_{-h}(v') \cup\{h\},h) \le \mathcal{T}$   result in the same type of structure.
 
\smallskip 

\par\noindent\underline{Case 2a:} 
If $h$ originally had at least three children in  $T_{H(S)}$, then $h$ still retains at least two children containing  sink leaves below it.  Thus $h$ remains a hub. As noted above, the path $\Pi(v,h)$ (if it exists) is an outstanding  branch and  therefore does not contain any edge satisfying the peaking criterion.  Since all edges satisfying the peaking criterion must lie on the path no such edges exist.
 Thus the lemma is correct for  Case 2(a).

\smallskip

Next assume that $h$ originally had two children in  $T_{H(S)}$. $h$ now only has one branch below it that contains a sink. So $h$ is no longer a hub.  This splits into cases 2b and 2c.

\smallskip
\par\noindent\underline{Case 2b:}  If $h \not=r$ then $p_H(h)$ exists and   remains  a hub  because all of its old branches containing sinks still contain sinks. Again, the only possible location for an edge that satisfies the peaking criterion would be on the branch that  contained  $\Pi(v,h)$ (if it exists) but since this is now an outstanding branch, no such edge exists. Thus the lemma is correct for  Case 2(b).

\par\noindent\underline{Case 2c:} 
 $h=r$,    now only has one  branch (the one that did not contain $u$) containing   sinks, so it is no longer a valid root for the hub tree.   Consider the tree as being re-rooted at $h'$. The old subtree rooted at $h'$ remains rooted at $h'.$ In addition, $h'$ is now the root for  the path $\Pi(v,h').$  Note that all branches  falling off $\Pi(v,h')$ are outstanding branches because of $T$'s RC-viability.  Thus the only possible edges that could satisfy the peaking criterion are on $\Pi(v,h').$ 

\smallskip
\par\noindent\underline{Case 2ci:} 
No edge on $\Pi(v,h')$ can satisfy the peaking criterion and thus the lemma is correct for  Case 2(c)(i).

\smallskip
\par\noindent\underline{Case 2cii:} 
Lemma \ref{lemma:PCRecurseInto}  implies that  an edge satisfying the peaking criterion must exist on $\Pi(v,h)$.  Furthermore, similar to the argument in Case 1, only one such edge can exist
and thus the lemma is correct for  Case 2(c)(ii).  

\smallskip
\par\noindent\underline{Case 2ciii:} 
No edge on $\Pi(v,h')$ can satisfy the peaking criterion and thus $h'$ serves all of $V.$

\smallskip
\par\noindent\underline{Case 2civ:} 
Lemma \ref{lemma:PCRecurseInto}  implies that  an edge satisfying the peaking criterion must exist on $\Pi(v,h)$. Similar to the argument in Cases 1 and 2cii, only one such edge can exist. Thus, the one new sink $u'$ is created so
$\barbarS = \{u',h'\}.$  If $\barbarS=\barbarV$ then   the algorithm terminates as in Corollary \ref{corollary:PCStoppingCondition}(3). Otherwise $v''$ is on the path conecting the two sinks 
and thus the lemma is correct for  Case 2(c)(iv).  

\medskip

Finally,  again by checking the cases individually, it is easy to derive by definition  that 
 (\ref{eq:barbarSnew})  and (\ref{eq:barbarSnew2})  correctly state the  new values for $\barbarS(w)$.
\qed
\end{proof}

Lemma  \ref{lem:hubevolution} implies that after  the Reaching Lemma is applied,
%reaching criterion  is satisfied and appropriate nodes committed, 
at most one  edge in  the remaining working tree $T$, located on an easily identifiable path, can satisfy the peaking criterion before $T$ becomes RC-viable again.   Furthermore,  the new resulting hub tree can be constructed easily from the old one.

\section{Designing an Algorithm for The Discrete Bounded Cost Problem}
\label{Sec: Bounded Algorithm}

Combining the pieces from Section \ref {Section: Bounded Cost}  yields a 
%The pieces developed previously can now all be combined to create a 
generic algorithm for solving the discrete bounded-cost problem.  This is shown in Algorithm \ref {alg:BoundedCostFull2}.

\begin{algorithm}[t]
	\begin{algorithmic}[1]
		\State  $T := (V,E) := T_{\mathrm{in}}$
		\State $S := S_\mathrm{out} = \emptyset, \mathcal{P}_{\mathrm{out}} := \emptyset$. 
		\State
		\Procedure{Peaking.Phase}{}
		\While{Some edge $(u,v) \in E$ satisfies Peaking Condition}
		\Comment Lemma \ref{lemma:PeakingCriterionPutSink}
		\State  Commit$(V_{-v}(u) \cup \{v\},\, s)$
		\EndWhile
		\If {$S = \{s\}$ for some $s \in S$}
		\Comment Corollary \ref{corollary:PCStoppingCondition}
		\State{Commit$(V,s)$.}
		\ElsIf{$S = \emptyset$}
		\Comment Corollary \ref{corollary:PCStoppingCondition}
		\State {Choose any $v \in V$ and  Commit$(V,v)$.}
		\ElsIf{$V=S=\{s,s'\}$}
		\Comment Corollary \ref{corollary:PCStoppingCondition}
		\State {Commit$(\{s\},s)$  and  Commit$(\{s'\},s')$.}
		\EndIf
		\If{$|\Sout| >k$}
		\State{BREAK and Return ``Infeasible"}
		\EndIf
		\EndProcedure
		\State

		\Procedure{Reaching.Phase}{}
		\If {$\exists (u,v) \in E_{H(S)}$ that satisfies Reaching Condition}	
		\Comment Lemma \ref{lemma:RCTrimTree}
		\State Remove $V_{-v}(u)$ from $T$
		\State {Commit blocks for $V_{-v}(u)$}
		\Else
		\Comment {$T$ is recursively self-sufficient}
		\State{Commit all of $V$ to $S$}
		\Comment  Corollary \ref{corollary: SSFinal}
		\EndIf
		\EndProcedure

		\State
		\State
		\Comment{Start Algorithm}
		\State \Call{Peaking.Phase}{}	
	         \State
		\While{Tree $T$ is not empty}
		\State{Create Hub Tree $T_{H(S)}$ from $T,S$}
		\State \Call{Reaching.Phase}{}	
		\If {Tree $T$ is not empty}
		\State \Call{Peaking.Phase}{}	
		\EndIf
		\EndWhile
		\State Output $S_\mathrm{out}$, $\mathcal{P}_{\mathrm{out}}$
	\end{algorithmic}
	\caption{Generic Bounded cost algorithm}
	\label{alg:BoundedCostFull2}
\end{algorithm}

This  algorithm initializes  by setting 
$(\Sout,\Pout) = (\emptyset,\emptyset)$,   $T := (V,E) := T_{\mathrm{in}}$ and $S = \Sout \cap V = \emptyset.$
%$S := S_\mathrm{out} = \emptyset, \mathcal{P}_{\mathrm{out}} := \emptyset$ and $T := (V,E) := T_{\mathrm{in}}$.
This   $(\Sout,\Pout)$ is trivially optimal relative to $T.$

The algorithm then attempts to find edge $(u,v)$ that satisfies the peaking criterion. Every time it finds such an edge it performs an open commit.  From Lemma \ref{lemma:PeakingCriterionPutSink}, this maintains   $(\Sout,\Pout)$ as being  optimal relative to $T.$ 

 If adding a sink via the peaking criterion ever sets  $|\Sout|>k$, the algorithm reports that no feasible sink configuration exists.
If no  edge satisfying the criterion can be found and  $|S|=| \Sout \cap V|<2$ then the algorithm  finds an optimal feasible configuration using 
Corollary \ref{corollary:PCStoppingCondition}.   More specifically
\begin{itemize}
\item 
 if $|S|=0$ then ${\rm Commit}(V,v)$  for any $v \in V$ 
\item if $S=\{s\}$ then ${\rm Commit}(V,s)$ and the algorithm concludes with $(\Sout,\Pout)$ being an optimal feasible configuration for the original $\Tin.$
\end{itemize}

If no  edge satisfying the peaking criterion  exists and  $|\Sout| \ge 2$ then   $T$ is $RC$-viable.  The algorithm  then attempts  to find an edge satisfying the reaching criterion. 
If it succeeds,   it performs the corresponding closed commits and  returns to trying to find an edge satisfying the peaking criterion.   By Lemma \ref{lemma:RCTrimTree} this maintains   $(\Sout,\Pout)$ as being  optimal relative to $T.$  If no edge satisfying the reaching criterion  exists, then by
  Corollary \ref{corollary: SSFinal}, $T$ can  be fully committed to the sinks in $S;$  the resulting sink configuration  $(\Sout,\Pout)$ is  an optimal feasible configuration for the original $\Tin.$

The algorithm described is  {\em generic} because it does not specify an order or methods for finding edges that satisfy  the peaking or reaching criteria.  The remainder of this section develops efficient techniques for both.
 It proceeds as follows:
 
 \begin{itemize}
 \item [\ref{subsec:centroid}] Implementation of the first peaking phase via a tree centroid decomposition
 \item[\ref{subsec:peaking bin}] Implementation of  all other peaking phases.
 \item[\ref{subsec:new reaching bs}] Creation of the hub tree after a peaking phase
 \item[\ref{subsec:reaching imp}] Implementation of the reaching phase after constructing the hub tree.
 \end{itemize}
The decomposition into these parts will make it easier to apply  parametric searching  in section \ref{Section: Full Problem}  to  solve the Minmax $k$-sink problem.

\subsection{Implementing the Peaking Phase}

\subsubsection{The First Peaking Phase  via Tree Centroid Decomposition.}
\label{subsec:centroid}

The peaking phase checks the peaking criterion  on  all possible directed\footnote{It thus needs to check each  edge in the tree twice; once in each direction.} edges $(u,v)$, committing $V_{-v}(u)$ to $S_u$  if appropriate.  
Explicitly  checking every edge would require  
$O(n)$ oracle calls.
This subsection develops a  method that only requires $O(\log n)$  (amortized) oracle calls plus $O(n \log n  + nk)$ extra work for the first peaking phase. Section \ref{subsec:peaking bin}  will then show how to implement each subsequent peaking phase using $O(\log n)$  (actual) oracle calls plus $O(nk)$ total extra work over all the remaining phases.

We start by noting that information garnered when checking an edge for the peaking criterion will often imply that many other edges will not satisfy the criterion and therefore need not  be tested. The algorithm will take advantage of this and  create an order for checking the edges -- based on a recursive centroid decomposition of $\Tin$ -  that will essentially guarantee that either many calls will not have to be made OR that  the average {\em size} of an oracle call will be small.  The asymptotic subadditivity of the oracle $\mathcal{A}$ will  then yield an amortized running time of the first peaking stage equivalent to $O\log n)$ oracle calls.

\begin{definition} Let $(u,v)$ be any directed edge. 
\label{def:labels}
\begin{itemize}
\item  [(i)] $(u,v)$ satisfies Condition {\bf L1} if  
\begin{equation}
\label{eq:def L1}
 f(V_{-v}(u) \cup \{v\}, v) >  f(V_{-v}(u), u) > \mathcal{T}.
 \end{equation}
\item  [(ii)] $(u,v)$ satisfies Condition {\bf L2} if
\begin{equation}
\label{eq:def L2}
f(V_{-v})(u), u)  \le f(V_{-v}(u) \cup \{v\}, v) \le \mathcal{T}.
\end{equation}
\end{itemize}
\end{definition}

\begin{definition} (Fig.~\ref{fig:AboveBelow})
Let $u$ be a neighbor of $v$ and $u'$ a neighbor of $v'$ Then
\begin{itemize}
\item Directed edge $(u',v')$ is {\em above}  directed edge $(u,v)$ if $v$ is on path $\Pi(u,u')$ and $u'$ is on path $\Pi(v,v')$.
\item Directed edge $(u',v')$ is {\em below}  directed edge $(u,v)$ if
$(u,v)$ is above  $(u',v').$
\end{itemize}
\end{definition}
%Path monotonicity immediately  implies the following lemma

 \begin{figure}
 \centerline{\includegraphics[width=3in]{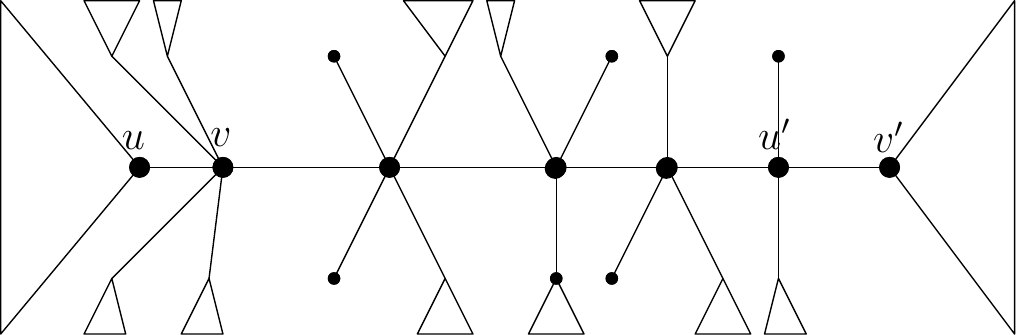}}
 \caption{Edge $(u',v')$ is {\em above}  $(u,v).$  Edge $(u,v)$ is {\em below} $(u',v').$}
 \label{fig:AboveBelow}
 \end{figure}

\begin{lemma}  Let $(u,v)$ be a directed edge.   Then, one of the following three cases must hold with the corresponding consequences.
\label{L1L2 lemma}
\begin{itemize}
\item [(i)]
$(u,v)$  satisfies  {\bf L1}. Then 
all edges  $(u',v')$ above  $(u,v)$  satisfy  {\bf L1}.
\item [(ii)] $(u,v)$  satisfies  {\bf L2}.
Then  all edges  $(u',v')$  below  $(u,v)$  satisfy  {\bf L2}.
\item[(iii)] $(u,v)$  does not satisfy  {\bf L1} or {\bf L2}.
Then 
(a) $(u,v)$ satisfies the peaking criterion,  (b)  all edges  $(u',v')$ above  $(u,v)$  satisfy  {\bf L1} and (c)
all edges  $(u',v')$  below  $(u,v)$  satisfy  {\bf L2}.
\end{itemize}
\end{lemma}
\begin{proof}
Follows immediately from the definitions and path  monotonicity. \qed
\end{proof}

\begin{lemma}
If at anytime during the first peaking phase edge $(u',v')$ satisfies either  {\bf L1} or {\bf L2},   $(u',v')$ will never satisfy the peaking condition anytime later during the first peaking phase.
\label{L1L2 lemma2}
\end{lemma}
\begin{proof}

First suppose that $(u',v')$ satisfied {\bf L1} at some time. If $T_{-v'}(u')$  never changes during the phase then (\ref{eq:def L1}) will remain satisfied  and $(u',v')$ will never  satisfy the peaking criterion.   $T_{-v'}(u')$ can only  change during the phase if a  sink is committed {\em inside} $V_{-v}(u)$.  But  sinks are never removed during  the phase so if  a sink is placed in $V_{-v}(u)$, $(u',v')$ will still not be able to  satisfy the peaking criterion  during  the phase.

Now suppose that $(u',v')$ satisfies {\bf L2} at some time. If $T_{-v'}(u')$  never changes during the phase then (\ref{eq:def L2}) will remain satisfied  and $(u',v')$ will never  satisfy the peaking criterion nor will any edge below it.   The only way for  $T_{-v'}(u')$ to change in this case is for some sink to be placed {\em above} it and remove i$T_{-v'}(u')$.  But, once it is removed $(u,v)$ will obviously never again satisfy the peaking condition.
 \qed
\end{proof}

At the start of the first peaking phase all edges will be initialized and marked as   
{\bf U}(nknown).   

Whenever  a  $(u,v)$ is tested for the peaking criterion,   one of  the three cases in  Lemma \ref {L1L2 lemma} will occur. If case (i), label all edges $(u',v')$ above   $(u,v)$  as {\bf L1}.  If case (ii),  label all edges $(u',v')$ below   $(u,v)$  as {\bf L2}.  If case (iii) do both  before removing the edge and committing
$T_{-v}(u)$ to $u.$

After  labelling an edge  {\bf L1} use   Breadth-First Search  to label all edges above it as {\bf L1} as well.  If the procedure ever encounters an edge already labelled {\bf L1} it does not continue past that edge (since all of the edges above it were already  labelled {\bf L1}).
Thus the {\em  total} time to mark edges as {\bf L1}  in the phase is $O(n).$  A similar analysis shows that the total time required to mark edges as {\bf L2} in the phase is also $O(n).$

The algorithm will check the edges in $(u,v)$  in a  special order to be  defined below. When checking an edge $(u,v)$   it first checks whether 
it is already marked as {\bf L1} or {\bf L2}.  
If it is, it skips it since   from  Lemma \ref{L1L2 lemma2}, it doesn't satisfy the peaking criterion.
 Only if  $(u,v)$  is  still marked {\bf U} does the algorithm actually run the oracle to evaluate  $f(V_{-v}(u), u)$ and 
$f(V_{-v}(u) \cup {v}, v) $ After  completing the calculation it  marks further edges using Lemma  \ref{L1L2 lemma} and then performs a commit if  required.

\medskip

Recall that the {\em centroid}  $\rho(T)$ of a tree $T$ with $n$ nodes is a node $u$ such that all subtrees falling off of $u$ contain  $\le n/2$ nodes. A centroid exists and can be found in $O(n)$ time \cite{kariv1979algorithmic}.  The algorithm will use a standard recursive centroid decomposition process to specify the  edge checking order.

The process creates two sequences  $\mathcal{F}_i$  and $L_i$,  containing, respectively, forests of trees,  and sets of vertices.   For node $u \in \Vin,$ let $\mathcal{N}(u)$ denote the set of neighbors of $u$ in the full working tree $T= \Tin$. % at the start of the phase.

\medskip

\par\noindent\underline{Stage $i=0:$}  Set $\mathcal{F}_0 = \{\Tin\}$. % and  $W_0 = \emptyset$.

\medskip

\par\noindent\underline{Stage $i$,\, $i>0:$}
%Given $W_{i-1}$, define $W_{i}$ and $\mathcal{F}_{i}$ as follows. 
Initialize $L_{i} := \emptyset$ and $\mathcal{F}_{i} = \emptyset$. \\
For every tree $T'$ in the forest $\mathcal{F}_{i-1}$,\\
 \hspace*{.2in}   Remove $\rho(T')$ from $T'$,  resulting in a forest of subtrees\\  %remaining nodes into a forest\\
 \hspace*{.2in}    Move  the resulting forest of subtrees  into $\mathcal{F}_{i}$\\
  \hspace*{.2in}    Add  $\rho(T')$ into $L_{i}$. 
 % \hspace*{.2in}  $\forall u \in \mathcal{N}(\rho(T))$  write down   $(u,\rho(T))$ 
 % 
 % is marked {\bf U} and is in current working tree $T,$\\
%\hspace*{1.2in} check peaking criterion on 
%                           $(u,\rho(T))$ 
\medskip

This processes terminates when  $\mathcal{F}_t$ is empty.  Note  that every $v \in V$ is  chosen as  the centroid of exactly  one tree   in this process so the $L_i$ are a partition of $V.$
%i.e., $\cup_{i=1}^t L_i = V$ and $\forall i \not= j,$  $L_i \cap L_j = \emptyset.$ 
Set $W_j = \bigcup_{i=1}^j L_i.$  Note that $W_t = \Vin.$

From the definition of a centroid, trees in $\mathcal{F}_{i}$  all have size $\le n/2^{i}$ so $t = O(\log n)$. Furthermore, the trees in $\mathcal{F}_i$ are disjoint so each stage requires only $O(n)$ time and the entire decomposition uses $O( n \log n)$ time.

The peaking phase will now process the edges in $E$ by examining $i=1,2,\ldots,t$ in order and,  for every $v \in L_i$, checking all $(u,v)$ where $u \in N(v).$
Since  $W_t = V,$  this checks all edges.

 See 
Algorithm \ref{alg:Centroid Processing}. When edge $(u,v)$ is encountered in line 5, the algorithm first determines if it is still marked {\bf U}.  If it is,  the algorithm  saves it in set $E'_i$ and performs the appropriate oracle calls but defers the actual checking of the peaking condition to later
 in the stage.\footnote{$(u,v)$ could have been checked immediately.    The deferment is introduced to simplify the later use of  parametric searching
 in Section \ref{Section: Full Problem}.}.

Lines 12-14  check for the degenerate case in which % $T$ contains no sinks and
 the current tree centroid $v=\rho(T')$ can support $\Tin.$  Since $T$ can be served by just the one sink $v$, the algorithm terminates.

Otherwise,  the algorithm  examines the results of the oracle calls on edges in $E'_i$, applying Lemma 
\ref {L1L2 lemma}  to appropriately label edges as {\bf L1} or  {\bf L2} and then  applying the Peaking Lemma to the edges that satisfy the criterion to create new sinks and commit blocks to them.

\begin{algorithm}[h]
	\begin{algorithmic}[1]
	\State Perform the tree centroid decomposition process defined in the text
	        \For {$i$ = 1 to $t$}
	        \Comment { Stage $i$. Checks $(u,v)$ for $v \in L_i.$}
	         \State {Set $E'_i = \emptyset$}
	         \For {all $v \in L_i$}   
		   \Comment {Process all centroids $v = \rho(T')$ for some $T' \in \mathcal{F}_{i-1}$ }
	         	\For {all $u \in   \mathcal{N}(v)$ }
	         	     \If  {$ (u, v)$ is marked  {\bf U}}
	         		\State{Add $(u,v)$ to $E'_i$}
	         		\State{Evaluate   $a(u,v) = f(V_{-v}(u), u)$,\   and  \ $b(u,v) = f(V_{-v}(u) \cup \{v\}, v)$}%	         		 
		           \EndIf
		           	        \EndFor
		  \EndFor
		  \For {all $v \in L_i$} 
		  \Comment{Checks for degeneracy condition  (\ref{eq:ass w})}
		          \If {$\forall u \in   \mathcal{N}(v)$,  $(u,v) \in E'_i$ and  $b(u,v) \leq \mathcal{T}$ }
	         	     \State {Commit $V$ to $v$ and terminate the algorithm.}
                       \EndIf
		  \EndFor  
		  \For { $(u, v) \in E'_i$}
                   \Comment{Mark edges appropriately and check peaking condition}
		   \State{Apply Lemma \ref {L1L2 lemma} appropriately based on whether  $a(u,v) \le  \mathcal{T}$ and $b(u,v) >  \mathcal{T}$ }
		   \If {$a(u,v) \le \mathcal{T}$  and $b(u,v) > \mathcal{T}$}
		   \Comment {$(u,v)$ satisfies peaking Criterion}
		   \State { apply Peaking Lemma to commit $V_{-v}(u)$ to $u$ }
		   \EndIf
%	         		\If { $(u,v)$ satisfies the Peaking Criterion}  
		  \EndFor
%	         		\State{Check \  $f(V_{-v)}(u), u) \leq \mathcal{T}$ \ and  \ $f(V_{-v)}(u) \cup \{v\}, \rho(T')) > \mathcal{T}$ \ via Oracle}%
%		         	\State{Use Oracle Information to apply Lemma \ref {L1L2 lemma} appropriately}
%	         		\If { $(u,v)$ satisfies the Peaking Criterion}  
%	         		\State{ apply Peaking Lemma to commit $V_{-v}(u)$ to $u$. }
%
%	         	\EndIf
	         	\EndFor         	
 		  
	\end{algorithmic}
	\caption{Processing the edges, stage by stage}
	\label{alg:Centroid Processing}
\end{algorithm}

 We  now examine the running time of  Algorithm  \ref {alg:Centroid Processing}. We already saw that the total cost of labelling edges is $O(n)$. The remainder of the algorithm with the exception of line 8 can also be implemented in $O(n)$ time.  Now  define
 $$
 C(u,v) = 
 \left\{
 \begin{array}{ll}
  V_{-v}(u)  & \mbox{if $(u,v)$ was  added to $E'_i$ in  line 7.} \\ % labelled {\bf L1} or  {\bf L2}  before it was  processed,}\\
  \emptyset & \mbox{otherwise.}
 \end{array}
 \right.
 $$
This is well defined since every edge appears in at most on $E'_i$.
By definition, the cost of implementing  line 8 for $(u,v)$ is  
  $O(t_\mathcal{A}(|C(u,v)|+1))$.

We first prove  a utility lemma. 
\begin{lemma}
\label{eq:Cempty}
Let $T_1$ and $T_2$ be two trees in $\mathcal{F}_{i-1}$,  $v_1 = \rho(T_1),$  $v_2=\rho(T_2)$ be their centroids and 
 $u_1 \in \mathcal{N}(v_1)$,  $u_2 \in \mathcal{N}(v_2)$.
 Then 
 \begin{equation}
 \label{eq:CC}
 C(u_u,v_1)\cap C(u_2,v_2) = \emptyset.
 \end{equation}
 In the statement of this lemma,  it is possible that $T_1 = T_2.$
 \end{lemma}
 
 \begin{proof} %[of Lemma \ref{lem:stage}]
 If $i=1$ the lemma is trivially true since $\mathcal{F}_{0}= \{T\}$  so  $v_1=v_2$ and thus
 $V_{-v_1}(u_1)$ and $V_{-v_1}(u_2)$ are disjoint.
 
 We therefore assume that $i >1$. We  also assume that $v_1 \not = v_2$ since otherwise   $V_{-v_1}(u_1)$ and $V_{-v_1}(u_2)$ are obviously  disjoint.
 %Now set $W_j = \bigcup_{i=1}^j L_i.$  
 We finally   assume that 
  \begin{equation}
  \label{eq:ass w}
  \mbox{There does not exist  $w  \in  W_{i-1}$ s.t.  
  $\forall z \in \mathcal{N}(w)$,  $f\left(V_{-w}(z) \cup \{w\}, w\right)\leq \mathcal{T}$. }
\end{equation}
This is because if such a $w$ existed then lines 12-14 in Algorithm \ref{alg:Centroid Processing} would have terminated the algorithm before the start of stage $i$.

 Observe that all nodes  in $V$  lie in $\Tin$ during stage  $0$ but by the end of stage $i-1$,  $T_1$ and $T_2$   are  disconnected.  By construction there must exist some node   $w \in  W_{i-1}$  (whose removal disconnected $T_1$ and $T_2$) that  lies on the  path  $\Pi(v_1,v_2)$ connecting $v_1$ and $v_2$.

 \begin{itemize}
 \item  (\ref{eq:ass w})  implies  $ \exists z\in\mathcal N(w)$ satisfying   $f(V_{-w}(z) \cup \{w\}, w) > \mathcal{T}$.  
 \item   If   $v_1 \not\in T_{-v_2}(u_2)$  and
 $v_2\not\in T_{-v_1}(u_1)$  then
  $V{-v_1}(u_1) \cap V{-v_1}(u_1)  = \emptyset $ 
  so  (\ref{eq:CC})  is  trivially true. (Fig.~\ref{lem:Cempty1})
  \item If  either of  $(u_1,v_1)$ or $(u_2,v_2)$ were  labelled {\bf L1} or {\bf L2} at the end of stage $i-1$ then (\ref{eq:CC})   would be trivially true.
 \end{itemize}

 \begin{figure}
 \centerline{\includegraphics[width=3in]{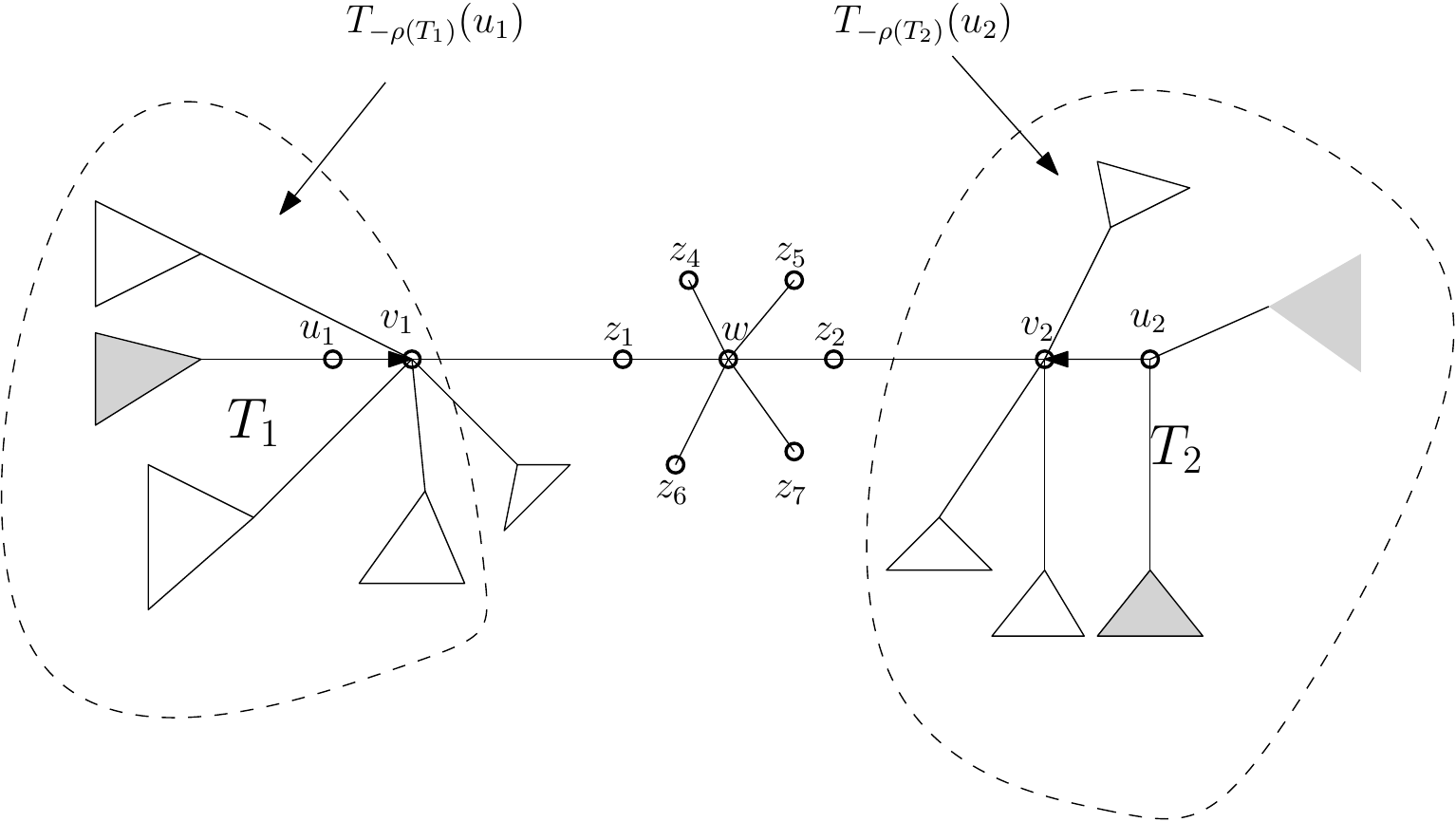}}
 \caption{Case in proof of Lemma \ref{eq:Cempty}  in which $v_1 \not\in T_{-v_2}(u_2)$  and
 $v_2\not\in T_{-v_1}(u_1)$.}
 \label{lem:Cempty1}
 \centerline{\includegraphics[width=3in]{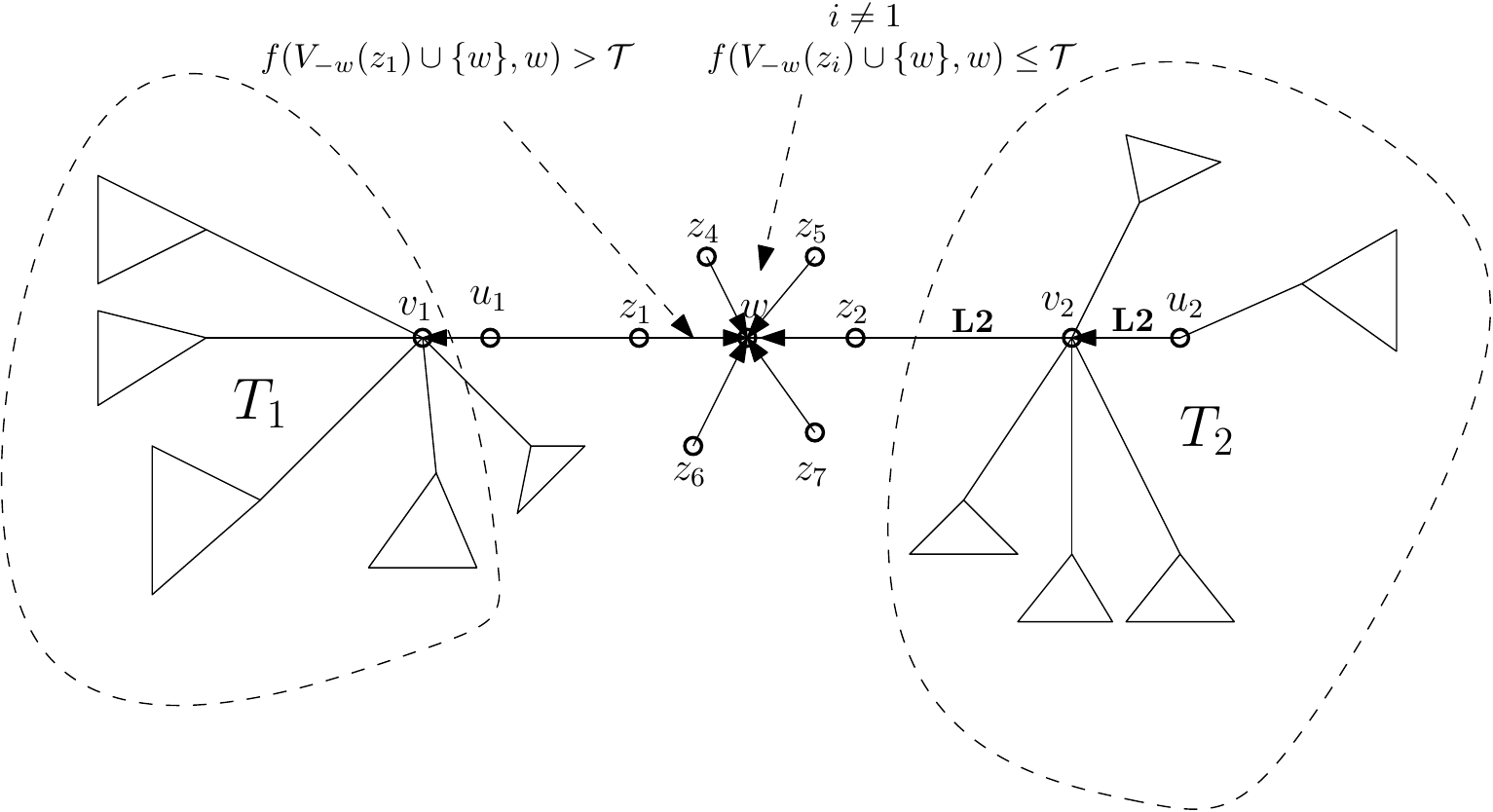}}
 \caption{Case in proof of Lemma \ref{eq:Cempty}  in which $v_1 \not\in T_{-v_2}(u_2)$  and
 $v_2 \in T_{-v_1}(u_1)$.}
 \label{lem:Cempty2}
 \end{figure}

If the Lemma is incorrect we may therefore assume that neither $(u_1,v_1)$ or $(u_2,v_2)$ were  labelled {\bf L1} or {\bf L2} at the end of stage $i-1$ 
and at least one of  $v_1 \in T_{-v_2}(u_2)$  or $v_2 \in T_{-v_1}(u_1)$  is true.  WLOG assume that $v_2 \in T_{-v_1}(u_1)$.

(Fig.~\ref {lem:Cempty2})
Label the neighbors of $w$ so that $z_1$ is on the path from $w$ to $v_1$,  $z_2$ is on the path from $w$ to $v_2$ and $z_3, \ldots, z_s$ are the others (if they exist).

 Note that, if, for any $j >1$,  $f(V_{-w}(z_j) \cup\{w\}, w) > \mathcal{T}$ then since $(u_1,v_1)$ is above  $(z_j,w)$,  $(u_1,v_1)$ would have been labelled {\bf L1} by the end of stage $i-1$  which we assumed was not the case.  Thus, for all $j > 1$,  $f(V_{-w}(z_j) \cup\{w\}, w) \le \mathcal{T}$.  (\ref{eq:ass w}) then implies
  $f(V_{-w}(z_1) \cup\{w\}, w) > \mathcal{T}$.

Next  note that if  $v_1 \in T_{-v_2}(u_2)$ the exact same argument would show that  for all $j \not=2$,  $f(V_{-w}(z_j) \cup\{w\}, w) \le \mathcal{T}$, otherwise  $(u_1,v_1)$  would have been labelled  
{\bf L1} by the end of stage $i-1$ , which we assumed was not the case.
In particular this  would imply $f(V_{-w}(z_1) \cup\{w\}, w) \le \mathcal{T}$, contradicting the result of the previous paragraph.

Thus  $v_2 \in T_{-v_1}(u_1)$   and $v_1 \not\in T_{-v_2}(u_2)$.  But this and the fact that $f(V_{-w}(z_2) \cup\{w\}, w) \le \mathcal{T}$ immediately imply that $(u_2,v_2)$ which is below  $(z_2,w)$ would have been labelled {\bf L2} by the end of stage $i-1$, contradicting  our assumptions.
% Figure \ref {lem:Cempty2}.
 \qed 
 \end{proof}
 %This lemma immediately implies that any node $v'\in V$ can appear in at most one set $C(u,v)$ in stage $i$, immediately proving 
 %Lemma \ref {lem:stage}.

We now prove
 \begin{lemma} \ 
 
 \begin{enumerate}
 \item  Algorithm  \ref{alg:Centroid Processing} works in  $O(\log n)$ stages
with the  specific oracle calls made during stage $i$ only dependent upon the results of the oracle calls made in stages $j <i$  and not on the results of any oracle calls during  stage $i.$
  \item   In each stage the total work performed by the oracle calls is  
 $O(t_{\mathcal{A}}(n)).$
 \item 
 The total amount of work performed by Algorithm \ref {alg:Centroid Processing} is
 %$$O\left(n \log n + t_{\mathcal{A}}(n)\, \log n \right).$$
$$O\left( t_{\mathcal{A}}(n)\, \log n \right).$$
 \end{enumerate}
 \label{lem:stage}
  \label{lem:peaking centroid}
 \end{lemma}
 \begin{proof}
 (1) is from the definition of the algorithm.
 
For (2) Let $E_i$ be the edges processed in stage $i$, i.e., $(u,v) \in E_i$ if $v \in L_i.$ By definition,  $E'_i \subseteq E_i$.
From Lemma \ref {eq:Cempty} no vertex $w \in V$ can appear in more than one set $C(u,v)$  for  $(u,v) \in E_i.$  Thus
 $$\sum_{(u,v) \in E'_i}  |C(u,v)|  \le \sum_{(u,v) \in E_i} |C(u,v)|  \le  n.$$
 So, by asymptotic subadditivity, the total amount of work done in stage $i$ 
in  line 8  will be
 \begin{eqnarray*}
 O\left(\sum_{(u,v) \in E_i} t_\mathcal{A}(|C(u,v)|+1)\right) &=&
 O\left(\sum_{(u,v) \in E_i}O\left( t_\mathcal{A}(|C(u,v)|)\right) \right)\\
& =&  O\left(t_{\mathcal{A}} (n)\right)
 \end{eqnarray*}
 proving (2).
 
 (3) then follows from the fact that 
$t = O(\log n)$, the remainder of the work in  Algorithm \ref{alg:Centroid Processing} outside of line 8 is $O(n)$   and only another 
 $O(n \log n)$ time is  required for the decomposition.
\qed
 \end{proof}

 \subsubsection{Later Peaking Phases  by  Binary Search.}
 \label{subsec:peaking bin}

\begin{algorithm}[t]
	\begin{algorithmic}[1]
		\State \hspace*{.001in}
                 \Comment {$(u,v)$ satisfies the reaching condition in $T_H(S)$ and $V_{-v}(u)$ has been removed \quad}\\
                 \Comment{Checks the scenarios from Lemma \ref{lem:hubevolution}. $u',h,v', v'',h'$ as defined in Lemma \ref{lem:hubevolution}.}
		\State \ 
		\Procedure {$Psearch$}{$v, b$}
		\State  {Create the path  $a = v_1, v_2,\ldots, v_t = b$ in $T_H(S)$}
                  \State{Binary search to find smallest   $i \in [1,t-1]$ such that $f(V_{-v_{i+1}}(v_{i})\cup \{v_{i+1}\}, v_{i+1}) >{ \mathcal T}$}
                  \State{$Return(v_i,v_{i+1})$}
		\EndProcedure
		\State \ 
		\State \ 
                 \If {$h \not= v$}
		\State {$ \alpha = f(V_{-h}(v') \cup\{h\},h)$}
	       \EndIf
	        \If {($h \not= v$ AND $\alpha > \mathcal T$)}
		   \Comment {Scenario 1}
		   \State   $(u',p(u')) =PSearch(v,r)$	
		   \State  { Commit $V_{-p(u')}(u')$ to $u'$}
		\ElsIf {($h =r$ AND $r$ has exactly two children in $T_{H(s)}$)}
		   \Comment {Scenario 2c}
		\State  {$\beta = f(V_{-h'}(v') \cup\{h'\},h') $}
			\If {($\beta > {\mathcal T}$)}
			   \Comment {Scenario 2ci}
				 \State $(u',p(u')) =PSearch(v,h')$	
		   		\State  {Commit $V_{-p(u')}(u')$ to $u'$}
		   		      \If {$h'$ is a sink   and  $S = \{u',h'\}$}
		   		         \Comment {Scenario 2civ, first part}
		   		         \State {Commit$(\{u'\},u')$ and Commit$(\{h'\},h')$ and terminate}
		   		     \EndIf
                     \ElsIf {$h'$ is a sink}
                     \Comment{Scenario 2ciii}
                     \State{Commit$(V,h')$ and terminate}
			\EndIf
	        \EndIf
	\end{algorithmic}
	\caption{Peaking Phase after Reaching Phase}
	\label{alg:PeakingAfter Reaching}
\end{algorithm}

All later peaking phases start immediately after a reaching phase has  completed by finding an edge  $(u,v)$ satisfying the  reaching criterion.

 %and a Peaking Phase was then called.
Lemma \ref{lem:hubevolution}  splits this into seven  scenarios. In  Cases 2ciii and the first  half of Case 2civ the algorithm terminates. In Cases 2a, 2b and 2ci  the resulting tree remains RC-viable and therefore the peaking phase can be skipped. In the remaining  Cases
   1, 2cii and the second half of  Case 2civ  it is known that only one new edge $(u,v)$ might now satisfy the peaking criterion and that edge is on  the path
  $\Pi(v,h)$ or $\Pi(v,h')$ (defined in the Lemma). 
By path and set monotonicity, if this edge exists, it can be found   by binary searching on the path.   This is formalized in Algorithm \ref{alg:PeakingAfter Reaching}

The procedure performs $O(n)$ book-keeping, $O(1)$ oracle calls  and, possibly, one binary search requiring  an additional
$O(\log n)$ oracle calls.  Thus
\begin{lemma} 
Each individual peaking phase after the first one can be implemented using only $O(\log n)$ oracle calls and 
 $O(t_\mathcal{A} (n) \, \log n )$ time.
\end{lemma}

\subsection{Creating and maintaining the hub tree.}
\label{subsec:new reaching bs}

\begin{algorithm}[t]
	\begin{algorithmic}[1]
		\State { }
                 \Comment     {The structure  $T_{H(S)}$  with root $r$ is given.  $s\in S$ is a known sink. \hspace*{.3in} }
%		\State \ 
		\Procedure {$Rsearch$}{$s, r$}
		\State  {Create the path  $s = v_1, v_2,\ldots, v_t = r$ in $T_H(S)$}
                  \State{Binary search to find largest   $i \in [1,t-1]$ such that $f(BP(v_i,s),s)  \le { \mathcal T}$}
                  \State{$Return(v_i)$}
		\EndProcedure
	
	\end{algorithmic}
	\caption{Binary Search in Reaching Phase}
	\label{alg:ReachingAfterPeaking}
\end{algorithm}

At the start of each reaching phase the algorithm must construct the appropriate hub tree.  This entails identifying an appropriate root $r$,  the hub nodes $V_{H}(S)$ and, for each hub node $v,$ pointers from  $v$  to its children and  to $p(v)$ and $p_H(v)$.  In addition,  the  sink set  $S(v)$ must be calculated for each node $v$.

At the completion of the first peaking phase,  the first hub tree must be built from scratch.  Everything except for the calculation of the  $S(v)$ can be easily done in $O(n)$ time.
Assume the hub-tree structure has been built and let $s \in S$ be any sink. $f(BP(u,s),s)$ is a non-decreasing  function as $u$ moves up the tree path  $\Pi(s,r)$.  Thus, 
 a binary search using 
 $O(\log n)$  oracle $\mathcal{A}$ calls  finds the highest node $u$ on $\Pi(s,r)$  satisfying  $f(BP(u,s), s) \le \mathcal{T}$.  This is shown in Algorithm \ref{alg:ReachingAfterPeaking} which uses
$O(t_\mathcal{A} (n) \, \log n )$ time.

After finding $u$ the algorithm walks up the path $\Pi(s,u)$ adding $s$ to every node on the path on this path.
  Since $|S| \le k$ this can be done using  a total of $O(k \log n)$ oracle calls and  $O(nk)$ extra time (for walking up all of the paths and creating the lists).
  For each $u,$   we maintain the list $S(u)$ of sinks  partitioned into sublists;  each sublist is associated with the hub child of $u$ that contains those sinks. Combining all of the above,  the time required for constructing the first hub tree is  $O(n k + k \log n\, t_{\mathcal{A}}(n))=   O (kt_{\mathcal{A}}(n)).$

At the start of every subsequent reaching phase, Lemma \ref{lem:hubevolution} shows that the hub tree could only have changed in a very constrained way from the previous hub tree.  After the closed commits of the previous reaching stage at most one new sink could have been  added (with a corresponding subtree removed) in the peaking stage. New edges are never added to the hub tree; once an edge is removed from the hub tree it never returns. The root can only  change in very restricted  circumstances. The structure of the new  hub tree  can easily be constructed from the old one in $O(n)$ time.  

After the new hub tree is built,  the  $S(v)$ lists need to be updated using  (\ref{eq:barbarSnew}) in  Lemma  \ref{lem:hubevolution}.  First, remove from $S(v)$ all  sinks that were committed in the last reaching phase. Since one sink can be removed in $O(n)$ time and at most $k$ sinks need to be removed this uses  $O(nk)$ time over the entire algorithm.
Finally, if a sink $s$ was created in the preceding peaking phase (according to Lemma  \ref{lem:hubevolution}  at most one such sink can be created) it needs to be added to the appropriate $S(v)$ lists.
  This can be done similarly as in the construction of the first hub tree, by calling $Rsearch(s,r)$, using  $O(n)$ time plus $O(\log n)$ calls to the oracle $\mathcal{A}$.  Since at most $k$ sinks can be added, the total work performed by the algorithm to create the hub tree at the start of each reaching phase {\em  taken over  the entire algorithm}  is $O(k t_{\mathcal{A}}(n) \, \log n)$ using $O (k \log n)$ oracle calls.

\subsection{Implementing the Reaching  Phase}
\label{subsec:reaching imp}
Assume that the hub tree is given along with the lists $S(u)$ for each node in the hub tree as introduced in Definition \ref{def: HT stuff}.
The self-sufficiency tests in Section \ref{subsec: ss testing}can now be restated  in terms of $S(u)$. % and $u_v(s)$.

\begin{lemma}
Let $v \in V_{H(S)}$ be  a non-hub node, $u$ its unique descendent in $T_{H(S)}$ and $T(v)$  the subtree of $T$ rooted at $v$.
% such that  $T_H(u)$ is recursively self sufficient.
If $T_{-v}(u)$ is recursively self-sufficient then one of the following two cases must occur:
 \begin{itemize}
\item[(i)]  $|S(v)| = 0$ and  $(u,v)$ satisfies the reaching criterion.
\item[(ii)]  $|S'(v)| > 0$ and $T(v)$ is recursively self-sufficient with every sink in $S(p(v))$ as a witness to its self sufficiency.
\end{itemize}
\label{lem: nr2}
\end{lemma}
\begin{proof}
This lemma is essentially a restatement of Corollary \ref{corollary: SStesting} rewritten for this special case in which $v$ only has the one child  $u$ in $T_{H(S)}.$   \qed
\end{proof}

\begin{lemma}
Let $v$ be  a non-sink hub in $T_{H(S)}$, $u_1,\ldots,u_t$ be its hub-children and all the $T_{-v}(u_i)$ are recursively self-sufficient.
Then one of the following two cases must occur:
\begin{itemize}
\item[(i)]
$\exists i,$  such that  $S(u_i) \cap S(v) = \emptyset$ \\
$\Rightarrow$  $(u_i,v)$ satisfies the reaching criterion.
\item[(ii)]
$\forall i,$  $S(u_i) \cap S(v) \not= \emptyset$ \\
 $\Rightarrow$ $T(v)$ is recursively self-sufficient with  every sink in $S(v)$ as a witness to its self sufficiency.
\end{itemize}
\label{lem: nr3}
\end{lemma}
\begin{proof}
This lemma is essentially a restatement of Corollary \ref{corollary: SStesting} rewritten for this special case when $v$ has more than one child in  $T_{H(S)}.$ \qed
\end{proof}

 \begin{algorithm}[h]
	\begin{algorithmic}[1]
%	\State{\%Before starting a binary search was performed for each $s \in S.$  }
%	\State{\% In total $O(\log n)$ oracle calls they found, for each $s$, the highest $v \in \Pi(r,s)$}
%	\State{\% such that $f(BP(v,s) \le \mathcal{T}.$ This permitted constructing al of the $S(u).$}
	\State{\%$ v_i$  are  topologically sorted so that if $v_i$ is the child  of $v_j$  then $i < j$.  $t=|V_{H(S)}|.$}
         \State{\% If line 26 is reached without Break Out, then $T$ is self-sufficient }
	\State{}
	\For {$i=1$ to $t$}
	 \State {$v := v_i$}
	 \If {$v$ is a non-sink hub}	
	  \Comment{Apply Lemma \ref{lem: nr3}}
	 	\For {$u$ a child of $v$ in $T_H(S)}$
	 		\If {$S(u) \cap S(v) = \emptyset$}
	 		\Comment{$(u,v)$ satisfies Reaching Criterion}
	 		\State {Remove $V(u)$ from $T$}
		       \State {Commit blocks for $V(u)$ to sinks in $V(u) \cap S$ using Lemma \ref{lemma:RCTrimTree}}
	 		\State {Break out of Procedure}
	 		\EndIf
	 	\EndFor
	 	\State{$S(v)$ are witnesses to recursive self-sufficiency of $T(v)$}
	\Else
		  \Comment{Apply Lemma \ref{lem: nr2}}
		      \State{Set $u$ to be the unique hub-child of $v$}
	 	       \If {$|S(u)| > 0$ and $|S(v)| =0$}
	 		\Comment{$(u,v)$ satisfies Reaching Criterion}
	 		 		\State {Remove $V(u)$ from $T$}
		       \State {Commit blocks for $V(u)$ to sinks in $V(u) \cap S$ using Lemma \ref{lemma:RCTrimTree}}
	 		\State {Break out of Procedure}
	 		\Else
			\Comment{ $T(v)$ is recursively  self-sufficient }
	 		\State{$S(v)$ are witnesses to recursive self-sufficiency of $T(v)$}
	 	\EndIf
	 \EndIf
	\EndFor
         \State{\% Entire $T$ is recursively  self-sufficient}
	 		 \State {Commit all of $T$ to $S$ and terminate algorithm.}	 		 		
	\end{algorithmic}
	\caption{Reaching Stage}
	\label{alg:Binary Reaching}
\end{algorithm}

Lemmas  \ref{lem: nr2}   and \ref{lem: nr3} permit   implementing a Reaching Phase in $O(n)$ time as shown in Algorithm \ref{alg:Binary Reaching}.  

First, in $O(|V_{H(S)}|) = O(n)$ time, preprocess the nodes in $V_{H(S)}$  by  topologically sorting them so that if $v_i$ is the child of $v_j$ then $i < j$.

Next, process the  nodes in $V_{H(S)}$  in this topological order.   This will ensure that a node will  be processed only after its hub-children have already been processed. 
By induction, after a node $v$ has been processed, if the algorithm hasn't halted, $T(v)$ will be recursively self sufficient.

Processing  a non-sink hub node $v$ uses   Lemma \ref{lem: nr3} to check if any of the edges leading to $h_i$ satisfy the reaching criterion.  If yes, the algorithm commits the proper nodes to sinks in $O(n)$ time and exits.  Otherwise the tree rooted at $v$ will be recursively-self sufficient and the algorithm  continues.

Processing a non-hub node $v$    uses Lemma \ref{lem: nr2} to check in $O(1)$ time if $(u,v)$ satisfies the reaching criterion, where $u$ is $v$'s unique hub child. If yes, the algorithm commits the proper nodes to sinks in $O(n)$ time and exits.  Otherwise $v$ will be recursively self sufficient and the algorithm continues.

If the algorithm completes the entire For loop and reaches line 27 then  the entire tree $T$ is recursively self-sufficient so $T$ can be fully committed to $S$ and the algorithm terminates.

Lines 8 and 17  can be implemented in $O(1)$  time because of the way the lists were stored. Lines 14 and 22 can also be implemented in $O(1)$ time since it is only necessary to set a flag stating that the entire list $S(v)$ are witnesses.

We have therefore just proven 
\begin{lemma}
If the hub tree is already given then the  reaching phase can be implemented in $O(n)$ time.
\end{lemma}

\subsection{Combining the Pieces}
\label{subsec:comb bounded}
This section has shown how to implement the entire bounded cost algorithm.  It follows the generic structure of 
Algorithm \ref{alg:BoundedCostFull2}, alternating Peaking and Reaching Phases.  

The actual work was done by five  logically distinct parts listed  below.  This decomposition will permit the parametric search extension in the next section.
\begin{enumerate}
\item {\bf The First Peaking Phase}
\begin{itemize}
\item Implemented using  tree centroid decomposition method of Section \ref{subsec:centroid}.
\item Divided into $O(\log n)$ stages.
  Each stage performs $O(n)$ extra work plus   one  {\em amortized} Oracle call.
 \item Total time required $O\left( \log n\, t_\mathcal{A} (n))\right)$.\\
 Number of actual oracle calls made could be as high as $\Theta(n).$
  \end{itemize}
\item  {\bf Creating the   First Hub tree}
\begin{itemize}
\item  lmplemented using binary search method of  Section \ref{subsec:new reaching bs}
\item Total time required  $O(k\log n t_\mathcal{A} (n))$
\item Uses $O(k \log n)$ total oracle calls.
\end{itemize}
\item {\bf Reaching Phases}
\begin{itemize}
\item  Implemented  using Algorithm  \ref  {alg:Binary Reaching}  % using hub trees and binary searching method of Section \ref{subsec:reaching imp}.
\item  Uses $O(n)$ time with no oracle calls per each reaching phase\\
Assumes   pre-existing hub tree with  preconstructed lists lists $S(u)$
\item  At most $k$ reaching stages;  $O(nk)$ total time for all reaching stages
\end{itemize}
\item {\bf All other Peaking Phases}
\begin{itemize}
\item Implemented using binary search method of  Section \ref{subsec:peaking bin}
\item Uses   $O( \log n t_\mathcal{A} (n))$ time with $O( \log n)$ oracle  calls per peaking stage
\item  At most $k$ reaching stages;    $O( k \log n t_\mathcal{A} (n))$   total time for all reaching stages
\end{itemize}
\item {\bf Creating Hub Tree after Non-Initial Peaking phase}
\begin{itemize}
\item   lmplemented using binary search method of  Section \ref{subsec:new reaching bs}
\item Removes $k'$ sinks from old hub tree and adds at most one new sink
\item  Total time required $O(nk' +  \log n t_{\mathcal{A}}(n))$ using    $O( \log n)$ oracle calls
\item At most $k$ peaking stages;    $O( k \log n t_\mathcal{A} (n))$   total time for all peaking stages
\end{itemize}
\end{enumerate}

Combining these parts proves Theorem \ref{theorem:FastBC}.
For later use we denote this complete algorithm for solving the 
 bounded-cost minmax $k$-sink problem as $\mathcal{B}$ and its running time on an input of  size $n$ as $B(n).$

\section{Full Problem: Cost Minimization via Parametric Searching}
\label{Section: Full Problem}

By  binary searching over all possible values of $\mathcal{T}$  and using $\mathcal{B}$ to test the feasibility of these  $\mathcal{T}$, 
 it is straightforward to construct a \emph{weakly} polynomial time algorithm for the general minmax $k$-sink problem of finding  $\mathcal{T}^*$, the smallest $\mathcal{T}$ for which $k$ sinks suffice.

Modifying $\mathcal{B}$
to produce a \emph{strongly} polynomial time algorithm, as in Theorem \ref{theorem:FastC}, though, will require %modifying  it 
using  a variation on  Megiddo's {\em parametric searching} technique
\cite{megiddo1979combinatorial}. 

\begin{definition}
The {\em State} of algorithm $\mathcal{B}$ at any given time will be the current  $(\Sout,\Pout)$, the  edge labels in the first peaking phase and the $S(v)$ values in the hub tree.
\end{definition}
Note that all of the information saved by  $\mathcal{B}$, i.e.,  
$T$ and $S$  and the rest of the hub tree information, can be directly constructed from  its state.  Thus if two invocations of $\mathcal{B}$ on two different values
$\mathcal{T}'$ and  $\mathcal{T}''$ both stop mid-calculation in the same state there is no way to distinguish between them.

% in the lists $S(v)$ for $v \in V.$  

In the
 parametric search version,  $\mathcal{T}$  will no longer be a constant; instead we {\em interfere} with the normal course of   $\mathcal{B}$  by changing $\mathcal{T}$ during runtime.  

This interfered version is denoted by Algorithm  $\mathcal I$.  The decision to interfere is based on a  \emph{threshold range}
 $(\mathcal{T}^L, \mathcal{T}^H]$. % that maintains allowable candidate values of $\mathcal{T}^*$. 
 $\mathcal I$ starts with  $(\mathcal{T}^L, \mathcal{T}^H] = (0,+\infty]$ and always maintains the following invariants:
\begin{enumerate}[label=(\textrm{I}\arabic{*})]
\item    $\mathcal{T}^L <\mathcal{T}^H$. 
\item $\mathcal{T}^L$ never decreases and $\mathcal{T}^H$ never increases.
\item $\mathcal{T}^L$ will be  infeasible and $\mathcal{T}^H$  will be   feasible. 
\item At each step of $\mathcal I$,  the corresponding  state of $\mathcal B$ would  be identical  for ALL
values of  $\mathcal{T} \in  [\mathcal{T}^L, \mathcal{T}^H)$. (Note the flipping of open and closed intervals.)
\end{enumerate}

Intuitively, $\mathcal I$ ``pretends'' that is it running $\mathcal{B}$  for all $\mathcal{T} \in  [\mathcal{T}^L, \mathcal{T}^H)$ while pruning away ``useless values''. We will soon see that the properties above will imply that $\mathcal I$ terminates with the value $ \mathcal{T}^H$ being  the correct solution.

\medskip

This leads to defining  a {\em step} in $\mathcal I.$   There will be two types of steps,
%\footnote{Technically, the reason Megiddo's parametric searching technique \cite{megiddo1979combinatorial} can't be efficiently applied directly is  because of the existence of these two types of steps.  The method described here uses only $O(\log n)$  (amortized) extra calls.},
{\em Stage-Steps} and  {\em If-Steps}.

The Stage-Steps will correspond to the stages in the first peaking phase.  The If-Steps will correspond to an oracle call and associated work performed AFTER the first peaking phase
\medskip

\par\noindent\underline{\bf Stage-Steps:}\\
$\mathcal{B}$ starts by  implementing  the first peaking phase using Algorithm  \ref{alg:Centroid Processing}. This is  divided into $t=O(\log n)$ stages, where $t$ is the number of levels in the centroid decomposition of $\Tin.$  Recall that decomposition itself only depends upon $\Tin$ and not $\mathcal T.$  
From Lemma \ref {lem:peaking centroid} (1),   the  full set of oracle calls made during   each of those stages  depends upon the results of the calls from {\em previous} stages and not on any calls  made during  the current one. 
%More specifically, the oracle calls made within a stage are all made BEFORE acting upon the results of those calls.

A Stage-Step will correspond to the implementation of one stage as performed by lines 3-22 of Algorithm  \ref{alg:Centroid Processing}.
 AFTER making the oracle calls in lines 3-11,  the Stage-Step will binary search among the returned values  to find  $a^L$, the largest infeasible value and
 $a^H$, the smallest  feasible one.  It will then set ${\mathcal T}^L$ to be the larger of $a^L$ and the old  ${\mathcal T}^L$  and run the remaining lines  12-22 using  $T ={\mathcal T}^L.$  It will also set
 ${\mathcal T}^H$ to be the smaller of $a^H$ and the old  ${\mathcal T}^H$. 
 The details are in 
Figure \ref{fig:Stage Step}.

\medskip

\begin{lemma}
\label{lem:Stage Thresh}
Let $i \le t$ where $t$ is the number of stages in  the centroid decomposition of $\Tin.$ 
Let $\mathcal{T}^*$ be the optimal value of $\mathcal T$ and 
let $(\mathcal{T}^L, \mathcal{T}^H]$ be the threshold range after running $i$ Stage-Steps.
% Algorithm $\mathcal I$.  
Then

\begin{enumerate}
\item  $\mathcal{T}^* \in (\mathcal{T}^L, \mathcal{T}^H]$.
\item Let  $\mathcal{T}'  \in [\mathcal{T}^L, \mathcal{T}^H)$.  %{\em \small (Note the half-openness of the intervals has been reversed.) }\\
Then,  Algorithm  $\mathcal B$ run for $i$ stages with ${\mathcal T} = \mathcal{T}' $ would end in the same state as algorithm  $\mathcal B$ run for $i$ stages with ${\mathcal T} = \mathcal{T}^L$.
\item Algorithm  $\mathcal I$ run for $i$ Stage steps  would end in the same state as algorithm  $\mathcal B$ run for $i$ stages with ${\mathcal T} = \mathcal{T}^L$.
\end{enumerate}
\end{lemma}
\begin{proof}
	
(1)   $\mathcal{T}^* \leq \mathcal{T}_{H}$ because $\mathcal{T}^H$ is always set to be a feasible value of $\mathcal{T}$. Similarly, $\mathcal{T}^* > \mathcal{T}_{L}$ because $\mathcal{T}^L$ is always set to be a non-feasible value of $\mathcal{T}$.
	
(2) Will be proven  by induction on $i.$    Let stage $0$ be the starting process  of setting
	$(\mathcal{T}^L, \mathcal{T}^H) = (0, \infty).$  Then (2) is valid for $i =0.$
	
	Now assume (2) is correct for $i-1$.   
	%Note that during stage $i$, $\mathcal{T}^L$ never decreases and   $\mathcal{T}^H$ never increases.  
	From the induction hypothesis,   the state  of $\mathcal B$ for
${\mathcal T} = \mathcal{T}^L$ and ${\mathcal T} = \mathcal{T}'$ will be identical at the end of stage $i-1$, i.e., the start of stage $i$.  In particular,  for both
${\mathcal T} = \mathcal{T}^L,\mathcal{T}'$, 
all edges will be labelled  the same at the start of stage $i.$
	
	Next note that, after  completing  Line 4 in the  $i$'th Stage-Step   there does not exist any $a_j \in \mathcal{V}$ satisfying  $a_j \in (\mathcal{T}^L, \mathcal{T}^H).$  Thus, if
$\mathcal{T}' \in [\mathcal{T}^L, \mathcal{T}^H),$\\
$$a_j \le \mathcal{T}^L \ \Rightarrow \  a_i \le \mathcal{T}'
\quad\mbox{and } \quad
a_j > \mathcal{T}^L \  \Rightarrow \ a_j \ge \mathcal{T}^H \  \Rightarrow \ a_j > \mathcal{T}'.$$

Thus, for any $a_j$ evaluated during stage $i$ via an oracle call,  Algorithm $\mathcal B$  can not distinguish between  the answer to $a_j \le \mathcal{T}^L$ and $a_j \le \mathcal{T}'$.   Since the decisions made by 
$\mathcal B$ only depend upon the prior labels of edges and the results of the $a_j \le \mathcal{T}$ queries, 
$\mathcal B$ will behave identically for both
${\mathcal T} = \mathcal{T}^L,\mathcal{T}'$.

\medskip

		 (3) follows from the analysis of (2).		\qed
\end{proof}

\par\noindent\underline{\bf  If-Steps:}\\
 The remainder of of $\mathcal B$  will be divided into  {\bf  If-Steps}. The first If-Step starts right after the first peaking phase concludes.  All subsequent  If-Steps start whenever an oracle call  is made.
 
 Note that  $\mathcal B$  works by making oracle calls of the type $a = f( \cdot, \cdot )$ followed by a clause
{\em ``If  $a \le \mathcal{T}$''},  e.g., during the peaking phases or creation of a hub tree. Further note that $\mathcal B$ never actually uses the {\em value} of $a$ or $\mathcal T$ when deciding what to do next.   Its actions only depend upon whether $a \le \mathcal{T}$  or $a >\mathcal{T}$.
 
 \medskip
 
  $\mathcal{I}$, the interfered version of Algorithm $\mathcal B$,  will run the remainder of  $\mathcal B$ If-Step by If-Step but replacing each If-Step in $\mathcal B$ by the corresponding interfered If-Step in $\mathcal{I}$ as defined in Figure \ref{fig:If Step}.

\begin{figure}[t]
%\par\noindent\underline{One step of the interfered algorithm:}

\par\noindent\underline {Stage-Step:}
\begin{enumerate}
 \item  Start by performing all of the oracle calls required by lines 3-11 of  Algorithm  \ref{alg:Centroid Processing}\\
 {\em By Lemma \ref  {lem:peaking centroid} (2), this only requires 
 $O(t_{\mathcal{A}}(n))$ time.}
%\item  Let $\mathcal{V}= \{\mathcal{T}_1,\, \mathcal{T}_2,\, \ldots,\mathcal{T}_r\}$ be the values returned by the oracle calls.
\item  Let $\mathcal{V}= \{a_1,\, a_2,\, \ldots,a_r\}$ be the values returned by the oracle calls.\\
{\em Note that feasibility of $a_i$ can be tested in  $B(n)$ time by running $\mathcal{B}$ with 
$\mathcal{T} = a_i.$}
 \item In $O(r \log r + B(n) \log r ) = O(B(n)  \log n ) $ time, binary search in $\mathcal{V}$ 
 for a pair of values  
 %$$a^L = \max\{ \mathcal{T}_i \in \mathcal{V} :\, \mathcal{T}_i \mbox{ is  not feasible}\},\quad
$$a^L = \max\{ a_i \in \mathcal{V} :\, a_i \mbox{ is  not feasible}\},\quad
 %a^H = \min\{ \mathcal{T}_i \in \mathcal{V} :\, \mathcal{T}_i \mbox{ is feasible}\}
 a^H = \min\{ a_i \in \mathcal{V} :\, a_i \mbox{ is feasible}\}
 $$
 If all of the $\mathcal{T}_i$ are feasible set $a^L = \mathcal{T}^L$; \\
 If  all of the $\mathcal{T}_i$ are not feasible set $a^H = \mathcal{T}^H$.
 \item Set $\mathcal{T}^L = \max\{a^L, \mathcal{T}^L\}$, 
 $\mathcal{T}^H = \min\{a^H, \mathcal{T}^H\}$  
 \item  Continue the stage using the value $\mathcal{T} = \mathcal{T}^L$  (the new updated value) when running lines 12-22 in Algorithm \ref{alg:Centroid Processing}.
 \end{enumerate}

\caption{A Stage-Step in Algorithm $\mathcal I$:}
\label{fig:Stage Step}
\end{figure}

\begin{figure}[t]
\par\noindent\underline {If-Step:}

\begin{enumerate}
\item Perform the evaluation $a:= f(\cdot,\cdot).$
\item Resolve the If-Clause and reset $\mathcal{T}^L, \mathcal{T}^H$ if necessary as follows:  %      Set $\mathcal{T}^L, \mathcal{T}^H$ as follows:
 	\begin{itemize}
 	\item (i) If $a \leq \mathcal{T}^L$,   resolve the associated If-clause as $f(\cdot,\cdot) \leq \mathcal{T}$.
 		\item(ii)  If $a \ge \mathcal{T}^H$, resolve the associated If-clause  as $f(\cdot,\cdot) > \mathcal{T}$.
 		\item If $a \in (\mathcal{T}^L, \mathcal{T}^H)$, run a separate version of $\mathcal B$ from scratch with value $\mathcal{T} := a$, and observe the output.
 \begin{itemize}
  \item (iii) If Output is `No': set $\mathcal{T}^L := a$.\\  Resolve the associated If-clause as $f(\cdot,\cdot) \leq \mathcal{T}$. %(because  $f(\cdot,\cdot) = a = \mathcal{T} ).$
  \item (iv)  If Output is `Yes', set $\mathcal{T}^H := a$.\\  Resolve the  associated  If-clause as $ f(\cdot,\cdot)  >\mathcal{T}$.
 \end{itemize}
 \end{itemize}
 \item Conclude the step by running the algorithm with the set value ${\mathcal T} = {\mathcal T}^L$ until the start of the next step.
\end{enumerate}
 \caption{An If-Step in Algorithm $\mathcal I$}
\label{fig:If Step}
\end{figure}
\medskip

\begin{lemma}
\label{lem:If Thresh}
Let $\mathcal{T}^*$ be the optimal value of $\mathcal T$ and 
let $(\mathcal{T}^L, \mathcal{T}^H]$ be the threshold range after running $m$ If-Steps after the conclusion of the first peaking phase.
% Algorithm $\mathcal I$.  
Then

\begin{enumerate}
\item  $\mathcal{T}^* \in (\mathcal{T}^L, \mathcal{T}^H]$.
\item Let  $\mathcal{T}'  \in [\mathcal{T}^L, \mathcal{T}^H)$. 
% {\em \small (Note the half-openness of the intervals has been reversed.) }\\
Then,  Algorithm  $\mathcal B$ run for $m$ steps with ${\mathcal T} = \mathcal{T}' $ would end in the same state as algorithm  $\mathcal B$ run for $m$ steps with ${\mathcal T} = \mathcal{T}^L$.
\item Algorithm  $\mathcal I$ run for $m$ steps  would end in the same state as algorithm  $\mathcal B$ run for $m$ steps with ${\mathcal T} = \mathcal{T}^L$.
\end{enumerate}
\end{lemma}
\begin{proof}
	
 (1) Same as  the proof of Lemma \ref{lem:Stage Thresh} (1).
 %, note that $\mathcal{T}^* \leq \mathcal{T}_{H}$ because $\mathcal{T}^H$ is always set to be a feasible value of $\mathcal{T}$. Similarly, $\mathcal{T}^* > \mathcal{T}_{L}$ because $\mathcal{T}^L$ is always set to be a non-feasible value of $\mathcal{T}$.
	
(2) 	From the induction hypothesis,   the state  of $\mathcal B$ for
${\mathcal T} = \mathcal{T}^L$ and ${\mathcal T} = \mathcal{T}'$ will be identical at the end of stage $m-1$, i.e., the start of stage $m$.
 Now consider what happens in the $m$'th If-Step. 
\medskip

	In case (i), $f(\cdot,\cdot) = a \le \mathcal{T}^L <  \mathcal{T}'$.
	
	In case (ii), $f(\cdot,\cdot)=a \ge \mathcal{T}^H> \mathcal{T}' \ge \mathcal{T}^L$.
	
	In case (iii), $f(\cdot,\cdot)= a = \mathcal{T}^L \le \mathcal{T}'$.
	
	In case (iv), $f(\cdot,\cdot)= a = \mathcal{T}^H > \mathcal{T}' \ge \mathcal{T}^L$.
	
	\medskip
	
	Thus, in all four cases,  the query ``$a \le \mathcal{T}^L?$'' resolves identically to the query
	``$a \le \mathcal{T}'$ and algorithm $\mathcal B$ can not distinguish between the two cases. Since the algorithm started the If-Step in an identical state for both
	${\mathcal T} = {\mathcal T}'$ and  ${\mathcal T} = {\mathcal T} ^L$ and can not distinguish between them during the If-Step,  it ends in the same state for both of them.

\medskip

(3)  follows directly from the analysis of (2).		\qed
\end{proof}

\begin{lemma}
\label{lem:inter1}\ 

The interfered algorithm $\mathcal{I}$ will terminate in  $O( \max(k,\log n) k \log^2 n t_{\mathcal{A}}(n))$ time. 

 Let $(\mathcal{T}_{<},\mathcal{T}_{>}]$ be the threshold range  when  $\mathcal I$  terminates and 
$\mathcal{T^*}$ be the optimal value of $\mathcal T$.
  Then $\mathcal{T^*}=\mathcal{T}_{>}$. In particular, we can then run the bounded cost Algorithm $\mathcal B$   on $\mathcal{T} := \mathcal{T}_{>}$ to retrieve the optimal feasible configuration.
 \label{lemma:UpperThresholdIsAnwser}
\end{lemma}
\begin{proof}
%Algorithm $\mathcal I$ starts with Stage-Steps and then follows with If-Steps.  
 By the definition of the Stage and If-Steps the open interval $(\mathcal{T}^L, \mathcal{T}^H)$ may  contract but is always  non-empty.  Let $\mathcal{T'}$ be any value falling in the intersection of all such intervals.

Point  (2)  of Lemmas \ref{lem:Stage Thresh}   and  \ref{lem:If Thresh}  imply that the number and type of  steps  run by algorithm $\mathcal I$ is exactly the same as those run by algorithm $\mathcal B$ on   $\mathcal{T'}$.  

For  a problem of size $n$, let  $S_S(n)$ denote  the maximum number of Stage-Steps run by Algorithm  $\mathcal B$ and  $S_I(n)$ the maximum number of If-Steps run.   As noted in Section \ref{subsec:comb bounded} $S_S(n)= O(\log n)$ and $S_I(n) = O(k \log n)$.  

The running time of $\mathcal I$ is the time  for running $\mathcal B$ on $\mathcal{T}'$ plus the work done on lines 3-4 of the Stage Steps  (Fig \ref{fig:Stage Step})  and line 2 of the If-Steps 
(Fig \ref{fig:If Step}).

Let  $W_S(n)$ be the total amount of work performed by  $\mathcal I$ on lines 3-4 in one  Stage-Step and $W_I(n)$ the total amount of work performed by $\mathcal I$ in line 2 of one If-Step.   $W_S(n) = O( \log n B(n))$ and  $W_I(n) = O(B(n)).$

Thus, the total amount of work performed by Algorithm $\mathcal I$ before it terminates is 
\begin{eqnarray*}
T(n) &=& B(n)  + S_S(n) W_S(n)  + S_F(n) W_F(n)\\
	   &=& B(n)   + O\bigl(\log^2 n B(n)\bigr) + O\bigl(k \log n B(n)\bigr)\\
%	   &=& O(\log^2 n + k \log n) B(n)\\
	   &=&  O\left(\bigl(\log^2 n + k \log n\bigr)(k  t_{\mathcal{A}}(n)  \log n )\right)\\
	   &=&  O( \max(k,\log n)\,  k  t_{\mathcal{A}}(n) \log^2 n)
\end{eqnarray*}

Now consider the step  at which   ${\mathcal T}^L$  was set to $\mathcal{T}_{<}$ while running $\mathcal I$ . This occurred because  ${\mathcal T}^L=\mathcal{T}_{<}$ was found to be infeasible.    Consider  $\mathcal{T'} \in (\mathcal{T}_{<},\mathcal{T}_{>})$.
From Lemmas \ref{lem:Stage Thresh}  and \ref{lem:If Thresh} (2),  $\mathcal B$ would be in exactly the same state at that end of the algorithm  for both ${\mathcal T}^{<}$  and ${\mathcal T}'$. Thus  ${\mathcal T}'$ is infeasible as well.

Finally, from Lemmas \ref{lem:Stage Thresh}  and \ref{lem:If Thresh} (1) we have that $\mathcal{T}^* \in (\mathcal{T}_{<}, \mathcal{T}_{>}]$ but we have just seen that no $\mathcal{T'} \in (\mathcal{T}_{<},\mathcal{T}_{>})$ is feasible,  Thus
$\mathcal{T} := \mathcal{T}_{>}$.
\qed
\end{proof}

Theorem \ref{theorem:FastC} follows immediately from the previous Lemma.

{\em Note:
A classic application of parametric search to $\mathcal B$ would require a call to $\mathcal{B}$ every time the oracle $\mathcal{A}$ was called. This  first peaking phase can  require as many as $\Theta (n)$ oracle calls, resulting in  an 
$\Theta (n B(n))= \Theta (n t_\mathcal{A}(n))$ running time for that phase in the parametric search version.  The use of the centroid decomposition and distinction between Stage and If Steps were necessary to replace this extra factor of $\Theta(n)$ by $\Theta(\log n).$ % down to $\Theta (  \log^2 n B(n)).$ 
}

\section{The Continuous Case}
\label{Sec:Continuous}

Until this point the analysis  has always assumed the discrete version of the problem   in 
which sinks are required to be  % $S \subseteq V$,  
nodes in $\Vin.$
This section will extend those results to the continuous case in which sinks can be located on edges.

This first requires extending the definition of   minmax monotone  cost functions to  edges.

\begin{figure}
	\centering
	\includegraphics[width=0.5\textwidth]{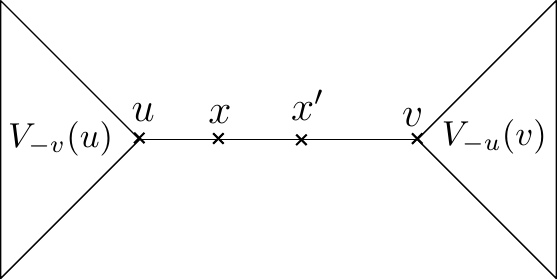}
	\caption{Let $(u,v)$ be oriented so that is starts at $u$ and ends at $v$.  Then $x < x'$. If the edge was oriented as ($v,u)$ then $x'< x'$  If $x \le x'$ then  $f(V_{-v}(u) \cup \{x\}, x) \le  f (V_{-v}(u)\cup \{x'\},x')$.}
		\label{fig:continuious}
\end{figure}

\begin{definition} (Fig.~\ref{fig:continuious})
\label{def:cont}
Let $T=(V,E)$ be  a tree and $f(\cdot,\cdot)$ a  monotone minmax cost function as defined in Section \ref{subsec:minmaxdef}.

For  $e=(u,v) \in E$,  orient $e$ so that it starts at $u$  and ends at $v$.  Let $x,x' \in e$.  Denote
\begin{eqnarray*}
x \le x'   &\quad \mbox{ if and only if} \quad  & \mbox{$x$ is on the path from $u$ to $x'$ },\\
x < x'   &\quad \mbox{ if and only if} \quad  & \mbox{ $x \le x'$  and  $x \not= x'.$}
\end{eqnarray*}
 $f(\cdot,\cdot)$ is  {\em continuous} if it satisfies:
\begin{enumerate}
\item  $f(V_{-v}(u) \cup \{x\}, x)$ is a continuous function in $\{x \,:\, u < x \le v\}.$
\item  $f (V_{-v}(u)\cup \{x\},x)$ is non-decreasing in $\{x \,:\, u \le x \le v\}$, i.e.,
$$\forall u \le x < x' \le v,\    f(V_{-v}(u) \cup \{x\}, x) \le  f (V_{-v}(u)\cup \{x'\},x') .$$
\end{enumerate}
\end{definition}
Point 2 is the natural generalization of path-monotonicity.  

\medskip
{\em \small Note:  
This definition  is satisfied in the sink evacuation problem.  Let $d(x,v)$ denote  the time required to travel from $x$ to $v$. It is natural to assume that this is non-increasing continuous function in $x$. Since flow travels smoothly without congestion {\em  inside} an edge, if the last flow  arrived at node $v$ at time $t$, then  it had been at  $x>u$ at time  $t-d(x,v)$. Thus 
\begin{equation}
\label{eq:dcontdef}
f(V_{-v}(u) \cup \{x\},x) = f(V_{-v}(u) \cup \{v\},v) - d(x,v)
\end{equation}
so condition (1) is satisfied and condition (2) is satisfied for every $x$ except possibly $x=u.$
Now consider the time $t'$ that the last flow arrives at node $u$ and  let $t'+w$  be the time that this last flow {\em  enters} edge $(u,v)$.
Since flow doesn't encounter congestion inside  an edge, it arrives at $v$ at time   $t'+w + d(u,v).$ Then 
$$f(V_{-v}(u),u) = t'  \le  t'+w =   (t'+w + d(u,v)) - d(u,v) =  \lim_{x \downarrow u}f(V_{-v}(u) \cup \{x\},x).$$ 
Thus condition (2) is also satisfied at $x=u.$  Note that $w >0$ only occurs if there is congestion at $(u,v)$ and this forces a left discontinuity, which is why 
 the range in point (1) does  not include $x=u.$  
}

The following lemma follows easily from the definitions and the continuity.
% function in $x$ and   $f(V_{-u}(v) \cup \{x\},x) $   is  a continuous monotonically decreasing one the following lemma is obvious

\begin{lemma}
\label{lem:coneq}
Let $T=(V,E)$ be  a tree,  $f(\cdot,\cdot)$ a  minmax monotone  cost function and $e=(u,v) \in E.$ 
% Assume that  $f(V_{-v}(u) \cup \{x\},x)$  and  $ f(V_{-u}(v) \cup \{x\},x) $ are not both identically zero.
Then both
\begin{equation}
\label{eq:st}
s_{\mathcal{T}} = \max_{x \in e} \Bigl( f(V_{-v}(u)\cup \{x\},x) \le \mathcal{T}\Bigr)
\end{equation}
and
\begin{equation}
\label{eq:alpha}
a:=\min_{x\in e} \max\Bigl(f(V_{-v}(u) \cup \{x\},x), f(V_{-u}(v) \cup \{x\},x)  \Bigr)
\end{equation}
exist.  
\end{lemma}

%We assume that $x'$ and $x_{\mathcal{T}}$ can also be calculated using   $O(1)$ oracle calls, i.e., in $O(t_{\mathcal{A}}(n))$ time where $n = |V|.$  
We finally assume that  $S_{\mathcal{T}}$ and
the largest $x' \in e$ for which  
$$s_{\mathcal{T}} =   f(V_{-v}(u)\cup \{x'\},x') ),$$
as well as
 $a$ and  $x' \in e$ for which
 $$a=\max\Bigl(f(V_{-v}(u) \cup \{x'\},x'), f(V_{-u}(v) \cup \{x'\},x')  \Bigr),$$
%and the associated $x$ value  at which it is minimized  % in (\ref{eq:st}) and (\ref{eq:alpha})   
can be calculated using   $O(1)$ oracle calls, i.e., in $O(t_{\mathcal{A}}(n))$ time where $n = |V|.$  
This is obviously true in the sink evacuation case because of the linearity of the functions as given by (\ref{eq:dcontdef}).

\subsection{Extending Theorem  \ref{theorem:FastBC}   to the continuous case}
\label{subsec:ContBd}

Recall that the Peaking Lemma (Lemma \ref{lemma:PeakingCriterionPutSink}) found $(u,v)$ such that  $f(V_{-v}(u),u) \le \mathcal{T}$ but  $f(V_{-v}(u)\cup\{v\},v) > \mathcal{T}$ and then placed a sink on $u.$  The motivating  intuition was that the peaking condition implies that $V_{-v}(u)$ MUST contain at least one sink. Placing that sink on the most extreme location possible for a single sink serving all of  $V_{-v}(u)$, i.e., $u$, could only improve the sink assignment. 

 In the continuous case, the analogous argument is again that  placing the  sink on the most extreme location possible for serving $V_{-v}(u)$ 
  can only improve the sink assignment.  But now,  the most extreme location possible is no longer required to be $u;$ it 
is the unique point   $s_{\mathcal{T}}$ defined in  (\ref{eq:st}). (Fig.~\ref{fig:ContHubTree})

\begin{figure}
	\centering
	\includegraphics[width=0.5\textwidth]{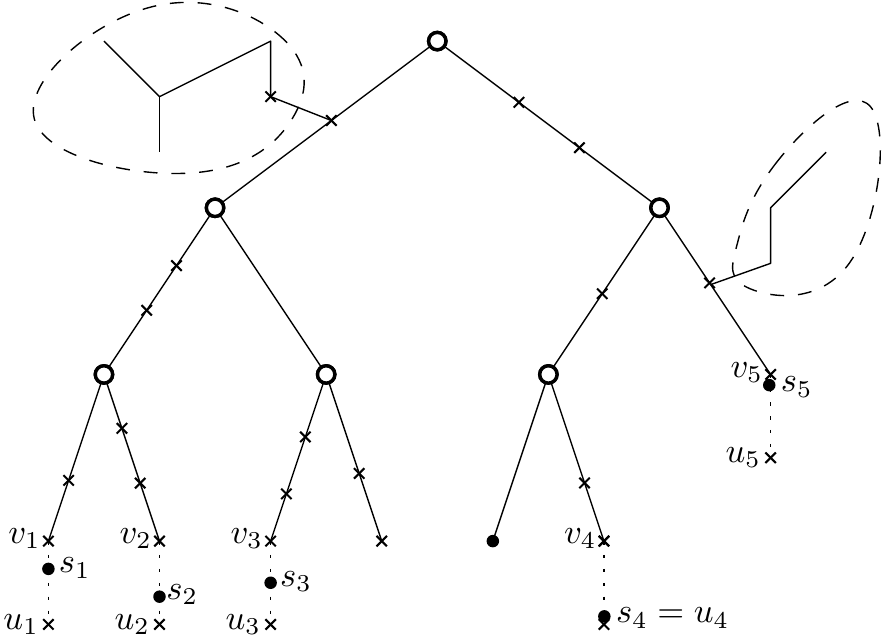}
	\caption{The hub tree after the sinks have been placed in the continuous problem.  Five sinks $s_1,\ldots,s_5$  have  been placed by the peaking lemma.   Note that the $s_i$ are not  necessarily in $V;$  they  can be located somewhere on the edge  $(u_i,v_i).$  In the feasibility version of the problem, the exact location of  $s_i$ on the edge is  known.  In the minimization version, only the fact  that $s_i$ falls in the edge $(u_i,v_i)$ is known but its exact location  might not be.}
		\label{fig:ContHubTree}
\end{figure}
The Peaking Lemma for the continuous case will now create a new node at $s=s_{\mathcal{T}}$, splitting $(u,v)$ into two pieces. It will then place a sink on $s$,  committing all of $V_{-v}(u)$ to $s$ and adding $s$ to $\Sout$. No changes need to be made to the Reaching Lemma which will remain correct as stated. It can then  be verified that the implementation of the peaking and reaching phases (including the first peaking phase via centroid decomposition) remain valid.  Thus, the  remainder of the  bounded-cost minmax $k$-sink  algorithm will  follow exactly as it did before, with the running time remaining the same as well.

\subsection{Extending Theorem  \ref{theorem:FastC}   to the continuous case}
Let ${\mathcal B}'$ be the new  bounded cost minmax $k$-sink  algorithm for the continuous case  described in the previous subsection and $B'(n)=B(n)$ be the cost of running the algorithm on an input of size $n.$   
We now     apply parametric search to  ${\mathcal B}'$
to create a  general algorithm  $\mathcal{I}'$ for the continuous case. Some subtle differences between this and the application of parametric search to the bounded algorithm $\mathcal{B}$ in Section \ref  {Section: Full Problem}
will  be needed.

Let
 ${\mathcal I}'$ be the interfered (parametric search) version of ${\mathcal B}'$ to be developed.  Similar to  $\mathcal I$,  ${\mathcal I}'$
 maintains a   \emph{threshold range}
 $(\mathcal{T}^L, \mathcal{T}^H]$. %  of allowable candidate values of $\mathcal{T}^*$. 
 $\mathcal I$ starts with  $(\mathcal{T}^L, \mathcal{T}^H] = (0,+\infty]$ and maintains the same   invariants:
\begin{enumerate}[label=(\textrm{I}\arabic{*})]
\item    $\mathcal{T}^L <\mathcal{T}^H$, 
\item $\mathcal{T}^L$ never decreases and $\mathcal{T}^H$ never increases.
\item $\mathcal{T}^L$ will be  infeasible and $\mathcal{T}^H$  will be   feasible. 
\item At each step of $\mathcal {I}'$,  the corresponding  state of $\mathcal B'$ would  be identical  for ALL
values of  $\mathcal{T} \in  [\mathcal{T}^L, \mathcal{T}^H)$.
\end{enumerate}

The major difference will be in the definition of  {\em state} and, in particular, what is stored in $\Sout$.
Recall that previously $\Sout=\{s_1,\ldots,s_t\}$ was the set of  known sinks (created by the peaking lemma).

As noted in Section \ref {subsec:ContBd},   sink $s=s_{\mathcal{T}}$  determined by the peaking lemma in the continuous case is no longer required to be  a  $v \in V$ but may lie inside an edge $(u,v)$.  ${\mathcal B}'$  explicitly determined the location of  $s_{\mathcal{T}}$  from $\mathcal T$ using (\ref{eq:st}). In  ${\mathcal I}'$, $\mathcal T$ is no longer exactly  known, so (\ref{eq:st}) can no longer be applied.
%
%What  is known though is that a sink $s$ exists on the  edge  $(u,v)$   and that $P_s = V_{-v}(u)$  (after the open commit). More explicitly we know $(u',v')$ such that $u \le u' \le v' \le v$ such that $s\in [u',v'].$  

To patch this, $\Sout$ will no longer store the (unknown)  {\em location} of sink $s$  but rather  the directed edge $(u(s),v(s))$  which is known to contain   $s$. 
(Fig.~\ref{fig:ContHubTree})
\begin{definition}
The {\em State} of algorithm $\mathcal{B'}$ at any given time will be $(\Sout,\Pout)$ and the $S(v)$ values in the hub tree.
 $s\in \Sout$ will be specified in the list by storing the edge  %that contains the sink, i.e., 
$s = (u(s),v(s))$ . % in the lists $S(v)$ for $v \in V.$  
\end{definition}

With this difference, Stage-Steps (Algorithm \ref {alg:Centroid Processing}) work exactly the same in ${\mathcal I}'$ as in ${\mathcal I}$.  That is,  after each Stage-Step the  edges containing sinks are stored in $\Sout$ and other edges are marked appropriately.  The proof of Lemma \ref{lem:Stage Thresh} for the continuous case will also follow.

The If-Steps now have to be further divided into two types depending upon the structure  of  the oracle call they make:
\begin{itemize}
\item {\bf Normal-Steps:}  These are the ones in (lines 5, 11, 17 of) Algorithm \ref{alg:PeakingAfter Reaching}.  The oracle evaluations are of the form $f(X,w)$ where $X$ is a subtree containing no sinks and  $w \in V \setminus \Sout.$  Normal-Steps will be processed using the same If-Step code from Figure \ref{fig:If Step} that was used by Algorithm $\mathcal{I}$,   except that $\Sout$ will store the edge $(u(s),v(s))$ known to contain $s$ rather than the unknown location  $s$.
\smallskip

\item {\bf Bulk-Steps:} These are the ones in (line 4  of)  Algorithm  \ref{alg:ReachingAfterPeaking}, the {\em Rsearch} procedure.  They are in the form  $f(BP(v_i,s'),s')$ where $s' \in S.$  These If-Steps will be processed using
the new description given in Figure \ref{fig:If  Rsearch Step}.  Note that this primarily  differs from the old If-Step of Figure \ref{fig:If Step} in Line 1.
\end{itemize}

\begin{figure}[t]
	\centering
	\includegraphics[width=0.8\textwidth]{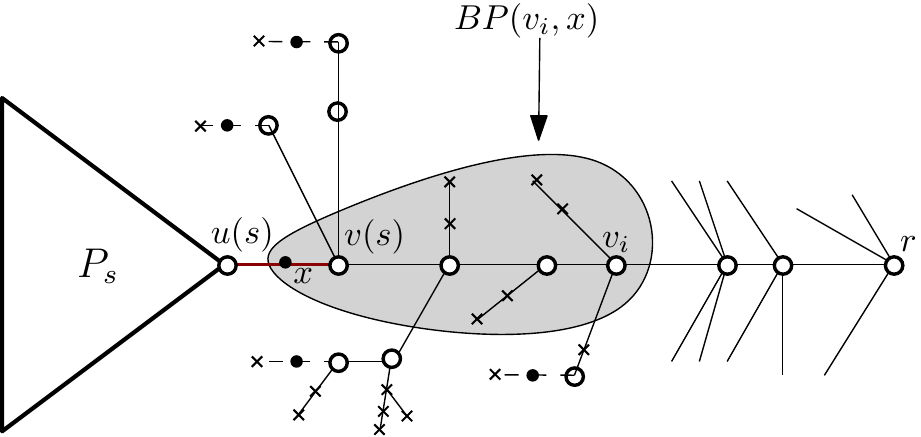}
	\caption{ An evaluation in  a Bulk If-Step.  A  previous peaking step determined that edge  $(u(s),v(s))$ contains a sink.
	For a given ancestor $v_i$ of $v(s)$ in $T$ the Bulk-If step finds the sink location  $x \in (u(s),v(s))$ that minimizes the maximum cost of servicing both $P_s =V_{-v(s)}(u(s))$ and   $BP(v_i,v(s)).$ In the diagram, the gray area is $BP(v_i,x)$; it is the path from $x$ to $v_i$ and all of the outstanding branches that fall off of it. The unfilled nodes are known hub nodes.  The filled nodes denote sink ``locations''.  The actual locations are unknown;  only the (dashed) edges which contain them are known.}
		\label{fig:ContHubTree2}
\end{figure}

\begin{figure}[t]
\par\noindent\underline{Bulk If-Step in Continuous Case to evaluate $f(BP(v_i,s),s)  \leq \mathcal{T}$ :}

\begin{enumerate}
\item Perform the evaluation (Fig.~\ref{fig:ContHubTree2})
$$a:=\min_{x\in (u(s),v(s))} \max\left( f(P_s \cup \{x\}, x), f(BP(v_i,x),x) )  \right).$$
\item  Resolve the If-Clause and reset $\mathcal{T}^L, \mathcal{T}^H$ if necessary as follows:
 	\begin{itemize}
 	\item (i) If $a \leq \mathcal{T}^L$,  resolve the associated If-clause as $f(BP(v_i,s),s)  \leq \mathcal{T}$.
 		\item(ii)  If $a \ge \mathcal{T}^H$, resolve the associated If-clause  as $f(BP(v_i,s),s) > \mathcal{T}$.
 		\item If $a \in (\mathcal{T}^L, \mathcal{T}^H)$, run a separate version of $\mathcal B$ from scratch with value $\mathcal{T} := a$, and observe the output.
 \begin{itemize}
  \item (iii) If Output is `No': set $\mathcal{T}^L := a$.\\  Resolve the associated If-clause as $f(BP(v_i,s),s) \leq \mathcal{T}$. %(because  $f(\cdot,\cdot) = a = \mathcal{T} ).$
  \item (iv)  If Output is `Yes', set $\mathcal{T}^H := a$.\\  Resolve the  associated  If-clause as $f(BP(v_i,s),s)  >\mathcal{T}$.
 \end{itemize}
 \end{itemize}
 \item Conclude the step by running the algorithm with the set value ${\mathcal T} = {\mathcal T}^L$ until the start of the next step.
\end{enumerate}
 \caption{ lf-Step for the Bulk-Step in  Continuous  Algorithm $\mathcal I$'.}
\label{fig:If  Rsearch Step}
\end{figure}
\medskip

Lemma \ref{lem:Stage Thresh} and its proof will still work for the new Stage-Steps.
We must now prove the equivalent of 
 Lemma \ref{lem:If Thresh} for these new If-Steps.

\begin{lemma}
\label{lem:If Thresh  Cont}
Let $\mathcal{T}^*$ be the optimal value of $\mathcal T$ and 
let $(\mathcal{T}^L, \mathcal{T}^H]$ be the threshold range after running $m$ If-Steps after the conclusion of the first peaking phase.
% Algorithm $\mathcal I$.  
Then

\begin{enumerate}
\item  $\mathcal{T}^* \in (\mathcal{T}^L, \mathcal{T}^H]$.
\item Let  $\mathcal{T}'  \in [\mathcal{T}^L, \mathcal{T}^H)$. 
% {\em \small (Note the half-openness of the intervals has been reversed.) }\\
Then,  Algorithm  $\mathcal B'$ run for $m$ steps with ${\mathcal T} = \mathcal{T}' $ would end in the same state as algorithm  $\mathcal B'$ run for $m$ steps with ${\mathcal T} = \mathcal{T}^L$.
\item Algorithm  $\mathcal I'$ run for $m$ steps  would end in the same state as algorithm  $\mathcal B'$ run for $m$ steps with ${\mathcal T} = \mathcal{T}^L$.
\end{enumerate}
\end{lemma}
\begin{proof}
(1) The proof is exactly the same as in  Lemma \ref{lem:If Thresh}.

 (2)  Again the proof is by induction on $m$,  that after $m$ steps   algorithm  $\mathcal B'$ will be  in the same state  when run on $\mathcal{T}'$ and $\mathcal{T}'$.
 Assume this is true for $m-1$ and now   consider what happens in the $m$'th If-step. 
\medskip

The  analysis of a Normal Step is exactly the same as it was in Lemma \ref{lem:If Thresh} so we do not repeat it except  to note again that 
in all cases,  the algorithm correctly processes  ``$a \le \mathcal{T}^L?$''. Furthermore  ``$a \le \mathcal{T}^L?$'' resolves identically to the query
	``$a \le \mathcal{T}'$'' so  $\mathcal B'$ can not distinguish between the two cases. 

\medskip
The analysis of the Bulk-Step is more interesting. For simplicity set %  \mc{Diagram}
$s_L := s_{\mathcal{T}_L}$ and  $s' := s_{\mathcal{T}'}$
% $s_H := s_{\mathcal{T}_H}$ 
as defined in (\ref{eq:st}).   By the  induction hypothesis both of these values are on the edge $(s(u),v(u)).$  
%Note that it is technically possible that $s_H = v.$ 
%Also set $s_a$ as defined in  (\ref{eq:st}).

\begin{figure}[t]
	\centering
	\includegraphics[width=0.8\textwidth]{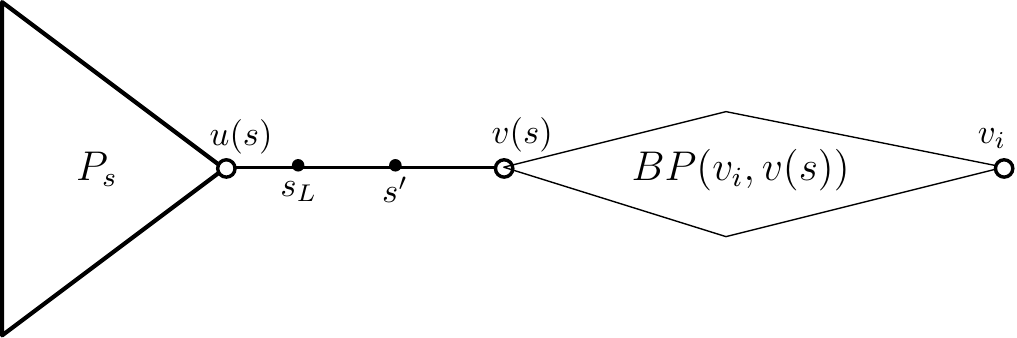}
	\caption{ Illustration of   $s_L \le s'$
	 in the proof of Lemma \ref{lem:If Thresh  Cont}. Recall that $s_L$ is the rightmost location of a sink on edge $(u(s),v(s))$ that supports $P_s$ when 
	$\mathcal{T} =  mathcal{T}^L$  and $s'$ is the rightmost that supports it when $\mathcal{T} =  \mathcal{T}'$.}
		\label{fig:ContBulkLemma}
\end{figure}

Since 
$\mathcal{T}^L <  \mathcal{T}'$, monotonicity implies  $s_L \le  s'$  (Fig.~\ref{fig:ContBulkLemma}).
%Furthermore, by definition and path monotonicity
%$$f(P_s \cup\{s_L\},s_L) \le \mathcal{T}^L
%\quad\mbox{and}\quad f(P_s \cup\{s'\},s') \le \mathcal{T}'$$ 
%$$f(P_s \cup\{s_L\},s_L) \le   f(P_s \cup\{s'\},s').$$ 
%and
Note  that 
$BP(v_i,s_L) = BP(v_i,v(s)) \cup \{s_L\}$ and   $BP(v_i,s') = BP(v_i,v(s)) \cup \{s'\}.$ Thus, from monotonicity,
\begin{equation}
\label{eq:BP Mon}
  f(BP(v_i,s_L),s_L)  \ge  f(BP(v_i,s'),s').
  \end{equation}
%The last equation follows from the fact that 

Let $x_a  \in (u(s),v(s))$ be a value    such that 
$$a=\max \left( f(P_s \cup \{x_a\}, x_a), f(BP(v_i,x_a),x_a )  \right).$$
We now analyze the cases separately:
\begin{enumerate}
\item[(i)] $f(P_s \cup \{x_a\}, x_a) \le a \le \mathcal{T}^L$ implies $x_a \le s_a \le s_L.$ From monotonicity and (\ref{eq:BP Mon})
$$
   \mathcal{T}' > \mathcal{T}^L \ge a \ge  f(BP(v_i,x_a),x_a ) 
				 \ge  f(BP(v_i,s_L),s_L)
				  \ge  f(BP(v_i,s'),s').$$
Then
\begin{equation}
\label{eq:cont case i}
 f(BP(v_i,s_L),s_L) \le \mathcal{T}^L  \quad \mbox{and} \quad   f(BP(v_i,s'),s') \le \mathcal{T}'.
\end{equation}
				  
\smallskip

\item[(ii)]  By the definition of $a,$ 
$$\max\Bigl( f(P_s \cup \{s'\}, s'), f(BP(v_i,s'),s') )  \Bigr) \ge   a  >\mathcal{T}^H.$$
Since 
%$f(P_s \cup\{s'\},s') \le  f(P_s \cup\{s_H\},s_H) \le \mathcal{T}^H$ this  and (\ref{eq:BP Mon}) imply
$f(P_s \cup\{s'\},s') \le  \mathcal{T}' \le \mathcal{T}^H$, this  and (\ref{eq:BP Mon}) imply
$$  f(BP(v_i,s_L),s_L)  \ge f(BP(v_i,s'),s') ) > \mathcal{T}^H > \mathcal{T}' > \mathcal{T}^L.$$
\end{enumerate}
Then
\begin{equation}
\label{eq:cont case ii}
 f(BP(v_i,s_L),s_L) > \mathcal{T}^L  \quad \mbox{and} \quad   f(BP(v_i,s'),s') > \mathcal{T}'.
\end{equation}

\smallskip
\begin{itemize}
\item [(iii)]  We set $\mathcal{T}^L=a$ so $ a < \mathcal{T}' < \mathcal{T}^H$. Furthermore,  $x_a \le s_a = s_L$ so from (\ref{eq:BP Mon}) and path monotonicity
$$\mathcal{T}' > \mathcal{T}'^L =a  \ge  f(BP(v_i, x_a),x_a) \ge f(BP(v_i, s_L),s_L) \ge f (BP(v_i, s'),s').$$
Then
\begin{equation}
\label{eq:cont case iii}
 f(BP(v_i,s_L),s_L) \le \mathcal{T}^L  \quad \mbox{and} \quad   f(BP(v_i,s'),s') \le \mathcal{T}'.
\end{equation}

\smallskip
\item[(iv)] We set $\mathcal{T}^H= a$.  Then  $ \mathcal{T}^L< \mathcal{T}' < a$. By definition
$f(P_s \cup \{s'\}, s') \le \mathcal{T}' < a.$  
As in the analysis in case (ii) we note that from   the definition of $a,$ 
$$\max\left( f(P_s \cup \{s'\}, s'), f(BP(v_i,s'),s') )  \right) \ge   a $$
and thus
$$f(BP(v_i,s'),s') \ge   a $$
Using  (\ref{eq:BP Mon}) again shows 
$$ f(BP(v_i, s_L),s_L) \ge f (BP(v_i, s'),s') \ge   a  >\mathcal{T}'  > \mathcal{T}^L. $$
Then
\begin{equation}
\label{eq:cont case iv}
 f(BP(v_i,s_L),s_L) > \mathcal{T}^L  \quad \mbox{and} \quad   f(BP(v_i,s'),s') > \mathcal{T}'.
\end{equation}
\medskip
\end{itemize}

Equations (\ref{eq:cont case i})-(\ref{eq:cont case iv}) show that  the queries
``$f(BP(v_i,s_L),s_L) \le \mathcal{T}^L?$''  and ``$f(BP(v_i,s'),s') \le \mathcal{T}'?$'' will always return the same answer.

Since the algorithm started the If-Step in an identical state for
	cases ${\mathcal T} = {\mathcal T}', {\mathcal T} ^L$, and it can not distinguish between those cases during the If-Step, it ends in the same state for both of them.

\medskip

 (3) follows from the analysis of (2).		\qed
\end{proof}

We can now prove
\begin{lemma}
\label{lem:inter2}

The interfered algorithm $\mathcal{I'}$  in the continuous case will terminate in  $O\bigl( \max(k,\log n) k \log^2 n t_{\mathcal{A}}(n)\bigr)$ time. 

 Let $(\mathcal{T}_{<},\mathcal{T}_{>}]$ be the threshold range  when  $\mathcal I$  terminates and 
$\mathcal{T^*}$ be the optimal value of $\mathcal T$.
  Then $\mathcal{T^*}=\mathcal{T}_{>}$. In particular, we can then run the bounded cost Algorithm $\mathcal B$   on $\mathcal{T} := \mathcal{T}_{>}$ to retrieve the optimal feasible configuration.
 \label{lemma:UpperThresholdIsAnwser}
\end{lemma}

The proof of this lemma is almost exactly the same as that of  Lemma \ref{lem:inter2} and will therefore be omitted.  The only difference is that $\mathcal{I}'$ needs to do a bit of extra work for the Stage Steps and If-Steps.  But this is only $O(t_{\mathcal{A}}(n))$ work per sink and there are at most $k$ sinks.  The extra work is therefore $O(kt_{\mathcal{A}}(n))$ which is subsumed by the remaining running time of the algorithm which is the same as that for the discrete case.

\section{The fixed sink problem (optimal partitioning)}
\label{Sec:Fixed Sink}

This section sketches a proof of  Theorem  \ref{theorem:FastCF}, i.e., the special case in which the locations of $S,$ the set of  $k$ sinks in the input tree $\Tin$,  are provided as part of the input.  The problem is thus to partition $\Tin$ into $k$ subtrees, each subtree containing exactly one $s \in S$, so as to minimize the max-cost of the subtrees.  
%For the evacuation problem with general capacities,  \cite{Mamada2005a} gave a  $O(n (c \log n)^{k+1})$ time algorithm where $c$ is some constant.  For ``large'' $k$ \cite{Mamada2005} later reduced this down  to $O(n^2 k \log^2  n)$.

Because the sinks are given they can be considered as nodes in  $V$ and thus this problem  is always  discrete. 
Also, as stated  in Theorem  \ref{theorem:FastCF}, the underlying function is now only required to be {\em relaxed} minmax monotone and not strictly  minmax monotone.  We explain below why this relaxation occurs.

\subsection{If $S$ are all leaves of $\Tin.$}

Consider the special case in which all nodes in $S$ are {\em leaves} of the input tree $\Tin$. As before, we start by constructing a feaasiblity test;  given $\mathcal T > 0$, decide whether 
there exists a partition with bounded cost $\le \mathcal T$.

As a first step,  the new algorithm will build 
 the hub tree $T_{H(S)}$  (Def.~\ref{def:hub tree}) off of the given sinks. Let $T_i$,  $i=1,2,\ldots,t$, be the corresponding outstanding branches  and $w_i \in T_{H(S)}$ the node off of which  $T_i$ falls.   Set
$$\Tmin = \min_i  f\left(T_i \cup\{w_i\}, w_i\right).$$
By asymptotic subadditivity,  calculating $\Tmin$  requires only $O(t_{\mathcal{A}}(n))$ time.

If $\Tmin > \mathcal T$, no feasible solution exists.  Otherwise,  $T_{H(S)}$ is RC-viable and, from the perspective of  algorithm
$\mathcal{B}$, the state of the problem is exactly the same as if the first peaking phase had just concluded with $T= \Tin,$ $\Sout=S$ and $\Pout$ being defined by setting  $P_s = \{s\}$ for all $s \in \Sout.$ 
Referring to Section \ref{subsec:comb bounded}, this is as if Step 1 of the algorithm  had just concluded and the algorithm is now starting Step 2 (Creating the First Hub Tree).  Continuing to run $\mathcal{B}$ from this point will now provide the correct answer.  Since  new sinks will never be added, the algorithm will never need to enter a peaking-phase.  Instead, if after a reaching phase the working tree stops being RC-viable, the algorithm will declare that  no solution exists.

 The path-monotonicity requirement of minmax monotone functions stated in  Section \ref{subsec:minmaxdef}  was only used in the derivation of the peaking condition; the reaching condition only required set-monotonicity. Since a peaking phase is never entered in this fixed-sink case, the algorithm remains correct even for
{\em  relaxed}  minmax monotone  functions.

The running time of this algorithm will be $O(nk + t_{\mathcal{A}}(n))$  (to calculate $\Tmin$ and the structure of the  first hub tree) plus the cost of running $\mathcal B$ when skipping  the first peaking phase, which is again $O(k t_\mathcal{A}(n) \log n)$.  Technically, the algorithm could now be simplified  by noting that no peaking phases are ever run, but this would not improve the worst case running time.

Now consider solving the general minmax $k$-sink problem by applying parametric search as in the creation of $\mathcal{I}.$

Again start the algorithm by building the hub tree and calculating $\Tmin.$  Next, run the bounded cost fixed-sink algorithm with  $\mathcal{T} = \Tmin.$ As previously noted,  
$\Tmin \le \mathcal{T}^*.$  So,  if  $\Tmin$ is feasible $\mathcal{T}^*=\Tmin$ and the algorithm concludes.
The total work performed so far is $O(kn+ t_{\mathcal{A}}(n))$  plus $O(k t_\mathcal{A}(n) \log n)$ for calling the bounded algorithm.

If $\Tmin$ is  not feasible, 
 set $(\mathcal{T}^L, \mathcal{T}^H] = (\Tmin,\infty]$.  The algorithm is now in the same state as $\mathcal{I}$ would have been in  if $\mathcal{I}$ had just completed  all of the Stage-Steps of the first peaking phase.  Continuing to  run $\mathcal{I}$ from this point onward will yield the final answer.  This cost of running $\mathcal I$, omitting the Stage -Steps,  is  $O(k^2 t_\mathcal{A}(n) \log^2 n)$; this dominates the Stage-Step, and is thus  the overall time complexity.

\subsection{If $S$ are not restricted to be  leaves of $\Tin.$}

\begin{figure}
	\centering
	\includegraphics[width=0.7\textwidth]{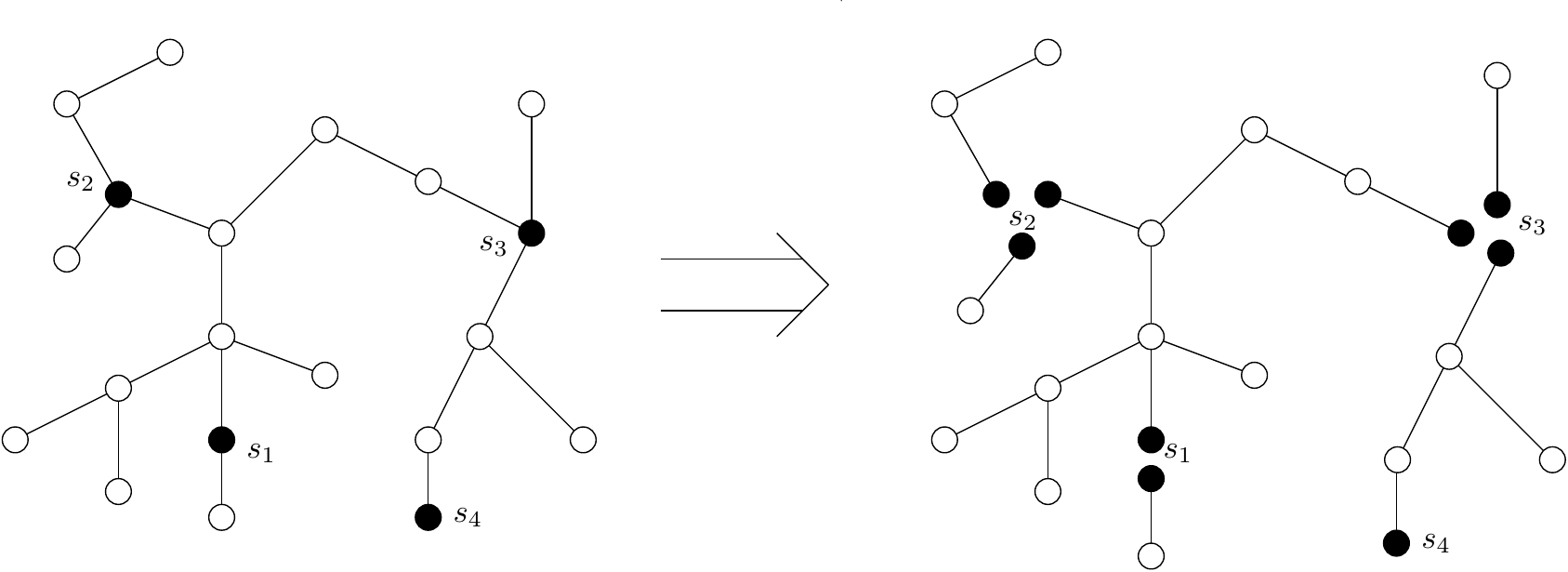}
	\caption{The original tree with specified sinks $s_1,s_2,s_3,s_4$ is on the left.  The transformed forest is on the right.  Note that a new sink has been created for each edge adjacent to a sink in the original tree.}
		\label{fig:Leaf Partition tree}
\end{figure}

The subsection above  solved the minmax $k$ fixed-sink problem in  $I'(k,n)=O(  k^2 n t_{\mathcal{A}}(n) \log^2) $ time  when $S$ is restricted to being leaves of $\Tin.$
 %First note that since $t_{\mathcal{A}}(n)$ is subadditive,  $I'(k,n)$ is also subbadditive for any fixed $k$. 
  Without loss of generality, we assume that $I'(k,n)$ is non-decreasing in $k.$  

 The solution for general position $S$ uses  a standard transformation of $\Tin$ into a forest.

For every given sink $s \in S$, let $\mathcal{N}(s)$ denote  its set of neighbors and for every $u \in \mathcal{N}(s)$ create a new sink node $s_u$ and edge $(u, s_u)$.     Delete the original nodes in $S.$  (Fig~ \ref{fig:Leaf Partition tree}). What remains is a forest $T_1,T_2,\ldots,T_r$ of trees in which each $T_i$ contains  sinks $S_i$ at its leaves, where $|S_i| =  k_i \le k$ sinks.  At most one new node is created for every edge in the original tree  so the total number of vertices in the forest is $<2n.$  Furthermore, every partition  on the forest corresponds in the natural way with a partition in the original tree such that the minmax-cost partition of the forest corresponds to a minmax-cost partition of the tree that has  the same cost.

 It is not difficult to see that 
 $$\min_{\mathcal{P} \in  \Lambda[S]} f(\mathcal{P},S)  = \min_{1 \le i \le r} F_i
\quad\mbox{where} \quad
F_i = \min_{\mathcal{P} \in  \Lambda[S_i]} f(\mathcal{P},S_i).
$$
That is, we can separately find the optimal partition for each subtree and knit them together to construct an optimal partition for the original tree.

Thus, to solve the problem on the original tree it suffices to solve it on each of the trees $T_i$ individually.  Let $n_i$ be the number of nodes in tree $T_i.$
Recall that the statement of Theorem \ref{theorem:FastCF} assumed asymptotic subadditivity  {\em and} that $t_\mathcal{A}(2n) = O(t_\mathcal{A}(n))$.
Thus,  the total cost is also  at most 
$$\sum_{i=1}^r I'(k_i, n_i) = O(I'(k,2n-1)) = O(I'(k,n))= O(  k^2 n t_{\mathcal{A}}(n) \log^2)$$
and we are done.

%The argument above was for the minmax $k$ fixed-sink problem.  The same composition technique also permits generalizing the solution of  the bounded cost  minmax $k$-sink problem  to the general problem  from the
%problem with $S$ restricted to be on leaves.

 Note that for the sink evacuation problem, plugging in the $O(n \log^2 n)$ oracle used previously, this leads to a $O(n k^2 \log^4 n)$ time algorithm for the partitioning problem, substantially improving upon the the 
 $O(n (c \log n)^{k+1})$  \cite{Mamada2005a} and   $O(n^2 k \log^2  n)$ \cite{Mamada2005} algorithms when $4 < k \ll n$.

\section{Conclusion}
Given a Dynamic Flow network on a tree $G=(V,E)$  we derived an algorithm  for finding the locations of $k$ sinks that  minimize the maximum time needed to evacuate the entire tree.

The algorithm was developed in two parts. Sections \ref{Section: Bounded Cost} and \ref {Sec: Bounded Algorithm}  developed a {\em feasibility test}, i.e., for $\mathcal{T}> 0,$ an algorithm for finding a  placement of $k$ sinks that permits evacuating the tree in 
 $\le \mathcal{T}$ time  (or determining that such a placement does not exist).   Section \ref{Section: Full Problem} showed how to apply parametric search to modify  this test to find  the minimum  feasible $\mathcal{T}^*$.  Section \ref{Sec:Continuous} extended the algorithms to work in the continuous case (in which sinks can be placed on edges).  Finally,   Section \ref{Sec:Fixed Sink} developed better algorithms for the case in which the $k$ sinks are known in advance.
 
 The sink-evacuation problem is  a special case of the minmax Centered $k$-partitioning problem on trees.
 All the results described could partition using   any minmax monotone  function for servicing trees from centers.
 Assuming an  $t_\mathcal{A}(n)$ time oracle for calculating the cost of the fixed $1$-sink problem on trees, our main algorithm works in $O( \max(k,\log n) k t_{\mathcal{A}}(n) \log^2 n )$ time, improved to $O(  k^2  t_{\mathcal{A}}(n)  \log^2 n)$ time if the sinks are known in advance.

 These were the first known polynomial time algorithms for these sink location problems.  The obvious direction for improvement would be to try to develop algorithms whose running times, like the $O(n)$ one for unweighted $k$-center \cite{frederickson1991parametric}  and the  $O(n \log^2 n)$ ones \cite{megiddo1981n,Cole87} for the weighted $k$-center problem, are only dependent upon $n$ and not $k.$ As noted earlier, the bottleneck to this generalization seems to be that, unlike in those previous tree-partitioning problems,  the cost  oracle  $f(U,s)$ here  is permitted to be  a complicated non-linear function of the  topology and all the vertex weights of the full  tree $U$, and can not be decomposed into simpler parts.

\begin{acknowledgements}
The work of both authors was partially supported by Hong Kong RGC CERG Grant 16208415
\end{acknowledgements}
\bibliographystyle{plainurl} 
%\bibliography{Evac,Evacuation_Problems,Extra_Bib}
\bibliography{Evac,Evacuation_Problems}

\begin{thebibliography}{10}

\bibitem{Agasi1993}
Eliezer Agasi, Ronald~I. Becker, and Yehoshua Perl.
\newblock {A shifting algorithm for constrained min-max partition on trees}.
\newblock {\em Discrete Applied Mathematics}, 45(1):1--28, 1993.

\bibitem{Aronson1989}
J.~E. Aronson.
\newblock {A survey of dynamic network flows}.
\newblock {\em Annals of Operations Research}, 20(1):1--66, 1989.

\bibitem{becker1983shifting}
Ronald~I. Becker and Yehoshua Perl.
\newblock Shifting algorithms for tree partitioning with general weighting
  functions.
\newblock {\em Journal of algorithms}, 4(2):101--120, 1983.

\bibitem{Becker1995}
Ronald~I. Becker and Yehoshua Perl.
\newblock {The shifting algorithm technique for the partitioning of trees}.
\newblock {\em Discrete Applied Mathematics}, 62(1-3):15--34, 1995.

\bibitem{Becker1980}
Ronald~I. Becker, Yehoshua Perl, and Stephen~R. Schach.
\newblock {A shifting algorithm for min-max tree partitioning}.
\newblock {\em Journal of the ACM (JACM)}, 29(1):58--67, 1982.

\bibitem{bhattacharya2017improved}
Binay Bhattacharya, Mordecai~J Golin, Yuya Higashikawa, Tsunehiko Kameda, and
  Naoki Katoh.
\newblock Improved algorithms for computing k-sink on dynamic flow path
  networks.
\newblock In {\em Proceedings of WADS'17}, pages 133--144. Springer, 2017.

\bibitem{bhattacharya2015improved}
Binay Bhattacharya and Tsunehiko Kameda.
\newblock Improved algorithms for computing minmax regret sinks on dynamic path
  and tree networks.
\newblock {\em Theoretical Computer Science}, 607:411--425, 2015.

\bibitem{chen2015efficient}
Danny~Z. Chen, Jian Li, and Haitao Wang.
\newblock Efficient algorithms for the one-dimensional k-center problem.
\newblock {\em Theoretical Computer Science}, 592:135--142, 2015.

\bibitem{Chen2007}
Jiangzhuo Chen, Robert~D Kleinberg, L{\'a}szl{\'o} Lov{\'a}sz, Rajmohan
  Rajaraman, Ravi Sundaram, and Adrian Vetta.
\newblock {(Almost) Tight bounds and existence theorems for single-commodity
  confluent flows}.
\newblock {\em Journal of the ACM}, 54(4), jul 2007.

\bibitem{Chen2006}
Jiangzhuo Chen, Rajmohan Rajaraman, and Ravi Sundaram.
\newblock {Meet and merge: Approximation algorithms for confluent flows}.
\newblock {\em Journal of Computer and System Sciences}, 72(3):468--489, 2006.

\bibitem{Cole87}
Richard Cole.
\newblock Slowing down sorting networks to obtain faster sorting algorithms.
\newblock {\em Journal of the ACM (JACM)}, 4(1):200--208, 1978.

\bibitem{Dressler2010b}
Daniel Dressler and Martin Strehler.
\newblock {Capacitated Confluent Flows: Complexity and Algorithms}.
\newblock In {\em 7th International Conference on Algorithms and Complexity
  (CIAC'10)}, pages 347--358, 2010.

\bibitem{Fleischer2007}
Lisa Fleischer and Martin Skutella.
\newblock {Quickest Flows Over Time}.
\newblock {\em SIAM Journal on Computing}, 36(6):1600--1630, January 2007.

\bibitem{fleischer1998efficient}
Lisa Fleischer and {\'E}va Tardos.
\newblock Efficient continuous-time dynamic network flow algorithms.
\newblock {\em Operations Research Letters}, 23(3):71--80, 1998.

\bibitem{Ford1958a}
L.~R. Ford and D.~R. Fulkerson.
\newblock {Constructing Maximal Dynamic Flows from Static Flows}.
\newblock {\em Operations Research}, 6(3):419--433, June 1958.

\bibitem{frederickson1991parametric}
Greg~N Frederickson.
\newblock Parametric search and locating supply centers in trees.
\newblock In {\em Proceedings of the Second Workshop on Algorithms and Data
  Structures (WADS'91)}, pages 299--319. Springer, 1991.

\bibitem{garey1979computers}
Michael~R Garey and David~S Johnson.
\newblock {\em Computers and intractability: A Guide to the Theory of
  NP-Completeness}.
\newblock W.H. Freeman and Company, 1979.

\bibitem{Golin2017sink}
Mordecai Golin, Hadi Khodabande, and Bo~Qin.
\newblock Non-approximability and polylogarithmic approximations of the
  single-sink unsplittable and confluent dynamic flow problems.
\newblock In {\em Proceedings of the 27th International Symposium on Algorithms
  and Computation (ISAAC'16)}, 2017.

\bibitem{Higashikawa2014}
Y.~Higashikawa, M.~J. Golin, and N.~Katoh.
\newblock {Minimax Regret Sink Location Problem in Dynamic Tree Networks with
  Uniform Capacity}.
\newblock In {\em Proc of the 8'th Intl Workshop on Algorithms and Computation
  (WALCOM'2014)}, pages 125--137, 2014.

\bibitem{Higashikawa2014c}
Yuya Higashikawa.
\newblock {\em {Studies on the Space Exploration and the Sink Location under
  Incomplete Information towards Applications to Evacuation Planning}}.
\newblock PhD thesis, Kyoto University, 2014.

\bibitem{Hoppe2000b}
B~Hoppe and \'{E} Tardos.
\newblock {The quickest transshipment problem}.
\newblock {\em Mathematics of Operations Research}, 25(1):36--62, 2000.

\bibitem{Kamiyama}
Naoyuki Kamiyama, Naoki Katoh, and Atsushi Takizawa.
\newblock {Theoretical and Practical Issues of Evacuation Planning in Urban
  Areas}.
\newblock In {\em The Eighth Hellenic European Research on Computer Mathematics
  and its Applications Conference (HERCMA2007)}, pages 49--50, 2007.

\bibitem{kariv1979algorithmic}
Oded Kariv and S~Louis Hakimi.
\newblock An algorithmic approach to network location problems. i: The
  p-centers.
\newblock {\em SIAM Journal on Applied Mathematics}, 37(3):513--538, 1979.

\bibitem{Lari2016}
Isabella Lari, Justo Puerto, Federica Ricca, and Andrea Scozzari.
\newblock {Algorithms for uniform centered partitions of trees}.
\newblock {\em Electronic Notes in Discrete Mathematics}, 55:37--40, 2016.

\bibitem{Lari2015}
Isabella Lari, Federica Ricca, Justo Puerto, and Andrea Scozzari.
\newblock {Partitioning a graph into connected components with fixed centers
  and optimizing cost-based objective functions or equipartition criteria}.
\newblock {\em Networks}, 67(1):69--81, 2015.

\bibitem{Mamada2005a}
Satoko Mamada and Kazuhisa Makino.
\newblock {An Evacuation Problem in Tree Dynamic Networks with Multiple Exits}.
\newblock In {Tatsuo Arai}, Shigeru Yamamoto, and Kazuhi Makino, editors, {\em
  Systems \& Human Science-For Safety, Security, and Dependability; Selected
  Papers of the 1st International Symposium SSR2003}, pages 517--526. Elsevier
  B.V, 2005.

\bibitem{Mamada2005}
Satoko Mamada, Takeaki Uno, Kazuhisa Makino, and Satoru Fujishige.
\newblock {A tree partitioning problem arising from an evacuation problem in
  tree dynamic networks}.
\newblock {\em Journal of the Operations Research Society of Japan},
  48(3):196--206, 2005.

\bibitem{Mamada2006}
Satoko Mamada, Takeaki Uno, Kazuhisa Makino, and Satoru Fujishige.
\newblock {An $O(n \log^2 n) $algorithm for the optimal sink location problem
  in dynamic tree networks}.
\newblock {\em Discrete Applied Mathematics}, 154(2387-2401):251--264, 2006.

\bibitem{megiddo1979combinatorial}
Nimrod Megiddo.
\newblock Combinatorial optimization with rational objective functions.
\newblock {\em Mathematics of Operations Research}, 4(4):414--424, 1979.

\bibitem{megiddo1983new}
Nimrod Megiddo and Arie Tamir.
\newblock New results on the complexity of p-centre problems.
\newblock {\em SIAM Journal on Computing}, 12(4):751--758, 1983.

\bibitem{megiddo1981n}
Nimrod Megiddo, Arie Tamir, Eitan Zemel, and Ramaswamy Chandrasekaran.
\newblock An {$O(n \log^2n)$} algorithm for the k'th longest path in a tree
  with applications to location problems.
\newblock {\em SIAM Journal on Computing}, 10(2):328--337, 1981.

\bibitem{Pascoal2006}
Marta M.~B. Pascoal, M.~Eug\'{e}nia~V. Captivo, and Jo\~{a}o C.~N.
  Cl\'{\i}maco.
\newblock {A comprehensive survey on the quickest path problem}.
\newblock {\em Annals of Operations Research}, 147(1):5--21, August 2006.

\bibitem{Perl1985}
Yehoshua Perl and Uzi Vishkin.
\newblock {Efficient implementation of a shifting algorithm}.
\newblock {\em Discrete Applied Mathematics}, 12(1):71--80, 1985.

\bibitem{Shepherd2015}
F.~Bruce Shepherd and Adrian Vetta.
\newblock {The Inapproximability of Maximum Single-Sink Unsplittable, Priority
  and Confluent Flow Problems}.
\newblock {\em arXiv:1504.0627}, 2015.
\newblock URL: \url{http://arxiv.org/abs/1504.0627}, \href
  {http://arxiv.org/abs/1504.0627} {\path{arXiv:1504.0627}}.

\bibitem{Skutella2009}
Martin Skutella.
\newblock {An introduction to network flows over time}.
\newblock In William Cook, L\'{a}szl\'{o} Lov\'{a}sz, and Jens Vygen, editors,
  {\em Research Trends in Combinatorial Optimization}, pages 451--482.
  Springer, 2009.

\end{thebibliography}

%\input{Old_Code}
%\newpage
%\appendix

%\input{appendix}

\end{document}